\documentclass[%
 aip,
 amsmath,amssymb,
 reprint,%
]{revtex4-1}

\usepackage{graphicx}% Include figure files
\usepackage{dcolumn}% Align table columns on decimal point
\usepackage{bm}% bold math
\usepackage{amssymb}
\usepackage{amsmath}
\usepackage{amsthm}
\usepackage[utf8]{inputenc}
\usepackage[T1]{fontenc}
\usepackage{mathptmx}
\usepackage{etoolbox}
\usepackage{cleveref}
\renewcommand{\thefootnote}{\fnsymbol{footnote}}

\theoremstyle{plain}
\newtheorem{definition}{Definition}[section]
\newtheorem{lemma}[definition]{Lemma}

\newtheorem{proposition}[definition]{Proposition}
\newtheorem{corollary}[definition]{Corollary}
\theoremstyle{definition}
\newtheorem{remark}[definition]{Remark}
\newtheorem{example}[definition]{Example}

\newcommand{\beq}{\begin{equation}}
\newcommand{\eeq}{\end{equation}}

\newcommand{\ep}{\varepsilon}

\newcommand{\mb}[1]{\mathbf{#1}}

\newcommand{\gb}[1]{\boldsymbol{#1}}
\newcommand{\mc}[1]{\mathcal{#1}}

\DeclareMathAlphabet\mathbit
    \encodingdefault\rmdefault\bfdefault\itdefault
\DeclareOldFontCommand{\bi}{\normalfont\bfseries\itshape}{\mathbit}

\newcommand{\be}{\begin{equation}}
\newcommand{\ee}{\end{equation}}

\def\fakebold#1{\relax\ifvmode\leavevmode\fi%
\ifmmode%
\setbox0=\hbox{$#1$}%
\else%
\setbox0=\hbox{#1}%
\fi%
\kern-.02em\copy0 \kern-\wd0%
\kern .04em\copy0 \kern-\wd0%
\kern-.0125em\raise.02em\box0%
}%

\renewcommand{\geq}{\geqslant}

\newcommand{\hchi}{\hat{\mb{e}}_\chi}

\makeatletter
\def\@email#1#2{%
 \endgroup
 \patchcmd{\titleblock@produce}
  {\frontmatter@RRAPformat}
  {\frontmatter@RRAPformat{\produce@RRAP{*#1\href{mailto:#2}{#2}}}\frontmatter@RRAPformat}
  {}{}
}%
\makeatother
\begin{document}

\preprint{AIP/123-QED}

\title[How stagnation lines move, persist and break up]{How stagnation lines move, persist and break up in fluid flows}
% Force line breaks with \\
\author{Prabakaran Rajamanickam}
  \email{prabakaran.rajamanickam@strath.ac.uk}
\affiliation{Department of Mathematics \& Statistics, University of Strathclyde, Glasgow, UK}%

\begin{abstract}
Stagnation curves of arbitrary shape, on which the velocity of an unsteady incompressible flow $\mb{v}(\mb x,t)$ vanishes, are non-generic yet common to stagnation-point flows; we derive their motion. Differentiating $\mb v=\mb 0$ in time gives $\mb L\mb U=-\mb a$ ($\mb U$ curve velocity, $\mb a=\partial_t\mb v$ acceleration, $\mb L=\nabla\mb v$ velocity gradient, annihilating the tangent and acting on the normal plane via a $2\times2$ block $\mb K$), so dynamics enters only via $\mb a$. The normal velocity is $-\mb K^{-1}$ times the normal part of $\mb a$; tangential solvability is one scalar condition $\Sigma=0$ per point ($\Sigma$: part of $\mb a$ outside the range of $\mb L$). In potential flow, $\mb K$ is the strain rate, the normal velocity is pressure gradient over principal strain rate $\ep$, and persistence needs uniform pressure along it. In Stokes flow, tangential vorticity enters $\mb K$ (lines can be of centre type) and $\Sigma$ involves normal vorticity; in Navier--Stokes flow, $\Sigma$ is set by the instantaneous flow through nonlocal pressure and viscosity. If $\Sigma\neq0$, the line is replaced to first order in time by isolated points at zeros of $\Sigma$ (critical points of $\varphi_t$ in potential flow; possibly none in Stokes flow, so loops can disappear). Symmetry (line an orbit of a rigid-motion group, no velocity along it) suffices for persistence but is unnecessary: flows steady in a rigidly moving frame have $\Sigma\equiv0$. Incompressibility enters the kinematics only through $\mathrm{tr}\,\mb K=0$; for inviscid, irrotational, homentropic compressible flow, the dilatation rate on a line is $\varphi_{tt}/c^2$ ($c$ the sound speed), so lines can also be of node type, persistence requires a uniform enthalpy along the line, the singular normal block arises at ordinary points of a line (compressible case), and acoustic standing waves give exact examples, including a node line that breaks into alternating sinks and saddles. Also given: the degenerate case of a singular normal block (pairs of points born at distance $O(t^{1/2})$), curvature and torsion evolution, closure for flows driven by moving singularities, helical potential stagnation-point flow, and an exact, numerically verified Navier--Stokes example where saddle and centre lines split into four points per period.
\end{abstract}

\maketitle

\section{Introduction}\label{sec:intro}

A stagnation point of a flow is a point at which the velocity vanishes in the chosen frame. In its neighbourhood the velocity is, to leading order, the linear field $\mb x\mapsto\mb L\,\mb x$ generated by the velocity-gradient tensor $\mb L=\nabla\mb v$ evaluated at the point, so the type of $\mb L$ (saddle, node, spiral) decides how fluid approaches, turns and leaves the point. This is why stagnation points control the transport of heat, solute and vorticity in many flows. 

We begin with a count of dimensions, because it explains why a \emph{stagnation line} needs a separate theory. 
In three space dimensions, $\mb v(\mb x,t)=\mb 0$ consists of three scalar equations. In a steady flow, these are three equations for the three coordinates $(x,y,z)$, and they generically have isolated solutions: the stagnation set consists of points. In an unsteady flow, the unknowns are $(x,y,z,t)$, four in all, so the solution set of $\mb v=\mb 0$ in space--time is generically a single curve, the world line of a stagnation point that moves as time passes. A stagnation \emph{line} that exists during a time interval and moves through space is a \emph{surface} in $(x,y,z,t)$-space. That is one dimension more than the generic count allows: along it the three equations cannot be independent, the Jacobian of $\mb v$ with respect to $(x,y,z,t)$ has rank two instead of three, and to first order one equation is a consequence of the other two. (A \emph{stagnation surface} is a set of dimension two in space and three in space--time. It needs two redundancies and occurs even in steady flow; see below.)

In the classical stagnation-point flows with a stagnation line (planar, radial; radial stagnation flow was first introduced by Wang~\cite{wang1974axisymmetric,cunning1998radial,rajamanickam2021steady}) and in the helical flow (discussed herein in Appendix~\ref{app:helix}, (a helically symmetric potential flow, distinct from the centre-type stagnation helix at the core of a helical vortex filament; \cref{rem:filament,rem:helixvortex})), the redundancy is supplied by a symmetry of the flow (translation, rotation, screw motion) together with the absence of velocity along the symmetry direction (\cref{cor:sym}). The shape of the line is then a consequence of the symmetry and is not a free datum. Symmetry is not the only way to obtain a persistent line: in a frame in which the flow is steady, every zero curve is frozen, and the centreline of a vortex core, for example, is stagnant in a frame moving with the core (\cref{rem:filament}).

The count fails more drastically for flows with fewer than three independent velocity components. For a unidirectional flow $\mb v=u(y,z,t)\,\mb e_x$, two of the three equations hold identically, so the zero set $u=0$ is a \emph{surface}: a cylinder over a curve in the $(y,z)$-plane, or a plane if $u=u(y,t)$. Such stagnation surfaces are unavoidable in shear flows. They are a different problem from stagnation lines (their velocity-gradient tensor is nilpotent of rank one, and potential flow admits none; they will be treated in a separate study. The shear-layer plane enters the line theory as the degenerate case $\Delta=0$ (\cref{sec:delta0}). The question we address for lines is what survives when the flow is unsteady in the chosen frame:
\begin{itemize}
\item if a stagnation curve of \emph{arbitrary} shape exists at one instant, how does it move?
\item under what condition does it persist, and what replaces it when it does not?
\end{itemize}
For an isolated stagnation \emph{point}, tracking the zero-velocity set follows a classical kinematic relation governed by an invertible velocity-gradient tensor $\mb L$, where generic bifurcations manifest as saddle--node creation or annihilation events~\cite{perry1987description,chong1990general}. In a compressible flow, the dilatation rate at the point is $\nabla\cdot\mb v=-\partial_t\ln\rho$, which vanishes in steady flow, so compressibility changes the classification of stagnation points only through unsteadiness (\cref{sec:comp}). For a continuous stagnation \emph{line}, however, this formulation breaks down because $\mb L$ inherently possesses a zero eigenvalue along the entire curve. To the author's knowledge, how a stagnation line of arbitrary shape dynamically moves, persists, or disintegrates in an unsteady, unconstrained flow has remained unexplored. The closest specialized study is that of Wilson et al.~\cite{wilson2005nested}, who showed that a localized, fixed obstacle destroys a steady stagnation line, replacing it with a static chain of alternating critical points. In contrast, this paper develops a general kinematic theory for unconstrained stagnation curves of arbitrary geometry, proving that their evolution is dictated by a distinct normal block inversion coupled to a dynamic tangential solvability constraint.

We first work out the consequences kinematically for any incompressible flow (\cref{sec:kin}), and then for two flows where the local acceleration $\mb a$ is determined by the data of the problem: irrotational flow, $\mb v=\nabla\varphi$ (\cref{sec:pot}), where the problem is cleanest because $\varphi_t$ is harmonic too and the pressure follows from Bernoulli's equation. Incompressibility enters the kinematics only through the trace of the normal block (\cref{rem:wherenc}), so \cref{sec:comp} extends the potential-flow results to inviscid, irrotational, homentropic compressible flow, where lines can also be of node type and the singular normal block arises at ordinary points. Stokes flow (\cref{sec:stokes}) follows, where vorticity enters and the comparison shows which features of the potential case are special. A harmonic potential flow is also a Stokes flow (with constant pressure), so the Stokes results contain the potential ones as the case of vanishing vorticity. \Cref{sec:closure} gives the closure for flows driven by moving point singularities. The Navier--Stokes equations are treated in \cref{sec:disc}, where the momentum equation supplies $\mb a$ at the price of making the solvability condition a property of the instantaneous solution, and an exact periodic example with a numerical check is given in \cref{sec:nsex}. The stagnation helix of a helically symmetric potential flow is described in Appendix~\ref{app:helix}. The degenerate case of a singular normal block is treated in \cref{sec:delta0}, and \cref{sec:outlook} gives limitations (including the exclusion of shocks from \cref{sec:comp}) and a note on stagnation surfaces. The main results are as follows:
\begin{itemize}
\item[(i)] For any smooth incompressible flow, the velocity $\mb U$ of a stagnation curve of arbitrary shape satisfies $\mb L\mb U=-\mb a$, so its normal velocity is $-\mb K^{-1}\mb a_n$, where $\mb K$ is the $2\times2$ normal block of $\mb L$ and $\mb a_n$ the normal part of $\mb a$ (\cref{prop:law}). Solvability of the tangential part is one condition $\Sigma=0$ per point of the curve, where $\Sigma=\mb a\cdot(\mb L\mb e_1\times\mb L\mb e_2)$ is the component of $\mb a$ out of the range of $\mb L$; it constrains the \emph{driving} of the flow, not the line. The normal motion determines the evolution of the curvature and torsion of the line (\cref{sec:geom}).
\item[(ii)] In potential flow, the normal velocity is the acceleration field divided by the principal strain rate $\ep$ (lines in weak strain move fast, and $\ep\to0$ signals a bifurcation), and the solvability condition $\Sigma=-\ep^2\partial_s\varphi_t=0$ states that $\varphi_t$, equivalently the pressure (i.e., the stagnation pressure), is uniform along the line (\cref{prop:potlaw}). Persistence on an interval imposes a hierarchy of such conditions, the second of which is given in \cref{prop:second}.
\item[(iii)] When the condition fails at some instant, then to first order in time, the line breaks into isolated stagnation points located at the zeros of $\Sigma$ (\cref{prop:break}). In potential flow, these are the critical points of $\varphi_t|_{\mathcal{C}}$; for a closed loop there are at least two, and they alternate in type (\cref{cor:potbreak}). Exact examples are given in \cref{sec:examples}.
\item[(iv)] In Stokes flow, the normal block contains the tangential vorticity $\omega_T$, so that a stagnation line is of saddle type only if $\ep>|\omega_T|/2$ and of centre type otherwise; the solvability condition involves the vorticity normal to the line; and the zeros of $\Sigma$ can be absent, so that a stagnation loop can dissolve without leaving any stagnation point (\cref{sec:stokes,ex:stokesper}).
\item[(v)] In Navier--Stokes flow, $\mb a=-\nabla p+\nu\nabla^2\mb v+\mb f$ on the line, so $\Sigma$ is a property of the instantaneous solution through the nonlocal pressure, and the first-order loss of a line can be predicted from the initial field alone (\cref{sec:NS}). Viscosity destroys a helical line through a source term proportional to the vorticity along the axis. An exact periodic example (\cref{sec:nsex}) is a sum of solenoidal Laplacian eigenmodes with different decay rates, related to the eddy solutions of Taylor and Walsh. A saddle line and a centre line exist only at $t=0$ and each splits into four points per period, with $\Sigma$ obtained from a Poisson equation solved in closed form. The first-order prediction is confirmed by a numerical solution of the full equations.
\item[(vi)] If the line is an orbit of a one-parameter group of rigid motions and if the velocity has no component along the orbit, the solvability condition holds identically (\cref{cor:sym}). In potential flow, the second requirement is automatic, so the straight, circular and helical stagnation lines persist. In Stokes flow, it is an additional restriction. The condition is sufficient but not necessary: $\Sigma\equiv0$ for any flow that is steady in a rigidly moving
frame, for example a vortex filament that translates and rotates without change of shape
(\cref{rem:filament}). The translating ring and the helix are orbits in that frame, but orbit symmetry is not required (\cref{ex:ellipse}).
\item[(vii)] The equation of motion involves $\mb a$ on the line, which is determined nonlocally by the boundary conditions and singularities of the flow. We give the closure for flows driven by moving point sources and Stokeslets, and count the third-order data of a potential-flow line that the line does fix and that it does not (\cref{sec:closure,prop:third}).
\item[(viii)] Where the normal block is singular ($\Delta=0$), $\mb K^{-1}$ and the normal velocity diverge, but $\Sigma$ stays regular. The number of solvability conditions is $3-\mathrm{rank}\,\mb L$. When the rank drops, a line that fails to persist is replaced by pairs of points born or annihilated at a distance $O(t^{1/2})$ from the point of the line where the rank drops, the line analogue of a saddle--node event (\cref{sec:delta0}); in compressible potential flow this occurs at ordinary points of a line (\cref{rem:fold}).
\item[(ix)] Incompressibility enters the kinematics only through $\mathrm{tr}\,\mb K=0$ (\cref{rem:wherenc}). For inviscid, irrotational, homentropic flow the dilatation rate on a stagnation line is $\theta=\varphi_{tt}/c^2$ (\cref{lem:theta}). The motion law becomes $U_i=-a_i/\lambda_i$ with $\lambda_1+\lambda_2=\theta$, the solvability function is $\Sigma=\lambda_1\lambda_2\,\partial_s\varphi_t$, and persistence needs the enthalpy, equivalently the pressure and density, uniform along the line. Lines are saddles or nodes, never centres, and the breaking law $\lambda_T=t\,\eta''$ is unchanged (\cref{prop:comp}). Exact examples in linear acoustics are a standing-wave node line and a two-frequency wave that breaks it into alternating sinks and saddles (\cref{ex:acoustic,ex:acoustic2}).
\end{itemize}

\subsection{Conventions}\label{sec:conv}

\begin{itemize}
\item Indices $i,j,\dots$ run over the Cartesian components $1,2,3$ and repeated indices are summed. The velocity-gradient tensor is $L_{ij}=\partial v_i/\partial x_j$, so that $(\mb L\mb w)_i=L_{ij}w_j=(\mb w\cdot\nabla)v_i$.
\item $\mb a(\mb x,t)=\partial_t\mb v(\mb x,t)$ is the rate of change of the velocity \emph{at a fixed point of space}, called the local acceleration. It is not the material derivative; the two coincide on a stagnation curve (\cref{sec:curve}).
\item The fluid has unit density. Where the pressure appears, $p$ is the pressure divided by the density, with any conservative body force absorbed into it. In \cref{sec:comp} the density $\rho$ is retained and $p$ is the pressure.
\item \emph{Order of approximation.} Statements obtained by differentiating the identity $\mb v=\mb 0$ along a stagnation curve that exists on a time interval ($\mb L\mb U=-\mb a$ and its consequences, \cref{prop:law,prop:potlaw,prop:second}) are exact at the instant considered; they involve no small parameter. Statements about what replaces a line that does \emph{not} persist (\cref{prop:break,cor:potbreak}) are asymptotic for small $|t|$: positions are given to $O(t)$ (errors $O(t^2)$) and eigenvalues of $\nabla\mb v$ to the order stated. Where we write $O(t^2)$ the symbol refers to the time elapsed since the instant $t=0$ at which the line exists. In \cref{sec:comp}, \cref{lem:theta} and \cref{prop:comp}(i), (ii), (iv) are exact in the same sense (the model is the only approximation), \cref{prop:comp}(iii) is first order in $t$, \cref{rem:fold} is $O(t^{1/2})$, and the acoustic examples are exact within linear acoustics.
\item Each result is introduced by a physical argument, then derived in numbered steps, then followed by a paragraph on its meaning. 
\end{itemize}

\section{Kinematics of a stagnation curve}\label{sec:kin}

\subsection{The curve and its frame}\label{sec:curve}

Let $\mb v(\mb x,t)$ be a smooth velocity field with $\nabla\cdot\mb v=0$, and suppose that for $t$ in some interval the zero set of $\mb v$ contains a smooth curve
\be
\mc C(t)=\{\mb X(\sigma,t)\},
\label{eq:curve}
\ee
parametrised by a label $\sigma$ that is fixed to a point of the curve (a different $\sigma$ is a different point; the same $\sigma$ at a later time is the point that has been \emph{tracked}, in a way we are free to choose).

At a given instant, let $s=s(\sigma,t)$ be the arclength along $\mc C(t)$, so that $\mb t=\partial_s\mb X$ is the unit tangent. The curvature $\kappa\ge0$ and torsion $\tau$ are defined by the Frenet--Serret equations for the frame $(\mb t,\mb n,\mb b)$ (unit tangent, principal normal, binormal, with $\mb b=\mb t\times\mb n$):
\begin{equation}
    \partial_s\mb t=\kappa\mb n, \qquad \partial_s\mb n=-\kappa \mb t + \tau \mb b, \qquad \partial_s\mb b=-\tau\mb n.
    \label{eq:Frenet}
\end{equation}
When $\kappa=0$ the principal normal is not defined; it is then replaced by any smooth orthonormal frame of the normal plane, and only formulas containing $\kappa\mb n$ change trivially (that product is zero).

The tracked point has velocity $\mb U(\sigma,t)=\partial_t\mb X(\sigma,t)$ at fixed $\sigma$. Resolve it into a part along the curve and a part in the normal plane,
\be
\mb U=U_T\,\mb t+\mb U_n,\qquad \mb U_n\cdot\mb t=0 .
\ee
Only $\mb U_n$ changes the \emph{shape and position} of the curve. A displacement $U_T\mb t\,\mathrm{d}t$ slides a point along the curve it already lies on, to first order in $\mathrm{d}t$, and so only changes which label is attached to which point. $U_T$ is therefore a matter of parametrisation, and we are free to choose it; $U_T=0$ is the simplest choice.

Finally the local acceleration on the curve is
\be
\mb a=\partial_t\mb v\big|_{\mc C},
\label{eq:acc}
\ee
the rate of change of the velocity at the fixed point of space $\mb X$. Since $\mb v=\mb 0$ on $\mc C$, the convective term in the material derivative, $(\mb v\cdot\nabla)\mb v$, vanishes there, so $\mb a=\mathrm{D}\mb v/\mathrm{D}t$ on $\mc C$: it is the acceleration of the fluid particle that is instantaneously on the curve.

\begin{remark}[Stagnation curves do not end in the fluid]\label{rem:ends}
A stagnation curve cannot simply stop in the interior of the fluid. Near a point of the curve where the velocity-gradient tensor has rank two, two velocity components vanish together on a single curve through the point, and the third is a smooth function along that curve. If it vanishes at the point, it keeps vanishing along the curve on both sides, since a smooth fluid velocity has no reason to switch off abruptly. So the curve continues. In the fluid interior, stagnation curves therefore form closed loops, extend to infinity or repeat periodically. The only exceptions are isolated points where the rank of $\mb L$ drops (\cref{sec:delta0}), at which curves can merge, split or end, and the boundary of the fluid domain, where a curve can terminate.
At a no-slip wall, the whole wall is a zero set in the frame of the wall, and such lines are not considered.
\end{remark}

\subsection{The velocity-gradient tensor on the curve}\label{sec:grad}

\emph{Step 1: no strain along the line.} The identity $\mb v(\mb X(s),t)=\mb 0$ holds for every $s$ at a given $t$. Differentiating it with respect to $s$ by the chain rule,
\be
0=\frac{\partial v_i}{\partial s}=\frac{\partial v_i}{\partial x_j}\frac{\partial X_j}{\partial s}=L_{ij}t_j,\qquad\text{that is}\qquad \mb L\mb t=\mb 0\quad\text{on }\mc C.
\label{eq:LT}
\ee
\emph{Meaning.} A short fluid element lying along the curve is neither stretched nor rotated by the local flow; because its separation vector $\delta\mb x$ is parallel to the unit tangent ($\delta\mb x \parallel \mb t$), its rate of deformation is given by $\mb L\,\delta\mb x = \mb L(\delta s\,\mb t) = \mb 0$. Consequently, the surrounding linear velocity field $\mb v(\mb X(s)+\gb\delta,t)\simeq \mb L\gb\delta$ (where $\gb\delta$ is small normal distance vector) looks locally identical from one point to the next as you move along the backbone of the curve. This spatial independence to first order highlights the fundamental property that distinguishes a stagnation line from an isolated stagnation point: the velocity-gradient tensor $\mb L$ inherently possesses a zero eigenvalue along the entire curve, meaning its inverse $\mb L^{-1}$ does not exist.

\emph{Step 2: structure of $\mb L$.} At each point of $\mc C$, choose a right-handed orthonormal basis $(\mb e_1,\mb e_2,\mb t)$ with $\mb e_1\times\mb e_2=\mb t$, where $\mb e_1,\mb e_2$ span the normal plane. For any vector $\mb w$, we write $w_\alpha=\mb e_\alpha\cdot\mb w$ ($\alpha=1,2$), $w_T=\mb t\cdot\mb w$, and $\mb w_n=w_1\mb e_1+w_2\mb e_2$. The matrix of $\mb L$ in this basis has entries $\mb e_i\cdot\mb L\mb e_j$ (with $\mb e_3\equiv\mb t$). By \eqref{eq:LT}, its third column is zero, so
\begin{equation}
\begin{split}
\mb L &= \begin{pmatrix} K_{11} & K_{12} & 0 \\ K_{21} & K_{22} & 0 \\ m_1 & m_2 & 0 \end{pmatrix} = \begin{pmatrix} \mb K & \mb 0 \\ \mb m^{\mathsf T} & 0 \end{pmatrix}, \\[1ex]
&\quad K_{\alpha\beta} = \mb e_\alpha \cdot \mb L \mb e_\beta, \qquad m_\beta = \mb t \cdot \mb L \mb e_\beta = \partial_\beta v_T .
\end{split}
\label{eq:Lstruct}
\end{equation}
Equivalently, $\mb L\mb e_\beta=K_{\alpha\beta}\mb e_\alpha+m_\beta\mb t$ and $\mb L\mb t=\mb 0$. Here,
\begin{itemize}
\item $\mb K$ is the $2\times2$ block that describes how the \emph{normal} velocity $\mb v_n$ varies with normal position: a point displaced by a small normal vector $\gb\delta$ from the curve has normal velocity $\mb K\gb\delta$;
\item $\mb m$ is the shear of the \emph{axial} velocity across the line: the same point has axial velocity $\mb m\cdot\gb\delta$.
\end{itemize}
The third row of $\mb L$ is not forced to vanish. Nothing requires $\mb L$ to be symmetric, in general.

\emph{Step 3: trace.} Incompressibility says $\mathrm{tr}\,\mb L=L_{11}+L_{22}+L_{33}=0$. Because $L_{33}=\mb t\cdot\mb L\mb t=0$ by \eqref{eq:LT},
\be
\mathrm{tr}\,\mb K=K_{11}+K_{22}=0 .
\label{eq:trK}
\ee
\emph{Meaning.} With no strain along the line, the area of a small patch of the normal plane is conserved by the normal flow. The normal flow is a pure strain plus rotation.

\emph{Step 4: vorticity.} The vorticity $\gb\omega=\nabla\times\mb v$ has components $\omega_i=\epsilon_{ijk}\partial_jv_k$. In the right-handed basis $(\mb e_1,\mb e_2,\mb t)$, with $\partial_\alpha v_\beta=L_{\beta\alpha}=K_{\beta\alpha}$, $\partial_\beta v_T=m_\beta$ and $\partial_Tv_\beta=L_{\beta T}=0$ (third column of $\mb L$), we have
\be
\omega_T=K_{21}-K_{12},\qquad
\omega_1=m_2,\qquad
\omega_2=-m_1 .
\label{eq:vort}
\ee
So the antisymmetric part of $\mb K$ is the vorticity along the line, and $\mb m=(\,-\omega_2,\omega_1)$ is the vorticity \emph{normal} to the line rotated through a right angle in the normal plane. In particular $\mb m\cdot\gb\omega_n=-\omega_2\omega_1+\omega_1\omega_2=0$. \emph{Meaning.} Axial shear $\partial_\beta v_T$ across a line is the same thing as vorticity normal to it: a line with $\mb m\neq\mb 0$ is crossed by vortex lines.

\emph{Step 5: decomposition of $\mb K$.} Write $\mb K$ as symmetric part plus antisymmetric part. The symmetric part $\mb S=\tfrac12(\mb K+\mb K^{\mathsf T})$ is traceless by \eqref{eq:trK}, so its eigenvalues are $\pm\ep$ with $\ep\ge0$, the \emph{principal strain rate in the normal plane}. The antisymmetric part is $\tfrac12(K_{12}-K_{21})\begin{pmatrix}0&1\\-1&0\end{pmatrix}=\tfrac12\omega_T\mb J$ by \eqref{eq:vort}, with $\mb J=\begin{pmatrix}0&-1\\1&0\end{pmatrix}$ the rotation through $+\pi/2$. Hence
\be
\mb K=\mb S+\tfrac12\omega_T\mb J .
\label{eq:Kdecomp}
\ee
Put $p=K_{11}$, $q=\tfrac12(K_{12}+K_{21})$ and $w=\tfrac12\omega_T$. Then
\[
\mb K=\begin{pmatrix}p&q-w\\ q+w&-p\end{pmatrix},\qquad \ep^2=p^2+q^2.
\]
Its determinant is $\det\mb K=-p^2-(q-w)(q+w)=-(p^2+q^2)+w^2$, that is
\be
\det\mb K=-\Delta,\qquad \Delta:=\ep^2-\tfrac14\omega_T^2 .
\label{eq:detK}
\ee
Squaring $\mb K$ entry by entry: $(\mb K^2)_{11}=p^2+(q-w)(q+w)=\Delta$, $(\mb K^2)_{12}=p(q-w)-(q-w)p=0$, $(\mb K^2)_{21}=(q+w)p-p(q+w)=0$, $(\mb K^2)_{22}=(q+w)(q-w)+p^2=\Delta$. Therefore
\be
\mb K^2=\Delta\,\mb I,\qquad \mb K^{-1}=\frac{\mb K}{\Delta}\quad(\Delta\neq0).
\label{eq:Ksq}
\ee
(Equation \eqref{eq:Ksq} is the Cayley--Hamilton relation for a traceless $2\times2$ matrix.) The eigenvalues of $\mb K$ are $\pm\Delta^{1/2}$.

\emph{Step 6: the flow near the line in the normal plane.} To first order in the distance $\gb\delta$ from the line, the normal velocity is $\mb K\gb\delta$, so fluid particles move according to $\dot{\gb\delta}=\mb K\gb\delta$ and hence $\gb\delta(\vartheta)=e^{\mb K\vartheta}\gb\delta(0)$. Using \eqref{eq:Ksq}, in the series for the exponential, even powers of $\mb K$ reduce to powers of $\Delta$ times $\mb I$ and odd powers to powers of $\Delta$ times $\mb K$, so
\be
\gb\delta(\vartheta)=\Big[\cosh(\Delta^{1/2}\vartheta)\,\mb I+\frac{\sinh(\Delta^{1/2}\vartheta)}{\Delta^{1/2}}\,\mb K\Big]\gb\delta(0)\quad(\Delta>0),
\label{eq:expK}
\ee
and the same with $\cosh\to\cos$, $\sinh\to\sin$, $\Delta^{1/2}\to(-\Delta)^{1/2}$ if $\Delta<0$. (Check: for $\mb K=\mathrm{diag}(\ep,-\ep)$ this gives $\mathrm{diag}(e^{\ep\vartheta},e^{-\ep\vartheta})$.) The curve is therefore of
\begin{itemize}
\item \emph{saddle type} if $\ep>|\omega_T|/2$ ($\Delta>0$): nearby fluid is drawn in along one direction and expelled along another at the exponential rate $\Delta^{1/2}$;
\item \emph{centre type} if $\ep<|\omega_T|/2$ ($\Delta<0$): nearby fluid circulates around the line with angular frequency $(-\Delta)^{1/2}$ and does not escape; the rotation wins over the strain;
\item \emph{degenerate} if $\ep=|\omega_T|/2$ ($\Delta=0$): $\mb K$ is singular and nearby fluid moves linearly in time.
\end{itemize}
In irrotational flow, $\gb\omega=\mb 0$, so $\mb m=\mb 0$, $\mb K=\mb S$ and $\Delta=\ep^2>0$: a stagnation line of potential flow is always of saddle type. Centre-type lines require vorticity.

\subsection{The equation of motion}

Let $\mc C$ be a stagnation curve at $t=0$. At a short time $t$ later, the fluid at the point $\mb X+\gb\delta$ near the curve has velocity, by Taylor expansion in time (using $\mb v=\mb 0$ on $\mc C$ at $t=0$) and in position,
\be
\mb v(\mb X+\gb\delta,t)=t\,\mb a+\mb L\gb\delta+O(t^2,\,t|\gb\delta|,\,|\gb\delta|^2).
\label{eq:Taylor}
\ee
The first term is the velocity gained by the fluid on the curve under the local acceleration; the second is the velocity produced by displacing the point. For the displaced point to be stagnant at time $t$, $\gb\delta=t\mb U$ must be chosen so that $t(\mb a+\mb L\mb U)=O(t^2)$. This is the condition $\mb a+\mb L\mb U=\mb 0$ below. The point is that $\mb L$ is singular: $\mb L\mb U$ cannot produce a velocity out of the plane spanned by the columns of $\mb L$, and an acceleration with a component out of that plane cannot be compensated by moving the point. The component of $\mb a$ out of that plane is measured by the \emph{solvability function}
\be
\Sigma:=\mb a\cdot(\mb L\mb e_1\times\mb L\mb e_2)=\mb m\cdot\mb K\mb a_n-\Delta\,a_T ,
\label{eq:Sigma}
\ee
where the cross product is taken in the right-handed basis $(\mb e_1,\mb e_2,\mb t)$. The vector $\mb L\mb e_1\times\mb L\mb e_2$ is normal to the plane spanned by the columns of $\mb L$; it equals $\mathrm{cof}(\mb L)\,\mb t$, so $\Sigma$ does not depend on the choice of $(\mb e_1,\mb e_2)$ in the normal plane. The second form follows by expanding the cross product with \eqref{eq:Lstruct}, using $\mathrm{tr}\,\mb K=0$ and noting that its $\mb t$-component is $\det\mb K=-\Delta$. $\Sigma$ is a polynomial in the entries of $\mb L$ and $\mb a$; it involves no inverse and is defined whether or not $\Delta$ vanishes. The following proposition makes the motion law and the role of $\Sigma$ precise.

\begin{proposition}[Motion law and solvability]\label{prop:law}
Let $\mc{C}(t)$ be a stagnation curve for $t$ in an interval, with $\Delta\neq0$ everywhere on it, and let $\mb{U}=\partial_t\mb{X}$ be the velocity of a tracked point (at fixed label $\sigma$). Then, exactly, at every point of $\mc C$:
\begin{itemize}
\item[(i)] $\mb a+\mb L\mb U=\mb 0$.
\item[(ii)] \emph{Normal velocity.} $\mb U_n=-\mb K^{-1}\mb a_n=-\mb K\mb a_n/\Delta$.
\item[(iii)] \emph{Solvability along the curve.} $\Sigma=0$, that is,
\be
\Sigma=\mb m\cdot\mb K\mb a_n-\Delta\,a_T=0 .
\label{eq:solv}
\ee
Equivalently, $a_T+\mb m\cdot\mb U_n=0$, that is $a_T+\mb t\cdot(\gb\omega\times\mb U)=0$, because $\Sigma=-\Delta\,(a_T+\mb m\cdot\mb U_n)$.
\item[(iv)] The tangential velocity $U_T$ is arbitrary (it only reparametrises $\mc C$).
\item[(v)] The line element $\mathrm{d}s$ of a tracked segment changes at the rate $\partial_t(\mathrm{d}s)=(\partial_sU_T-\kappa\,\mb n\cdot\mb U)\,\mathrm{d}s$.
\end{itemize}
\end{proposition}

\begin{proof}
\emph{Step 1 (differentiate the identity).} For all $t$ in the interval, $\mb v(\mb X(\sigma,t),t)=\mb 0$ at fixed $\sigma$. The $t$-derivative, by the chain rule, is $\partial_t\mb v+U_j\,\partial_j\mb v=\mb a+\mb L\mb U=\mb 0$, which is (i). This agrees with the physical argument: the velocity of the tracked point, \eqref{eq:Taylor} with $\gb\delta=t\mb U$, is $t(\mb a+\mb L\mb U)$ to first order.

\emph{Step 2 (remove the tangential velocity).} Since $\mb L\mb t=\mb 0$, the component $U_T\mb t$ of $\mb U$ gives no contribution: $\mb L\mb U=\mb L\mb U_n$. So (i) reads $\mb L\mb U_n=-\mb a$.

\emph{Step 3 (split into components).} Insert the structure \eqref{eq:Lstruct}: $\mb L\mb U_n=\mb K\mb U_n+(\mb m\cdot\mb U_n)\,\mb t$, with $\mb K\mb U_n$ in the normal plane and $(\mb m\cdot\mb U_n)\mb t$ along the tangent. The equation $\mb L\mb U_n=-\mb a$ therefore splits into two,
\begin{align}
&\mb K\mb U_n=-\mb a_n\quad(\text{two scalar equations}),\\ &\mb m\cdot\mb U_n=-a_T\quad(\text{one scalar equation}).
\end{align}

\emph{Step 4 (solve the normal part).} For $\Delta\neq0$, $\mb K$ is invertible and $\mb K^{-1}=\mb K/\Delta$ by \eqref{eq:Ksq}, so $\mb U_n=-\mb K^{-1}\mb a_n=-\mb K\mb a_n/\Delta$. This is (ii). It determines $\mb U_n$ \emph{uniquely}.

\emph{Step 5 (the left-over equation).} $\mb U_n$ is now known, and the tangential equation has no free unknown left. It is a condition on $\mb a$ and $\mb L$: $a_T+\mb m\cdot\mb U_n=a_T-\mb m\cdot\mb K^{-1}\mb a_n=0$. Multiplying by $-\Delta\neq0$ and using $\mb K^{-1}=\mb K/\Delta$ gives $\mb m\cdot\mb K\mb a_n-\Delta a_T=\Sigma=0$, that is \eqref{eq:solv}. By \eqref{eq:vort}, $\mb m\cdot\mb U_n=m_1U_1+m_2U_2=-\omega_2U_1+\omega_1U_2$, which is the tangential component of $\gb\omega\times\mb U$ in a right-handed basis (its $\mb t$-component is $\omega_1U_2-\omega_2U_1$). This proves (iii).

\emph{Step 6 (the line element).} Write $\mathrm{d}s=|\partial_\sigma\mb X|\,\mathrm{d}\sigma$. Then $\partial_t|\partial_\sigma\mb X|=\partial_\sigma\mb X\cdot\partial_\sigma\mb U/|\partial_\sigma\mb X|=|\partial_\sigma\mb X|\,\mb t\cdot\partial_s\mb U$. Moreover, $\mb t\cdot\partial_s\mb U=\partial_s(\mb t\cdot\mb U)-\mb U\cdot\partial_s\mb t=\partial_sU_T-\kappa\,\mb n\cdot\mb U$, by \eqref{eq:Frenet}. This proves (v). Part (iv) was shown in \cref{sec:curve}.
\end{proof}

\emph{Meaning.}
\begin{itemize}
\item \emph{Geometry of the solvability condition.} The range of $\mb L$, the set of all vectors $\mb L\mb w$, is the plane $\mc P=\{\mb K\mb w_n+(\mb m\cdot\mb w_n)\mb t\}$ spanned by $\mb L\mb e_1$ and $\mb L\mb e_2$ (a vector $\mb w_n$ of the normal plane together with the axial component tied to it by $\mb m$), and its normal is $\mb L\mb e_1\times\mb L\mb e_2$. Thus the equation $\mb L\mb U=-\mb a$ can be solved precisely when $-\mb a$ lies in $\mc P$, that is when the component of $\mb a$ along this normal vanishes: $\Sigma=0$. (For $\Delta\neq0$ this is the statement that $w_T=\mb m\cdot\mb K^{-1}\mb w_n$ for $\mb w=-\mb a$.) This is the Fredholm alternative: solvability of a singular linear system requires the right-hand side to be orthogonal to the null space of the transposed matrix.
\item \emph{Why the line cannot adjust.} The acceleration $\mb a$ gives every fluid particle on the line an extra velocity $t\mb a$. Displacing the point by $\gb\delta$ adds $\mb L\gb\delta\in\mc P$. If $\mb a\notin\mc P$, the point cannot be made stagnant at any displacement: the sum cannot vanish, and the stagnation set cannot remain a curve. Fluid on the curve acquires a velocity of order $t$ out of the plane, and the line is lost (\cref{sec:loss}). When $\Sigma=0$, the acceleration of the fluid on the line is exactly balanced by the linear flow produced by a normal displacement of the line, and the line persists and moves with normal velocity $\mb U_n$.
\item \emph{Which quantity constrains which.} The solvability condition $\Sigma=0$ involves $\mb a$, $\mb K$ and $\mb m$, but not the shape of the curve. It restricts what the driving of the flow may do to a line that exists.
\item \emph{Stretching.} The tangential velocity $U_T$ is a choice of
labelling and has no physical content: the term $\partial_sU_T$ in (v) only records labelled points sliding along a curve that is not itself changing, like beads on a wire, and it integrates to zero around a closed loop. The length of the \emph{geometric} curve changes only through its normal motion. A part of the line that moves towards its centre of curvature (the direction of $\mb n$) is shortened, and a part that moves away from it is lengthened, at the rate $-\kappa\,\mb n\cdot\mb U_n$ per unit length. For a closed loop, the total length changes at the rate $\mathrm dL/\mathrm dt=-\oint\kappa\,\mb n\cdot\mb U_n\,\mathrm ds$. (Check: a circle of radius $R$ moving inwards at speed $w$ has $\kappa=1/R$, $\mb n\cdot\mb U_n=w$, so $\mathrm dL/\mathrm dt=-2\pi w=\mathrm d(2\pi R)/\mathrm dt$.)
\end{itemize}

\begin{remark}[Weak strain and degenerate normal block]\label{rem:weak}
By (ii), the normal velocity is $\mb K\mb a_n/\Delta$, so it grows like $\Delta^{-1}$ (times $\|\mb K\|$) as $\Delta\to0$; in potential flows, $\mb K=\mathrm{diag}(\ep,-\ep)$ and this is $|\mb a_n|/\ep$. The physical reason is that the restoring linear flow $\mb K\gb\delta$, which balances the acceleration, becomes weak, so a small acceleration must be balanced by a large displacement. Where $\Delta\to0$ along the line, the first-order description fails and the stagnation set can bifurcate there (a line can merge with another line or become a different point set). The equation of motion is valid only away from such points. The solvability function $\Sigma$ itself, \eqref{eq:Sigma}, is regular at $\Delta=0$: the singularity comes only from $\mb K^{-1}$. The degenerate case $\Delta=0$, in which $\mb K$ is nilpotent, is analysed in \cref{sec:delta0}: the number of solvability conditions grows with the rank deficiency of $\mb L$, and the $\Delta^{-1}$ divergence of $\mb U_n$ is replaced by displacements of order $t^{1/2}$ and by pairs of stagnation points created or annihilated along the line.
\end{remark}

\begin{remark}[Where incompressibility enters]\label{rem:wherenc}
Incompressibility is used only in Step~3, through \eqref{eq:trK}. For a general smooth flow, put
$\theta:=\nabla\cdot\mb v=\mathrm{tr}\,\mb L=\mathrm{tr}\,\mb K$ on $\mc C$ (the dilatation rate; $-\theta$ is the invariant $P$ of Chong, Perry and Cantwell~\cite{chong1990general}).
Write $\mb K=\tfrac12\theta\mb I+\mb S_0+\tfrac12\omega_T\mb J$, with $\mb S_0$ symmetric, traceless, with eigenvalues $\pm\ep$, and $\Delta=\ep^2-\tfrac14\omega_T^2$ as before.
The Cayley--Hamilton relation gives $\det\mb K=\tfrac14\theta^2-\Delta$, $\mb K^{-1}=(\theta\mb I-\mb K)/\det\mb K$, and the eigenvalues of $\mb K$ are $\tfrac12\theta\pm\Delta^{1/2}$.
\Cref{prop:law} holds with
\[
\mb U_n=(\mb K-\theta\mb I)\mb a_n/\det\mb K,\qquad
\Sigma=\det\mb K\,a_T+\mb m\cdot(\mb K-\theta\mb I)\mb a_n,
\]
and in \cref{prop:break} the eigenvalues at $\mb X_j$ are $\tfrac12\theta\pm\Delta^{1/2}+O(t)$ and $\lambda_T=t\,\Sigma'(s_j)/\det\mb K+O(t^2)$, with $\det\nabla\mb v=t\Sigma'(s_j)+O(t^2)$ as before.
The line is a saddle if $\Delta>\theta^2/4$, a node if $0\le\Delta<\theta^2/4$, and a focus (centre if $\theta=0$) if $\Delta<0$.
\Cref{cor:sym} is unchanged. The consequences for potential flow are worked out in \cref{sec:comp}.
\end{remark}

\subsection{Evolution of curvature and torsion}\label{sec:geom}

Let the curve move with $\mb U=U_N\mb n+U_B\mb b$ (we take $U_T=0$; a tangential part only adds the
convective terms $U_T\partial_s\kappa$, etc.). By \cref{prop:law}(v), $\partial_t(\mathrm ds)=-\kappa U_N\,\mathrm ds$, so
$[\partial_t,\partial_s]=\kappa U_N\,\partial_s$ since $\partial_t\partial_s=-\partial_t(ds)\partial_s + \partial_s\partial_t$. Put
\[
A:=\partial_sU_N-\tau U_B,\qquad B:=\partial_sU_B+\tau U_N .
\]

\begin{proposition}[Geometry under normal motion]\label{prop:geom}
For $\kappa>0$,
\begin{align}
\partial_t\mb t&=A\,\mb n+B\,\mb b,\\
\partial_t\kappa&=\partial_sA-\tau B+\kappa^2U_N
=\partial_s^2U_N+(\kappa^2-\tau^2)U_N \nonumber \\ &\quad-2\tau\,\partial_sU_B-(\partial_s\tau)\,U_B,\label{eq:kappat}\\
\partial_t\tau&=\partial_s\big(\omega_F\big)+\kappa B+\kappa\tau U_N,\qquad
\omega_F:=\frac{\partial_sB+\tau A}{\kappa},\label{eq:taut}
\end{align}
where $\partial_t\mb n=-A\,\mb t+\omega_F\,\mb b$ (so $\omega_F$ is the rotation rate of the Frenet frame about $\mb t$).
\end{proposition}

\begin{proof}
Using $[\partial_t,\partial_s]$ above, $\partial_t\mb t=\partial_s\mb U+\kappa U_N\mb t$, and \eqref{eq:Frenet}
gives the first line. Next $\partial_t(\kappa\mb n)=\partial_s(\partial_t\mb t)+\kappa U_N\,\kappa\mb n$.
The $\mb n$-component gives \eqref{eq:kappat}, the $\mb b$-component gives $\kappa\,\mb b\cdot\partial_t\mb n=\partial_sB+\tau A$, and
the $\mb t$-component is consistent with $\mb t\cdot\partial_t\mb n=-A$. Finally $\partial_t(\partial_s\mb n)=\partial_s(\partial_t\mb n)+\kappa U_N\partial_s\mb n$
with $\partial_s\mb n=-\kappa\mb t+\tau\mb b$; its $\mb b$-component is \eqref{eq:taut}.
\end{proof}

\emph{Check.} A circle of radius $R$ with $U_N=w$ (towards the centre) has $\partial_t\kappa=\kappa^2w=-\dot R/R^2$ with $\dot R=-w$.
For a planar curve ($\tau=U_B=0$), $\partial_t\kappa=\partial_s^2U_N+\kappa^2U_N$.

\begin{remark}[Not a closed system]
In \cref{prop:geom}, $U_N,U_B$ are given by \cref{prop:law}(ii) in terms of $\mb K$ and $\mb a_n$. In potential flow
they involve $\ep$, $\alpha$ and $\nabla\varphi_t$ (strain rate and its orientation angle and acceleration), and the evolution of $\ep,\alpha$ requires third-order data of
$\varphi$ that the line does not determine (\cref{prop:third}). So \eqref{eq:kappat}--\eqref{eq:taut} do not give an
autonomous system for $(\kappa,\tau)$. They are useful when $\mb U_n$ is available, as when the line is advanced
numerically with the singular representation of \cref{sec:closure}. The expression \eqref{eq:taut} is singular where $\kappa=0$.
\end{remark}

\subsection{Loss of a line}\label{sec:loss}

Because the solvability condition constrains the driving of the flow via $\mb a$, a generic unsteady perturbation of a stagnation line violates it. We now find what replaces the line, for small $t$. At a point $s$ of $\mc C$, we look at the normal plane. There the two normal components of the velocity can always be made to vanish at one point of the plane, a displacement $\gb\delta^*(s,t)$ from the curve, because $\mb K$ is invertible. As $s$ varies, these points trace out a new curve $\mc C^*(t)$, the zero set of the normal velocity. Along $\mc C^*$ only the axial velocity is left, $v_T=-t\,\Sigma(s)/\Delta(s)+O(t^2)$ (below). The true stagnation points are the points of $\mc C^*$ at which the axial velocity also vanishes, and these are the zeros of $\Sigma$.

\begin{proposition}[Loss of a line, first order in $t$]\label{prop:break}
Let $\mc{C}=\mc{C}(0)$ be a closed stagnation curve of $\mb v(\cdot,0)$ (or one period of a periodic curve), with $\Delta\neq0$ on it, and let $\Sigma(s)$ be the function \eqref{eq:Sigma} evaluated at $t=0$. Suppose that all zeros $s_1,\dots,s_N$ of $\Sigma$ are simple ($N=0$ is allowed). Then there exists $t_0>0$ such that for $0<|t|<t_0$, the zero set of $\mb v(\cdot,t)$ in a neighbourhood of $\mc C$ consists of exactly $N$ isolated points,
\be
\mb{X}_j(t)=\mb{X}(s_j)+t\left(-\mb K^{-1}\mb a_n+c_j\mb t\right)+O(t^2),
\label{eq:pts}
\ee
where $\mb K$ and $\mb a_n$ are evaluated at $s_j$ and $c_j$ is a tangential drift speed that is fixed only at the next order in $t$. The eigenvalues of $\nabla\mb v$ at $\mb X_j(t)$ are
\be
\lambda_{1,2}=\pm\Delta^{1/2}+O(t),\qquad \lambda_T=-\frac{t\,\Sigma'(s_j)}{\Delta}+O(t^2),
\label{eq:eig}
\ee
where $\Delta$ is evaluated at $s_j$ and $\lambda_{1,2}$ are a conjugate imaginary pair if $\Delta<0$.
\end{proposition}

\begin{proof}
\emph{Step 1 (the curve $\mc C^*$ where the normal velocity vanishes).} At a point $\mb X(s)$, consider the displaced point $\mb X(s)+\gb\delta$ with $\gb\delta$ in the normal plane. Using \eqref{eq:Taylor}, its normal velocity is
\[
\mb v_n(\mb X+\gb\delta,t)=t\,\mb a_n+\mb K\gb\delta+O(t^2,\,t|\gb\delta|,\,|\gb\delta|^2).
\]
Setting this to zero and solving for $\gb\delta$, which is possible by the implicit function theorem because $\mb K$ is invertible, gives the unique small solution
\be
\gb\delta^*(s,t)=-t\,\mb K^{-1}\mb a_n+O(t^2).
\ee
\emph{Step 2 (the axial velocity on $\mc C^*$).} At the point $\mb X(s)+\gb\delta^*$ the axial velocity is, by the same Taylor expansion,
\begin{align}
v_T&=t\,a_T+\mb m\cdot\gb\delta^*+O(t^2)=t\,(a_T-\mb m\cdot\mb K^{-1}\mb a_n)+O(t^2) \nonumber \\&=-\frac{t\,\Sigma(s)}{\Delta(s)}+O(t^2),
\end{align}
because $\mb v(\mb X,0)=\mb 0$ makes the axial velocity of the fluid on the curve $t\,a_T+O(t^2)$ and the displacement contributes $\mb m\cdot\gb\delta^*$; the last equality uses $a_T-\mb m\cdot\mb K^{-1}\mb a_n=-\Sigma/\Delta$, from \eqref{eq:Sigma} and \eqref{eq:Ksq}. Define $H(s,t)$ by $v_T(\mb X(s)+\gb\delta^*(s,t),t)=t\,H(s,t)$; then $H$ is smooth and $H(s,t)=-\Sigma(s)/\Delta(s)+O(t)$.

\emph{Step 3 (the zeros).} A stagnation point near $\mc C$ must lie on $\mc C^*$ (normal velocity zero) and have $H=0$. At $t\to0$, the zeros of $H$ are the zeros of $\Sigma$ (since $\Delta\neq0$). If $s_j$ is a simple zero, $\Sigma'(s_j)\neq0$, so $\partial_sH=-\Sigma'(s_j)/\Delta\neq0$, and the implicit function theorem shows that the zero persists as a unique, smooth path $s_j(t)=s_j+c_jt+O(t^2)$ for small $|t|$. Away from the zeros of $\Sigma$, $|H|>0$ for small $t$, so there are no other stagnation points. This gives exactly $N$ points. Substituting $s_j(t)$ into $\mb X(s)+\gb\delta^*(s,t)$ and expanding, $\mb X(s_j(t))=\mb X(s_j)+c_jt\,\mb t+O(t^2)$ and $\gb\delta^*=-t\mb K^{-1}\mb a_n+O(t^2)$, gives \eqref{eq:pts}.

\emph{Step 4 (the Jacobian of the reduced problem).} Let $\mb Y(s,\gb\delta)=\mb X(s)+\delta_1\mb e_1(s)+\delta_2\mb e_2(s)$ and $\mb G(s,\gb\delta)=\mb v(\mb Y,t)$. The Jacobian of $\mb G$ with respect to $(\delta_1,\delta_2,s)$ is $\mb J_t=\nabla\mb v\cdot\partial\mb Y/\partial(\delta_1,\delta_2,s)$, and since $\gb\delta=O(t)$, $\partial\mb Y/\partial(\delta_1,\delta_2,s)=(\mb e_1,\mb e_2,\mb t)+O(t)$. So $\mb J_t=\nabla\mb v\,(\mb I+O(t))$ in the basis $(\mb e_1,\mb e_2,\mb t)$. Block it as
\begin{align}
\mb J_t=\begin{pmatrix}\mb A&\mb B\\ \mb C&D\end{pmatrix},\quad \mb A=\partial_{\gb\delta}\mb G_n=\mb K+O(t),\\ \mb B=\partial_s\mb G_n,\ \ \mb C=\partial_{\gb\delta}G_T,\ \ D=\partial_sG_T.
\end{align}
\emph{Step 5 (elimination of the normal unknowns).} Along $\mc C^*$, we have $\mb G_n(s,\gb\delta^*(s))=\mb 0$. Differentiating in $s$ gives $\mb A\,\partial_s\gb\delta^*+\mb B=\mb 0$, so $\partial_s\gb\delta^*=-\mb A^{-1}\mb B$. The axial velocity along $\mc C^*$ is $\Phi(s)=G_T(s,\gb\delta^*(s))=t\,H(s,t)$, and
\[
t\,\partial_sH=\partial_s\Phi=D+\mb C\,\partial_s\gb\delta^*=D-\mb C\mb A^{-1}\mb B .
\]
Subtracting $\mb C\mb A^{-1}$ times the upper block row from the lower block row of $\mb J_t$ does not change the determinant (Gaussian elimination) and turns the matrix into one with the entry $D-\mb C\mb A^{-1}\mb B$ in the lower right corner, so $\det\mb J_t=\det\mb A\,(D-\mb C\mb A^{-1}\mb B)$. Therefore
\be
t\,\partial_sH=\frac{\det\mb J_t}{\det\mb A}=\frac{\det\mb J_t}{-\Delta+O(t)} .
\ee
At $s=s_j(t)$, $\partial_sH=-\Sigma'(s_j)/\Delta+O(t)$, so $\det\mb J_t=(-\Delta)\,t\,(-\Sigma'(s_j)/\Delta)+O(t^2)=t\,\Sigma'(s_j)+O(t^2)$, and $\det\nabla\mb v=\det\mb J_t\,(1+O(t))=t\,\Sigma'(s_j)+O(t^2)$ too.

\emph{Step 6 (eigenvalues).} At $\mb X_j(t)$, the matrix of $\nabla\mb v$ in the basis $(\mb e_1,\mb e_2,\mb t)$ equals its value on $\mc C$ at $t=0$, namely \eqref{eq:Lstruct}, plus $O(t)$, because $\nabla\mb v$ is smooth and $\mb X_j(t)$ is within $O(t)$ of $\mb X(s_j)$. Its trace is zero (incompressibility). The sum of its principal $2\times2$ minors is $\det\mb K+O(t)=-\Delta+O(t)$, since the other minors contain the zero third column up to $O(t)$ and are $O(t)$. The characteristic equation is $\lambda^3-\mathrm{tr}\,\lambda^2+c_2\lambda-\det=0$ with $c_2=-\Delta+O(t)$, $\mathrm{tr}=0$, $\det=O(t)$. Two roots are therefore $\lambda_{1,2}=\pm\Delta^{1/2}+O(t)$, and the third is small, $\lambda_T\simeq\det/c_2$, whence $\lambda_T=\det\nabla\mb v/(-\Delta)+O(t^2)=-t\,\Sigma'(s_j)/\Delta+O(t^2)$. This is \eqref{eq:eig}.
\end{proof}

\emph{Meaning.} The first-order position of a surviving stagnation point is the position the line would have if it could move, namely $\mb X+t\mb U_n$, with the zero of $\Sigma$ fixing where along the line it sits. The tangential eigenvalue $\lambda_T=-t\,\Sigma'(s_j)/\Delta$ is of order $t$: a point that emerges from a line has a very weak axial strain, which grows linearly in time. This is why the points replacing a line are almost degenerate (since $\lambda_T$ is almost zero) for a while.

\begin{remark}[Types and indices]\label{rem:types}
The \emph{index} (or winding number) of an isolated zero of a vector field is the sign of the determinant of its Jacobian: it records whether the linear map $\gb\delta\mapsto\mb v(\mb x_0+\gb\delta)\approx \mb L\gb\delta$ preserves or reverses orientation near the zero (the standard name is the Poincar\'e index, and for a sum over all zeros inside a region it is a topological invariant). Here, $\det\nabla\mb v=\lambda_1\lambda_2\lambda_T=(-\Delta)\big(-t\,\Sigma'(s_j)/\Delta\big)=t\,\Sigma'(s_j)$, so
\be
\text{Poincar\'e index of }\mb X_j=\mathrm{sign}\big(t\,\Sigma'(s_j)\big),
\ee
for saddle-type ($\Delta>0$) and centre-type ($\Delta<0$) lines alike. On a closed curve, $\Delta$ is continuous and never zero, hence of one sign: all the points have the same normal type. When $\Sigma$ is periodic with simple zeros,  $\Sigma'$ alternates in sign between consecutive zeros: the number of zeros is even, the indices alternate, and the total index is zero.
\end{remark}

\subsection{Persistence by symmetry}

A stagnation line fails to persist because an axial acceleration $a_T$ or an axial shear $\mb m$ appears. If the flow is arranged so that the velocity is everywhere perpendicular to a fixed direction field $\gb\xi$ that is tangent to the curve, neither can appear: there is no axial velocity, hence no axial shear and no axial acceleration. The natural way to arrange this is with a symmetry, and we make it concrete.

Let $\gb\xi(\mb x)=\mb c+\gb\Omega\times\mb x$ with constant vectors $\mb c,\gb\Omega$ be the velocity field of a rigid motion. It generates a one-parameter family of rigid motions $G_\theta$ (that is to say, $\gb\xi$ is a Killing field). The orbits of $G_\theta$ are straight lines parallel to $\mb c$ if $\gb\Omega=\mb 0$ (translations), circles about the axis of $\gb\Omega$ if $\mb c=\mb 0$ (rotations), and helices about that axis if $\mb c\parallel\gb\Omega$, both non-zero (screw motions). A flow \emph{invariant} under $G_\theta$ looks the same after the rigid motion is applied. For a scalar function $f$, this means $\gb\xi\cdot\nabla f=0$.

\begin{corollary}[Persistence by symmetry]\label{cor:sym}
Suppose that for $t$ in an interval the curve $\mc{C}(t)$ is an orbit of $G_\theta$ (so that $\mb t\parallel\gb\xi$ on it), that $\Delta\neq0$ on it, and that $\mb v\cdot\gb\xi=0$ in a neighbourhood of $\mc C(t)$ for all $t$ in the interval, with $\gb\xi$ independent of time. Then $\Sigma\equiv0$ on $\mc C$: the solvability condition holds identically.
\end{corollary}

\begin{proof}
\emph{Step 1 (the axial shear vanishes).} On $\mc C$, $\mb v=\mb 0$ and $\gb\xi=|\gb\xi|\mb t$. Differentiate the identity $\mb v\cdot\gb\xi=0$ in the direction $\mb e_\beta$: $0=\partial_\beta(\mb v\cdot\gb\xi)=(\partial_\beta\mb v)\cdot\gb\xi+\mb v\cdot\partial_\beta\gb\xi=|\gb\xi|\,\mb t\cdot\mb L\mb e_\beta=|\gb\xi|\,m_\beta$, since $\mb v=\mb 0$ on $\mc C$. So $\mb m=\mb 0$.

\emph{Step 2 (the axial acceleration vanishes).} Differentiate $\mb v\cdot\gb\xi=0$ in time at fixed $\mb x$: $\mb a\cdot\gb\xi=0$, so $a_T=0$ on $\mc C$.

\emph{Step 3.} Then $\Sigma=\mb m\cdot\mb K\mb a_n-\Delta a_T=0$.
\end{proof}

Symmetry alone is not enough in general, because the symmetry may allow a velocity component \emph{along} the orbit (axial flow, swirl). It is the absence of this component via the condition $\mb v\cdot\gb\xi=0$ that matters. In potential flow, the absence is automatic (\cref{sec:persym}), in Stokes flow, it must be imposed (\cref{sec:stokes}), and in Navier--Stokes flow, it is not preserved by viscosity for a helix (\cref{sec:NS}).

\begin{remark}[Vorticity and helicity in the symmetric case]\label{rem:helicity}
In the setting of \cref{cor:sym}, with $\gb\xi=\mb e_z+kr\,\mb e_\theta$ and a flow depending on $(r,\chi)$,
$\chi=\theta-kz$, the condition $\mb v\cdot\gb\xi=0$ is $v_z=-krv_\theta$. Then
$\omega_r=\tfrac1r\partial_\chi v_z+k\,\partial_\chi v_\theta=0$ and
$\omega_\theta-kr\,\omega_z=-\partial_r(v_z+krv_\theta)=0$, so $\gb\omega\parallel\gb\xi$. The line is therefore
a vortex line ($\mb m=\mb 0$), and $\mb v\cdot\gb\omega=0$ pointwise: the helicity density vanishes
identically. The limits $k=0$ and $k\to\infty$ are plane flow without axial velocity and axisymmetric flow
without swirl. The conservation law behind persistence is that $w=\mb v\cdot\gb\xi$ satisfies
$\mathrm{D}w/\mathrm{D}t=\nu(\nabla^2w-2k\omega_z)$ (\cref{sec:NS}(iii)), so that $w=0$ is preserved for
$\nu=0$ (with $\gb\xi\cdot\nabla p =\gb\xi\cdot\mb f=0)$. The condition $\mb v\cdot\gb\xi=0$ is the helical analogue of ``no swirl'' in axisymmetric flow and of ``no axial flow'' in plane flow. In Dritschel's helical decomposition $\mb u=\mb h\times\nabla\psi+\mb h\,v$~\cite{dritschel1991generalized}, where $\mb h\parallel\gb\xi$ (his pitch is $-k$), it is the case $v=0$, and the Euler equations conserve $v$ materially, which is the $\nu=0$ statement above. This is the opposite extreme to the Beltrami flows ($\gb\omega\parallel\mb u$) that are the main subject of that paper: here $\gb\omega\parallel\gb\xi\perp\mb v$ and the helicity vanishes.
\end{remark}

\begin{remark}[Symmetry is sufficient, not necessary; vortex filaments]\label{rem:filament}
\Cref{cor:sym} gives a sufficient condition for $\Sigma\equiv0$. It is not necessary: if the flow is
steady in a frame moving as a rigid body, then $\mb a=\mb 0$ in that frame, so $\Sigma\equiv0$ and
$\mb U_n=\mb 0$ for any zero curve, whether or not it is an orbit.
A physical example is the centreline of a vortex core. The velocity vanishes there in a frame
moving with the core, and the centreline has $\omega_T\neq0$, so it is of centre type
($\Delta<0$) when the strain on the core is weak, $\ep<|\omega_T|/2$. A single frame makes the whole centreline stagnant only if the centreline
velocity equals one rigid-body velocity $\mb c+\gb\Omega\times\mb X$ at every point. In the local
induction approximation~\cite{hasimoto1972soliton}, the centreline velocity is $c_\Gamma\kappa\mb b$ with $c_\Gamma$ constant.
For $\gb\Omega=\mb 0$, the condition $\kappa\mb b=\mb c/c_\Gamma$ gives $\kappa'=0$ and
$\kappa\tau=0$, so the filament is a circle: the translating vortex ring. In the frame of the ring, this is axisymmetric flow without swirl and \cref{cor:sym} applies; the closed non-orbit filaments below are the cases that it does not cover. For $\gb\Omega\neq0$, it
contains the helix, and closed filaments that are not orbits of any one-parameter group of
rigid motions are known~\cite{kida1981vortex}. In a rotating frame, the vorticity is $\gb\omega-2\gb\Omega$,
so $\omega_T$, $\Delta$ and the type of the line depend on the frame. For the helical vortex filament, the centreline is stagnant in the screw-moving frame, with $\mb a=\mb 0$ and $\Sigma\equiv0$ because the flow is steady there. In general, it is \emph{not} true there that $\mb v\cdot\gb\xi=0$ in a neighbourhood of the line, since the induced flow has swirl and axial components. So \cref{cor:sym} does not apply, whereas the steady-frame argument does (see also \cref{rem:helixvortex}).
\end{remark}

\section{Potential flow}\label{sec:pot}

Consider $\mb v=\nabla\varphi(\mb{x},t)$ with $\nabla^2\varphi=0$ in a fluid domain (this is $\nabla\cdot\mb v=0$ for a gradient field; compressible irrotational flow is treated in \cref{sec:comp}).

\subsection{Local structure of the stagnation curve}

For irrotational flow the velocity-gradient tensor is $\mb L=\mb E:=\nabla\nabla\varphi$, the Hessian of $\varphi$, which is symmetric. Then $\gb\omega=\mb 0$, so by \eqref{eq:vort} $\omega_T=0$ and $\mb m=\mb 0$, and $\mb K=\mb S$ is symmetric.

\emph{Step 1 (eigenstructure).} $\mb E\mb t=\mb 0$ by \eqref{eq:LT}, so $0$ is an eigenvalue with eigenvector $\mb t$. The trace is zero, so the other two eigenvalues are $\pm\ep$, $\ep\ge0$, with orthonormal eigenvectors $\mb e_1$ (eigenvalue $+\ep$, the \emph{outflow} direction) and $\mb e_2$ (eigenvalue $-\ep$, the \emph{inflow} direction), and $\mb E$ is built from its eigenvectors as
\begin{equation} \label{eq:Hat}
    \mb E = \ep\, (\mb e_1\otimes \mb e_1 - \mb e_2 \otimes \mb e_2),\qquad\text{so}\qquad \mb K=\mathrm{diag}(\ep,-\ep),\quad \Delta=\ep^2 .
\end{equation}
 Close to the line the flow is the planar stagnation-point flow $\mb v\simeq\ep(\delta_1\mb e_1-\delta_2\mb e_2)$ in each normal plane: fluid arrives along $\mb e_2$ and leaves along $\mb e_1$.

\emph{Step 2 (orientation of the strain axes).} The pair $(\mb e_1,\mb e_2)$ lies in the normal plane, so it is obtained from the Frenet pair $(\mb n,\mb b)$ by a rotation through an angle $\alpha(s)$:
\begin{equation}
    \mb e_1 = \cos\alpha \,\mb n + \sin\alpha\,\mb b, \qquad \mb e_2 = -\sin\alpha \,\mb n + \cos\alpha\,\mb b .
    \label{eq:rot}
\end{equation}
(Then $\mb e_1\times\mb e_2=(\cos^2\alpha+\sin^2\alpha)\,\mb n\times\mb b=\mb t$, as required.) Differentiating $\mb e_1$ along the curve with \eqref{eq:Frenet},
\begin{align}
\partial_s\mb e_1&=\alpha'(-\sin\alpha\,\mb n+\cos\alpha\,\mb b)+\cos\alpha(-\kappa\mb t+\tau\mb b)+\sin\alpha(-\tau\mb n)\nonumber \\
&=(\alpha'+\tau)\,\mb e_2-\kappa\cos\alpha\,\mb t . \nonumber
\end{align}
So the strain axes turn about $\mb t$, relative to a frame that is parallel-transported along the curve (one that does not rotate about $\mb t$), at the rate $\partial_s\alpha+\tau$ per unit arclength.

Below, $\varphi_{ijk}=\partial_i\partial_j\partial_k\varphi$ denotes third derivatives, which measure how the strain varies as we move off the point. A subscript $T$, $1$ or $2$ denotes contraction with $\mb t$, $\mb e_1$ or $\mb e_2$.

\begin{lemma}\label{lem:local}
At a fixed time, on $\mc{C}$:
\begin{itemize}
\item[(a)] $\varphi$ is constant along $\mc{C}$: $\varphi(\mb{X},t)=\Phi(t)$, and $\varphi_1=\varphi_2=\varphi_T=0$.
\item[(b)] $\mb E\mb t=\mb{0}$ and $\mb E=\ep\,(\mb{e}_1\mb{e}_1-\mb{e}_2\mb{e}_2)$; in particular $\varphi_{TT}=\mb t\cdot\mb E\mb t=0$.
\item[(c)] $\varphi_{ijk}t_jt_k=-\kappa\,(\mb E\mb n)_i$; in particular,
\be
\varphi_{TT1}=-\kappa\,\ep\cos\alpha,\qquad \varphi_{TT2}=-\kappa\,\ep\sin\alpha,\qquad \varphi_{TTT}=0 .
\label{eq:third}
\ee
\item[(d)] For any harmonic function $f$, $f_{11}+f_{22}=\kappa\,\mb n\cdot\nabla f-\partial_s^2\big(f|_{\mc{C}}\big)$.
\end{itemize}
\end{lemma}

\begin{proof}
(a) $\nabla\varphi=\mb v=\mb 0$ on $\mc C$, so $\mathrm{d}\varphi/\mathrm{d}s=\mb t\cdot\nabla\varphi=0$. (b) is Step 1.

(c) Differentiate the identity $\mb E(\mb X(s))\,\mb t(s)=\mb 0$ along the curve. Since $\mathrm{d}E_{ij}/\mathrm{d}s=\varphi_{ijk}t_k$ and $\partial_s\mb t=\kappa\mb n$,
\[
\varphi_{ijk}t_kt_j+\kappa E_{ij}n_j=0 .
\]
With $\mb E\mb n=\ep\big[(\mb e_1\cdot\mb n)\mb e_1-(\mb e_2\cdot\mb n)\mb e_2\big]=\ep(\cos\alpha\,\mb e_1+\sin\alpha\,\mb e_2)$ by \eqref{eq:rot}, taking components along $\mb e_1$, $\mb e_2$, $\mb t$ gives \eqref{eq:third}.

(d) By the chain rule, $\partial_s^2f(\mb X(s))=\partial_s(\mb t\cdot\nabla f)=t_it_jf_{ij}+\kappa\,\mb n\cdot\nabla f=f_{TT}+\kappa\,\mb n\cdot\nabla f$. Since $f$ is harmonic, $f_{TT}+f_{11}+f_{22}=0$, so $f_{11}+f_{22}=-f_{TT}=\kappa\,\mb n\cdot\nabla f-\partial_s^2f$.
\end{proof}

\subsection{The equation of motion and the pressure}

For an inviscid fluid of unit density with a conservative body force, $\partial_t\mb v+(\mb v\cdot\nabla)\mb v=-\nabla p$ (the force potential is absorbed in $p$). For irrotational flow $(\mb v\cdot\nabla)\mb v=\nabla(\tfrac12|\mb v|^2)$ and $\mb v=\nabla\varphi$, so $\nabla\big(\varphi_t+\tfrac12|\nabla\varphi|^2+p\big)=\mb 0$. In a connected fluid domain this gives
\be
p=c(t)-\varphi_t-\tfrac12|\nabla\varphi|^2 ,
\label{eq:bern}
\ee
where $c(t)$ is fixed by initial or boundary conditions. On the stagnation curve $\mb v=\mb 0$, so $p=c(t)-\varphi_t$ there. Moreover $\nabla(\tfrac12|\mb v|^2)=\mb E\mb v=\mb 0$ on $\mc C$, so taking the gradient of \eqref{eq:bern},
\be
\mb a=\nabla\varphi_t=-\nabla p\quad\text{on }\mc C .
\ee
The fluid particle on the curve is accelerated by minus the pressure gradient, as it must be, since it is at rest and feels no convective acceleration.

\begin{proposition}\label{prop:potlaw}
Let $\mc{C}(t)$ be a stagnation curve of a potential flow for $t$ in an interval, with $\ep>0$ everywhere on it. Then, exactly, on $\mc{C}$:
\begin{itemize}
\item[(i)] $\mb E\mb{U}=-\nabla\varphi_t=-\mb a$.
\item[(ii)] \emph{Solvability along the curve.} $\Sigma=-\ep^2\,\partial_s\varphi_t$, so the solvability condition is $\partial_s\varphi_t=0$; equivalently $\varphi_t=\Phi'(t)$ is constant along $\mc{C}$, and the stagnation pressure $p=c(t)-\Phi'(t)$ is uniform along $\mc{C}$.
\item[(iii)] \emph{Normal velocity.} With $\mb{U}=U_1\mb{e}_1+U_2\mb{e}_2+U_T\mb t$,
\be
U_1=-\frac{a_1}{\ep},\qquad U_2=+\frac{a_2}{\ep},\qquad a_i=\mb{e}_i\cdot\nabla\varphi_t ,
\label{eq:law}
\ee
and $U_T$ is arbitrary.
\end{itemize}
In terms of the Frenet frame, the normal and binormal velocities are $U_N=U_1\cos\alpha-U_2\sin\alpha$ and $U_B=U_1\sin\alpha+U_2\cos\alpha$.
\end{proposition}

\begin{proof}
This is \cref{prop:law} with $\mb m=\mb 0$ and $\mb K=\mathrm{diag}(\ep,-\ep)$, so $\Delta=\ep^2$. Then (i) is $\mb a+\mb L\mb U=\mb 0$ with $\mb L=\mb E$ and $\mb a=\nabla\varphi_t$. The solvability function \eqref{eq:Sigma} is $\Sigma=\mb m\cdot\mb K\mb a_n-\Delta a_T=-\ep^2a_T=-\ep^2\,\mb t\cdot\nabla\varphi_t=-\ep^2\partial_s\varphi_t$, which is (ii). The normal velocity is $\mb U_n=-\mb K^{-1}\mb a_n=-\mathrm{diag}(\ep^{-1},-\ep^{-1})(a_1,a_2)^{\mathsf T}=(-a_1/\ep,\,a_2/\ep)$, which is \eqref{eq:law}. 

A second derivation of (ii) uses part (a) of \cref{lem:local}: differentiate $\varphi(\mb{X}(\sigma,t),t)=\Phi(t)$ in $t$ to get $\varphi_t+\nabla\varphi\cdot\partial_t\mb X=\Phi'$, i.e.\ $\varphi_t=\Phi'(t)$ on $\mc C$ because $\nabla\varphi=\mb 0$ there. A function of $t$ alone is uniform along the curve. Finally the Frenet components follow by projecting $\mb U$ on $\mb n$ and $\mb b$ with \eqref{eq:rot}: $U_N=\mb n\cdot\mb U=U_1\cos\alpha-U_2\sin\alpha$, $U_B=\mb b\cdot\mb U=U_1\sin\alpha+U_2\cos\alpha$.
\end{proof}

\emph{Meaning.} 
\begin{itemize}
\item \emph{Why the pressure must be uniform.} The fluid on the line has no strain along it ($\mb E\mb t=\mb 0$) and so nothing to restrain it axially. If the pressure varies along the line, the particles on it are accelerated along the line, $\mb t\cdot\mb a=-\partial_sp\neq0$, acquire an axial velocity $-t\,\partial_sp$ after a time $t$, and the line is destroyed. This is the same statement as $a_T=0$.
\item \emph{Why the line moves as it does.} By $\mb a=-\nabla p$, \eqref{eq:law} says that the line is displaced by the pressure gradient divided by the strain: $U_1=\partial_1p/\ep$ and $U_2=-\partial_2p/\ep$. In the outflow direction, the line moves towards higher pressure, in the inflow direction towards lower pressure. 
\item \emph{Physical reading.} At an instant the fluid velocity near the line is $\ep(\delta_1\mb e_1-\delta_2\mb e_2)+t\,\mb a$. The point where this vanishes is $\delta_1=-ta_1/\ep$, $\delta_2=ta_2/\ep$: it is displaced just far enough that the strain flow cancels the velocity gained from the acceleration.
\end{itemize}

\begin{remark}[Weak strain]
The velocity of the curve scales as $\ep^{-1}$. Where $\ep\to0$ along the line, the saddle degenerates, the first-order description fails, and the stagnation set can bifurcate there (see \cref{rem:weak}).
\end{remark}

\begin{remark}[Frame]
The stagnation set is not Galilean invariant: the zero set of $\mb v$ depends on the frame. Consider a frame whose origin moves with velocity $\dot{\mb x}_0(t)$. The velocity relative to it is $\mb v'(\mb x',t)=\mb v(\mb x'+\mb x_0,t)-\dot{\mb x}_0$, again irrotational, with potential $\varphi'=\varphi-\dot{\mb x}_0\cdot\mb x'$. Its local acceleration at fixed $\mb x'$ is $\mb a'=\mb a+(\dot{\mb x}_0\cdot\nabla)\mb v-\ddot{\mb x}_0$ (combination of local acceleration in the original frame, convective change due to the frame moving through a spatially varying velocity field and inertial acceleration of the frame itself). At an instant when the frame is not moving relative to the original one, $\dot{\mb x}_0=\mb 0$ (i.e., even when the frame is instantaneously at rest relative to the original frame), so $\mb a'=\mb a-\ddot{\mb x}_0$ and the law \eqref{eq:law} applied in the new frame changes the normal velocity by $+\mb E_n^{-1}\ddot{\mb x}_0$, where $\mb E_n^{-1}=\mathrm{diag}(\ep^{-1},-\ep^{-1})$ is the inverse of $\mb E$ on the normal plane. The solvability condition in the new frame is $\partial_s\varphi_t=\mb t\cdot\ddot{\mb x}_0$: a uniform acceleration of the frame along the line is fatal to the line unless it is cancelled by $\partial_s\varphi_t$.
\end{remark}

\subsection{Persistence by symmetry}\label{sec:persym}

If $\varphi(\cdot,t)$ is invariant under $G_\theta$, then $\gb\xi\cdot\nabla\varphi=0$ (this is what invariance of a scalar means), that is $\mb v\cdot\gb\xi=0$ everywhere: the second hypothesis of \cref{cor:sym} is automatic. Hence in potential flow, a stagnation curve that is an orbit of $G_\theta$ satisfies the solvability condition identically. Directly: $\varphi_t$ is also $G_\theta$-invariant, so it is constant on each orbit, and $\partial_s\varphi_t=0$.

The straight line (translations), the circle (rotations) and the helix (screw motions) are the three orbit types; this explains why these are the stagnation curves of the classical stagnation-point flows. They are not the only persisting lines of potential flow (\cref{ex:ellipse}); symmetry is sufficient, not necessary. Since the stagnation helix of a helically symmetric potential flow does not appear to have been described, we present it briefly; more details can be found in Appendix~\ref{app:helix}.

\emph{The law for a stagnation helix.} Let the helix have radius $r_0$ and pitch $2\pi/k$ in a helically symmetric flow $\varphi=F(r,\chi,t)$, $\chi=\theta-kz$, and suppose the strain axes are $\mb{e}_1=\mb{e}_r$, $\mb{e}_2=\hchi=\nabla\chi/|\nabla\chi|$. 

\emph{Step 1.} For a function $\psi(r,\chi)$, $\partial_\theta\psi=\psi_\chi$ and $\partial_z\psi=-k\psi_\chi$, so
\[
\nabla\psi=\psi_r\mb e_r+\psi_\chi\Big(\frac{\mb e_\theta}{r}-k\mb e_z\Big)=\psi_r\mb e_r+\psi_\chi\nabla\chi=\psi_r\mb e_r+\Lambda\psi_\chi\,\hchi,
\]
with $\Lambda^2=1/r^2+k^2=|\nabla\chi|^2 $.
Hence $a_1=F_{tr}$ and $a_2=\Lambda F_{t\chi}$ on the helix.

\emph{Step 2.} By \eqref{eq:law}, $U_1=-F_{tr}/\ep$ and $U_2=\Lambda F_{t\chi}/\ep$.

\emph{Step 3.} $U_1$ is the rate of change of the radius. The phase $\chi$ of the helix changes only by displacements along $\nabla\chi$, i.e.\ along $\hchi$: a displacement $U_2\,\mathrm{d}t$ along $\hchi$ changes $\chi$ by $\nabla\chi\cdot\hchi\,U_2\,\mathrm{d}t=\Lambda U_2\,\mathrm{d}t$, and a displacement along $\mb e_r$ does not change $\chi$ since $\nabla\chi\cdot\mb e_r=0$. So
\be
\dot r_0=-\frac{F_{tr}}{\ep},\qquad \dot\chi_0=\frac{\Lambda^2F_{t\chi}}{\ep},\qquad \Lambda^2=\frac1{r_0^2}+k^2,
\label{eq:helix}
\ee
evaluated on the helix. The pitch cannot change because of the symmetry. If $F$ remains even in $\chi$, then $F_{t\chi}=0$ on $\chi=0$, the phase is fixed, and the helix only breathes.

\subsection{Second-order condition}

The condition $\partial_s\varphi_t=0$ must hold at every instant. So its time derivative along the moving curve must also vanish. This produces a second condition involving $\varphi_{tt}$.

\begin{proposition}[Second-order condition]\label{prop:second}
If $\mc{C}(t)$ is a stagnation line of a potential flow for $t\in[0,T)$, then in addition to $\partial_s\varphi_t=0$,
\be
\partial_s\left[\varphi_{tt}-\frac{a_1^2-a_2^2}{\ep}\right]=0\quad\text{on }\mc{C}(t).
\label{eq:second}
\ee
\end{proposition}

\begin{proof}
By \cref{prop:potlaw}(ii), $\varphi_t(\mb{X}(\sigma,t),t)=\Phi'(t)$ for all $t$. Differentiate in $t$ at fixed $\sigma$: $\varphi_{tt}+\nabla\varphi_t\cdot\mb{U}=\Phi''(t)$. Expand $\nabla\varphi_t\cdot\mb U=a_1U_1+a_2U_2+a_TU_T$. The last term vanishes since $a_T=\partial_s\varphi_t=0$, so $U_T$ drops out. By \eqref{eq:law}, $a_1U_1+a_2U_2=-a_1^2/\ep+a_2^2/\ep=-(a_1^2-a_2^2)/\ep$. So $\varphi_{tt}-(a_1^2-a_2^2)/\ep=\Phi''(t)$ is independent of $s$.
\end{proof}

 The first condition says that $\varphi_t$ is uniform; the second, that the quantity $\varphi_{tt}$, corrected by the effect of the line's own motion through the field $\nabla\varphi_t$, is uniform. The hierarchy continues: persistence on an interval means that all time derivatives $\mathrm{D}^n\varphi/\mathrm{D}t^n$ following the curve are uniform along it, i.e.,
 \begin{equation}
     \partial_s\left[\frac{\mathrm{D}^n\varphi}{\mathrm{D}t^n}\right]=0\quad\text{on }\mc{C}(t) \quad \text{for all } n.
 \end{equation}
 At each order the condition constrains the driving of the flow, not the line.

\subsection{Loss of a line}

\begin{corollary}[Breaking of a potential-flow line]\label{cor:potbreak}
Let $\mc{C}=\mc{C}(0)$ be a closed stagnation curve of a potential flow with $\ep>0$ on it and let $\eta(s)=\varphi_t(\mb{X}(s),0)$, so that $\Sigma=-\ep^2\eta'$. Suppose $\eta$ is not constant and all its critical points $s_1,\dots,s_N$ are nondegenerate. Then $N\geq2$, and for small $t\neq0$ the stagnation set of $\varphi(\cdot,t)$ near $\mc{C}$ consists of exactly $N$ isolated points, to first order in $t$,
\begin{equation}
\begin{split}
\mb{X}_j(t) &= \mb{X}(s_j)+t\left(U_{j,1}\mb e_1+U_{j,2}\mb e_2+c_j\mb t\right)+O(t^2), \\[1ex]
&\quad U_{j,1}=-\frac{a_1}{\ep},\qquad U_{j,2}=\frac{a_2}{\ep},
\end{split}
\label{eq:potpts}
\end{equation}
and the Hessian of $\varphi$ at $\mb{X}_j(t)$ has eigenvalues
\be
\lambda_{1,2}=\pm\ep+O(t),\qquad \lambda_T=t\,\eta''(s_j)+O(t^2).
\label{eq:poteig}
\ee
\end{corollary}

\begin{proof}
This is \cref{prop:break} with $\Sigma=-\ep^2\eta'$, $\mb m=\mb 0$ and $\Delta=\ep^2$: the zeros of $\Sigma$ are the critical points of $\eta$, and at such a point $\Sigma'=-\ep^2\eta''$, so that $\lambda_T=-t\Sigma'/\Delta=t\eta''$. A continuous periodic function that is not constant has at least a maximum and a minimum, so $N\ge2$.

\emph{A direct check of $\lambda_T$ that exhibits the cancellation of curvature.} Write $\psi=\varphi_t$. At $\mb X_j(t)$, the tangential diagonal Hessian entry is, to $O(t)$,
\begin{align}
\varphi_{TT}(\mb X_j,t)&=\underbrace{\varphi_{TT}(\mb X,0)}_{=0}+t\,\psi_{TT}+\gb\delta^*\!\cdot\nabla\varphi_{TT}\nonumber \\ &=t\,\psi_{TT}+t\,(U_1\varphi_{TT1}+U_2\varphi_{TT2}),
\end{align}
by Taylor expansion in time and in the displacement $\gb\delta^*=t(U_1\mb e_1+U_2\mb e_2)$ (the $O(t)$ displacement along the line contributes $t\,c_j\varphi_{TTT}=0$ by \eqref{eq:third}). Using \eqref{eq:third} and \eqref{eq:law}, $U_1\varphi_{TT1}+U_2\varphi_{TT2}=-\kappa\ep(U_1\cos\alpha+U_2\sin\alpha)=\kappa(a_1\cos\alpha-a_2\sin\alpha)=\kappa\,\mb n\cdot\nabla\psi$, because $\mb n=\cos\alpha\,\mb e_1-\sin\alpha\,\mb e_2$ by inverting \eqref{eq:rot}. Since $\psi$ is harmonic, \cref{lem:local}(d) with $f=\psi$ gives $\psi_{TT}=-\psi_{11}-\psi_{22}=\eta''-\kappa\,\mb n\cdot\nabla\psi$. Adding,
\[
\varphi_{TT}(\mb X_j,t)=t\big(\eta''-\kappa\,\mb n\cdot\nabla\psi+\kappa\,\mb n\cdot\nabla\psi\big)=t\,\eta''(s_j)+O(t^2).
\]
The curvature terms cancel. The off-diagonal Hessian entries coupling the tangential direction to the normal plane are $O(t)$ and shift $\lambda_T$ only by $O(t^2)/\ep$, giving \eqref{eq:poteig}.

The cancellation has a simple meaning. Because $\psi=\varphi_t$ is harmonic, the axial curvature of $\psi$ on a bent line is partly fed by its normal-plane curvature, which produces the term $-\kappa\,\mb n\cdot\nabla\psi$ in $\psi_{TT}$. On the other hand, the new stagnation point sits at the displaced position $\mb X+t\,\mb U_n$, and since the line is curved the axial strain $\varphi_{TT}$ is not zero off the line (\cref{eq:third}); the displacement therefore adds exactly $+\kappa\,\mb n\cdot\nabla\psi$. To first order in $t$, the axial strain rate of the newborn point is thus set only by the variation of $\varphi_t$ along the line, $\lambda_T=t\,\eta''=-t\,\partial_s^2p$, and does not depend on the curvature or torsion of the line, on the orientation $\alpha$ of the strain axes, or on $\ep$.
\end{proof}

\begin{remark}[Types and indices]
For $t>0$, at a maximum of $\eta$ ($\eta''<0$) the point has eigenvalues $(\ep,-\ep,-|\lambda_T|)$ and index $-\mathrm{sign}(\lambda_T)=+1$ as a zero of $\nabla\varphi$; at a minimum, $(\ep,-\ep,+|\lambda_T|)$ with index $-1$ (the roles are reversed for $t<0$). Since $\eta$ is periodic, it alternates between maxima and minima around a closed curve and a stagnation \emph{loop} breaks into at least two points, with equal numbers of each type, and the total index remains zero. Generally, a closed stagnation loop breaks into an even number of stagnation points. By \eqref{eq:bern}, $\eta$ is minus the pressure on the line (up to $c(t)$). At a pressure minimum, the fluid on the line converges axially towards the point (negative tangential eigenvalue), and at a pressure maximum it diverges from it.
\end{remark}

\begin{remark}[Open problem]
\Cref{cor:potbreak} describes the generic loss of a line. The converse question is when a stagnation line of \emph{arbitrary} shape can be maintained on a finite time interval by a suitable driving. The conditions of \cref{prop:potlaw}(ii) and \eqref{eq:second} are necessary, and at every order they only involve the driving. They are also sufficient in at least one non-trivial class: \cref{ex:ellipse} gives, for any analytic planar curve $f(x,y,t)=0$ with $\nabla f\neq0$, an exact (locally convergent, and for polynomial $f$ global) harmonic $\varphi$ with a persisting line of that shape and motion, which is not an orbit of any rigid motion. For curves that are not planar, we have not checked this and do not use it.
\end{remark}

\subsection{Examples}\label{sec:examples}

\begin{example}[Straight line in a time-dependent uniform stream]\label{ex:uniform}
Take $\varphi=\tfrac12\ep(t)(x^2-y^2)+W_1(t)x+W_2(t)y$, so $\mb v=(\ep x+W_1,\,-\ep y+W_2,\,0)$. The stagnation line is parallel to the $z$-axis, at $(x,y)=(-W_1/\ep,\,W_2/\ep)$. 

\emph{Step 1.} $\varphi_t=\tfrac12\dot\ep(x^2-y^2)+\dot W_1x+\dot W_2y$, so $\mb a=\nabla\varphi_t=(\dot\ep x+\dot W_1,\,-\dot\ep y+\dot W_2,\,0)$. At the line, $a_1=\dot W_1-\dot\ep W_1/\ep$, $a_2=\dot W_2-\dot\ep W_2/\ep$ and $a_T=0$.

\emph{Step 2.} The solvability condition holds ($a_T=0$; the line is invariant under translation along $z$), and \eqref{eq:law} gives
\[
U_1=-\frac{\dot W_1}{\ep}+\frac{\dot\ep\,W_1}{\ep^2},\qquad U_2=\frac{\dot W_2}{\ep}-\frac{\dot\ep\,W_2}{\ep^2},
\]
which are the time derivatives of $-W_1/\ep$ and $W_2/\ep$, as they should be: the law reproduces the exact motion of the line.

\emph{Step 3 (failure).} Add an axial stream $W_3(t)z$ to $\varphi$, with $W_3(0)=0$, $\dot W_3(0)\neq0$. Then $\partial_s\varphi_t=\dot W_3\neq0$ and the solvability condition fails: $v_z=W_3(t)\neq0$ for $t\neq0$, so there is no stagnation point at all (here $\eta=\dot W_3z$ is linear in $z$ and has no critical point). A uniform axial acceleration pushes the whole line of fluid along $z$; the flow has no axial strain to restrain it.
\end{example}

\begin{example}[Axial modulation: breaking into saddle points]\label{ex:break}
Take $\ep>0$ constant and
\be
\varphi=\tfrac12\ep(x^2-y^2)+\beta(t)\,\sin(qz)\cosh(qx),\quad \beta(0)=0,\ \ \dot \beta(0)\neq0,
\label{eq:exact}
\ee
which is harmonic for all $t$: the three second derivatives $\varphi_{xx}=\ep+\beta q^2\sin(qz)\cosh(qx)$, $\varphi_{yy}=-\ep$ and $\varphi_{zz}=-\beta q^2\sin(qz)\cosh(qx)$ sum to zero.

\emph{Step 1 (the line at $t=0$).} At $t=0$, $\beta=0$ and the $z$-axis is a stagnation line with strain rate $\ep$, and $\varphi_t|_{\rm axis}=\dot \beta\sin(qz)$, which is not constant: the solvability condition fails.

\emph{Step 2 (the exact stagnation set for $\beta\neq0$).} The velocity is
\[
\mb v=\big(\ep x+\beta q\sin(qz)\sinh(qx),\ -\ep y,\ \beta q\cos(qz)\cosh(qx)\big).
\]
$v_y=0$ gives $y=0$. Since $\cosh\ge1$ and $\beta\ne0$, $v_z=0$ gives $\cos qz=0$, so $\sin qz=\pm1\equiv\sigma_n=(-1)^n$ at $z_n=(n+\tfrac12)\pi/q$. Then $v_x=0$ reads $x\,[\ep+\sigma_n\beta q^2\,\sinh(qx)/(qx)]=0$. For $|qx|\le1$, $\sinh(qx)/(qx)\le\sinh1<1.2$, so if $|\beta|q^2<\ep/2$ the bracket is at least $0.4\ep>0$ and $x=0$ is the only solution. Hence the stagnation set near the axis consists of the isolated points
\[
x=y=0,\qquad z_n=\frac{(n+\tfrac12)\pi}{q},\quad n\in\mathbb{Z},
\]
the critical points of $\eta=\dot \beta\sin qz$, as predicted by \cref{cor:potbreak}. Here the result is exact rather than first order in $t$.

\emph{Step 3 (motion and type).} $\mb a=\nabla\varphi_t=\dot\beta\,(q\sin(qz)\sinh(qx),\,0,\,q\cos(qz)\cosh(qx))$ vanishes at the points, so by \eqref{eq:law} the points do not move (by symmetry). The Hessian at $z_n$ is diagonal, since the off-diagonal entries are $\varphi_{xz}=\beta q^2\cos(qz)\sinh(qx)=0$ and the others vanish identically:
\[
\mb E=\operatorname{diag}\big(\ep+\sigma_n\beta q^2,\,-\ep,\,-\sigma_n\beta q^2\big).
\]
The tangential eigenvalue is $-\sigma_n\beta q^2=t\,\eta''(z_n)$ for $\beta=\dot \beta\,t$, in agreement with \eqref{eq:poteig}, since $\eta''=-\dot\beta q^2\sin(qz_n)$. The points alternate between types $(+,-,-)$ and $(+,-,+)$ along the axis, with alternating index $\pm1$.

 On the axis, the pressure is $p=c(t)-\dot\beta\sin(qz)$ at $t=0$, so it oscillates with $z$. The axial velocity $v_z=\beta q\cos(qz)$ points towards the pressure minima and away from the maxima; the stagnation points sit where the axial pressure gradient vanishes, and each has axial strain of the sign dictated by whether it is a minimum (converging) or a maximum (diverging).
\end{example}

\begin{example}[A persisting line that is not an orbit: a pulsating ellipse]\label{ex:ellipse}
Let $f=\alpha(t)x^2+\beta(t)y^2-1$ with $\alpha=a^{-2}$, $\beta=b^{-2}$, and
\[
\varphi=f^2-\tfrac12z^2\,\Delta_2f^2+\tfrac1{24}z^4\,\Delta_2^2f^2,\qquad \Delta_2=\partial_x^2+\partial_y^2 .
\]
\emph{Step 1 (harmonic).} Here $\Delta_2f^2=(12\alpha^2+4\alpha\beta)x^2+(12\beta^2+4\alpha\beta)y^2-4(\alpha+\beta)$ and $\Delta_2^2f^2=24(\alpha^2+\beta^2)+16\alpha\beta$ is constant. Hence $\Delta\varphi=\Delta_2f^2-\Delta_2f^2+\tfrac12z^2(\Delta_2^2f^2-\Delta_2^2f^2)=0$ for all $t$.

\emph{Step 2 (the line).} $\varphi$ is even in $z$ and $\varphi_x=2ff_x+O(z^2)$, $\varphi_y=2ff_y+O(z^2)$, so $\nabla\varphi=\mb 0$ on $f=0$, $z=0$: the ellipse with semi-axes $a(t)$, $b(t)$ is a stagnation line for all $t$. There $\mb E=2\nabla f\nabla f-2|\nabla f|^2\mb e_z\mb e_z$, so $\ep=2|\nabla f|^2=8(\alpha^2x^2+\beta^2y^2)>0$, which varies along the line, and $\mb e_1=\nabla f/|\nabla f|$, $\mb e_2=\mb e_z$.

\emph{Step 3 (check against the law).} Every term of $\varphi_t$ contains $f$ or $z^2$, so $\varphi_t=0$ on the line: $\partial_s\varphi_t=0$ and the pressure is uniform, as \cref{prop:potlaw}(ii) requires. Moreover $\nabla\varphi_t=2f_t\nabla f$ there, so $a_1=2f_t|\nabla f|$, $a_2=0$, and \eqref{eq:law} gives $U_1=-a_1/\ep=-f_t/|\nabla f|$, which is nothing but the normal velocity of the level set $f(x,y,t)=0$. At $(x,y)=(a,0)$ this is $\dot a$.
\end{example}

Near the line, $\varphi\simeq f^2-z^2|\nabla f|^2$, so fluid is drawn in along $z$ towards the plane $z=0$ and expelled in that plane along the normal to the ellipse, outwards on the outer side and towards the centre on the inner side. The line is an elliptical analogue of the circular stagnation line of \eqref{eq:circle}, but it is regular, with no singular axis. Unlike the straight and circular cases, the strain rate is not uniform along the line: $\ep=8(x^2/a^4+y^2/b^4)$ equals $8/a^2$ at the ends of the major axis and $8/b^2$ at the ends of the minor axis, so it is largest where the nested level curves of $f$ are most crowded. The strain distribution along a persisting line is therefore set by the shape of the line and not only by symmetry. The strain never vanishes on the ellipse, so $\mb K^{-1}$ is regular everywhere on it.

\emph{Two flows, one line.} The sign of $\varphi$ is free: $\varphi$ and $-\varphi$ are both exact solutions of the Euler and Navier--Stokes equations, with the same pressure $p=c(t)-\varphi_t-\tfrac12|\nabla\varphi|^2$ up to the sign of $\varphi_t$, and both have the same stagnation ellipse. In the flow $\varphi$ the fluid approaches the plane $z=0$ from above and below and leaves in the plane along the normal to the ellipse: an \emph{elliptically conically inward} flow, fed from far away inside a wedge bounded by the stable surface $S$ of the line. In the flow $-\varphi$ all streamlines are reversed: fluid arrives in the plane from outside the ellipse and, after descending the axis, from inside it, the two streams collide along the ellipse, and both are ejected along $S$, an \emph{elliptically conically outward} flow. For $a=b$ the surface $S$ leaves the ring vertically and tends to the cone $z=\tfrac{\sqrt3}{2}r$; for $a\neq b$ we expect an elliptic cone, but have not computed it. The line theory does not distinguish the two flows: under $\varphi\to-\varphi$ both $a_1$ and $\ep$ change sign, so $U_1=-a_1/\ep$ is unchanged.

\emph{Symmetry and persistence.} For $a=b$ the flow is axisymmetric and has no swirl, $v_\theta=0$, so $\mb v\cdot\gb\xi=0$ with $\gb\xi=r\mb e_\theta$ and the circle persists by \cref{cor:sym}; it differs from the circular line of \eqref{eq:circle} in being regular on the axis. For $a\neq b$ the ellipse is an orbit of no rigid motion, so \cref{cor:sym} does not apply, yet the line persists: $\varphi_t$ vanishes on it identically, so the stagnation pressure is uniform, $\Sigma=0$, and by \eqref{eq:law} the line moves as the level set $f=0$, $U_1=-f_t/|\nabla f|$. For constant $a,b$ the flow is steady and the ellipse is a steady stagnation line; for time-dependent $a(t),b(t)$ it pulsates. Persistence therefore does not require symmetry, in line with \cref{rem:filament}: \cref{cor:sym} is sufficient but not necessary.

\section{Compressible potential flow}\label{sec:comp}

Consider an inviscid, irrotational, homentropic flow, $\mb v=\nabla\varphi$, with pressure $p(\rho)$, specific enthalpy $h(\rho)=\int\mathrm dp/\rho$ and sound speed $c^2=\mathrm dp/\mathrm d\rho$. With the function of time absorbed in $\varphi$, Bernoulli's equation is $\varphi_t+\tfrac12|\nabla\varphi|^2+h(\rho)=0$, and mass conservation is $\rho_t+\nabla\cdot(\rho\nabla\varphi)=0$. The flow is assumed smooth (no shocks; see \cref{sec:outlook}). No condition on the Mach number is imposed; at a stagnation point it is zero.

\begin{lemma}[Dilatation rate at a stagnation point]\label{lem:theta}
Where $\mb v=\mb 0$, $\ \theta:=\nabla^2\varphi=\varphi_{tt}/c^2$.
\end{lemma}
\begin{proof}
\emph{Step 1 (mass).} At $\mb v=\mb 0$, $\nabla\cdot(\rho\mb v)=\rho\,\theta$, so $\rho\theta=-\rho_t$.
\emph{Step 2 (momentum).} Differentiate Bernoulli's equation in $t$ at fixed $\mb x$. The term $\partial_t\tfrac12|\nabla\varphi|^2=\mb v\cdot\nabla\varphi_t$ vanishes at $\mb v=\mb 0$, so $\varphi_{tt}=-h'(\rho)\rho_t=-(c^2/\rho)\rho_t$.
\emph{Step 3.} Eliminate $\rho_t$.
\end{proof}
\emph{Meaning.} This is the full-potential (wave) equation evaluated where its nonlinear terms vanish, and it is exact. For steady flow, $\theta=0$ on a stagnation set. For a flow varying on a time scale $\omega^{-1}$ and a length scale $\ell$, $\theta/\ep=O\big((\omega\ell/c)^2\big)$: the small parameter is a Helmholtz number, not a Mach number.

\begin{proposition}[Compressible potential-flow lines]\label{prop:comp}
Let $\mc C(t)$ be a stagnation curve, $\mb E=\nabla\nabla\varphi$ with normal eigenvalues $\lambda_1,\lambda_2$ and eigenvectors $\mb e_1,\mb e_2$, $\lambda_1+\lambda_2=\theta$, and $\lambda_1\lambda_2\neq0$ on $\mc C$. On $\mc C$, $\mb a=\nabla\varphi_t=-\nabla h=-\nabla p/\rho$, and:
\begin{itemize}
\item[(i)] $\mb E\mb U=-\mb a$, so $U_i=-a_i/\lambda_i$ ($i=1,2$), $U_T$ arbitrary.
\item[(ii)] $\Sigma=\lambda_1\lambda_2\,\partial_s\varphi_t$. Persistence requires $\varphi_t$ uniform along $\mc C$, equivalently $h$, $p$ and $\rho$ uniform.
\item[(iii)] If $\Sigma\neq0$, the line is replaced to first order in $t$ by points at the critical points $s_j$ of $\eta=\varphi_t|_{\mc C}$, at $\mb X(s_j)+t(U_1\mb e_1+U_2\mb e_2+c_j\mb t)+O(t^2)$, with eigenvalues $\lambda_{1,2}+O(t)$ and $\lambda_T=t\,\eta''(s_j)+O(t^2)$. The Poincar\'e index is $\mathrm{sign}\big(t\,\lambda_1\lambda_2\,\eta''(s_j)\big)$.
\item[(iv)] Persistence on an interval implies $\partial_s\big[c^2(\lambda_1+\lambda_2)-a_1^2/\lambda_1-a_2^2/\lambda_2\big]=0$.
\end{itemize}
\end{proposition}
\begin{proof}
Apply \cref{rem:wherenc} with $\mb m=\mb 0$. For (iv), differentiate $\varphi_t(\mb X(\sigma,t),t)=\Phi'(t)$ in $t$ as in \cref{prop:second}: $\varphi_{tt}-a_1^2/\lambda_1-a_2^2/\lambda_2$ is uniform. By \cref{lem:theta}, $\varphi_{tt}=c^2(\lambda_1+\lambda_2)$, and $c=c(\rho)$ is uniform by (ii).
\end{proof}

\begin{remark}[Types]
$\mb E$ is symmetric, so lines are saddles ($\lambda_1\lambda_2<0$) or nodes ($\lambda_1\lambda_2>0$; sources if $\theta>0$, sinks if $\theta<0$), never centres or foci. Node lines need $\theta=\varphi_{tt}/c^2\neq0$, that is, a density changing on the line.
\end{remark}

\begin{remark}[Degenerate lines at ordinary points]\label{rem:fold}
In incompressible potential flow, $\det\mb K=-\ep^2$ vanishes only if $\mb K=\mb 0$ (\cref{rem:potdelta0}). In compressible flow, $\lambda_1\lambda_2=0$ with $\lambda_1+\lambda_2=\theta\neq0$ means that one eigenvalue, say $\lambda_2$, vanishes and $\lambda_1=\theta$, so $\mb K=\theta\,\mb e_1\otimes\mb e_1$ has rank one. This happens at isolated points $s_*$ of a line where $\lambda_2$ changes sign, that is, where the line changes between saddle and node type. There $U_2=-a_2/\lambda_2$ diverges (\cref{rem:weak}). Let $\eta=\varphi_t|_{\mc C}$ and $\lambda_2'=\partial_s\lambda_2(s_*)=\varphi_{T22}$ (differentiate $\mb E\mb e_2=\lambda_2\mb e_2$ along the curve and use $\mb e_2\cdot\mb E\mb e_1=\mb e_2\cdot\mb E\mb t=0$). At the point $\mb X(s_*+\varsigma)+\delta_1\mb e_1+\delta_2\mb e_2$ at time $t$, with $\delta_2=O(t^{1/2})$, $\varsigma=O(t^{1/2})$ and $\delta_1=O(t)$,
\begin{align}
v_T&=t\,\eta'(s_*)+\tfrac12\lambda_2'\,\delta_2^2+O(t^{3/2}),\\
v_2&=t\,a_2+\lambda_2'\,\varsigma\,\delta_2+\tfrac12\varphi_{222}\,\delta_2^2+O(t^{3/2}),
\end{align}
while $v_1=0$ fixes $\delta_1$. Put $\delta_2=t^{1/2}\xi$ and $\varsigma=t^{1/2}\sigma$. If $t\,\eta'(s_*)/\lambda_2'<0$ there are exactly two stagnation points near $\mb X(s_*)$, with
\[
\xi_\pm=\pm\Big(-\frac{2\eta'(s_*)}{\lambda_2'}\Big)^{1/2},\qquad
\sigma_\pm=-\frac{a_2+\tfrac12\varphi_{222}\,\xi_\pm^2}{\lambda_2'\,\xi_\pm},
\]
and there are none for the opposite sign of $t$. This is a saddle--node event of the line. It differs from \cref{prop:rank1,prop:K0}, where the rank of $\mb L$ drops on a whole arc (two solvability conditions) rather than at an isolated point of a line that is otherwise of rank two, so that one condition ($\Sigma=0$) still governs the line away from $s_*$. As a check, take $\varphi=\tfrac12x^2+\tfrac12\mu zy^2+\tfrac16ky^3+t(a_Tz+a_2y)$, for which $\lambda_1=1$, $\lambda_2=\mu z$, $\eta'=a_T$ and $\varphi_{222}=k$. Its exact zeros are $x=0$, $y=\pm(-2a_Tt/\mu)^{1/2}$ and $z=y\,(\mu a_2-ka_T)/(2\mu a_T)$, in agreement. (This model tests the kinematics only. The condition $\varphi_{tt}=c^2\nabla^2\varphi$ fixes $\varphi_{tt}$ on the line and enters the velocity at $O(t^2)$.)
\end{remark}

\begin{example}[Acoustic standing wave: a node line]\label{ex:acoustic}
Let $\varphi=\cos(\alpha x)\cos(\beta y)\cos\omega t$ with $\alpha^2+\beta^2=\omega^2/c^2$ (linear acoustics, constant $c$). Then $\mb v=\mb 0$ on $x=y=0$ for all $t$, $\mb E=-\cos\omega t\,\mathrm{diag}(\alpha^2,\beta^2,0)$, $\theta=-(\omega^2/c^2)\cos\omega t=\varphi_{tt}/c^2$, and $\mb a=\mb 0$ there, so $\Sigma\equiv0$ and $\mb U_n=\mb 0$. The line is a sink for $\cos\omega t>0$ and a source for $\cos\omega t<0$, and is degenerate at $\cos\omega t=0$. The pressure perturbation $-\rho_0\varphi_t$ is uniform along the line, as \cref{prop:comp}(ii) requires.
\end{example}

\begin{example}[Two standing waves: a node line breaks into sinks and saddles]\label{ex:acoustic2}
Work in linear acoustics (constant $c$, $|\mb v|\ll c$); the statements below are exact within this model. Let
\be
\varphi=\cos(\alpha x)\cos(\beta y)\cos\omega t+\delta\cos(\lambda y)\cos(\gamma z)\sin\Omega t,
\label{eq:acoustic2}
\ee
with $\delta\neq0$, $\alpha^2+\beta^2=\omega^2/c^2$ and $\lambda^2+\gamma^2=\Omega^2/c^2$, so that $\varphi_{tt}=c^2\nabla^2\varphi$.

\emph{Step 1 (the line at $t=0$).} At $t=0$ the second term vanishes and $\mb v=\nabla[\cos(\alpha x)\cos(\beta y)]$, which is zero on the $z$-axis. There $\mb E=\mathrm{diag}(-\alpha^2,-\beta^2,0)$: a sink-type node line with $\theta=-\omega^2/c^2$ and $\lambda_1\lambda_2=\alpha^2\beta^2>0$. On the axis $\eta(z)=\varphi_t(0,0,z,0)=\delta\Omega\cos\gamma z$ is not constant, so $\Sigma=\lambda_1\lambda_2\,\eta'=-\alpha^2\beta^2\delta\Omega\gamma\sin\gamma z\not\equiv0$ and the solvability condition of \cref{prop:comp}(ii) fails.

\emph{Step 2 (the exact stagnation set).} The velocity is
\begin{align*}
v_x&=-\alpha\sin\alpha x\cos\beta y\cos\omega t,\\
v_y&=-\beta\cos\alpha x\sin\beta y\cos\omega t-\delta\lambda\sin\lambda y\cos\gamma z\sin\Omega t,\\
v_z&=-\delta\gamma\cos\lambda y\sin\gamma z\sin\Omega t .
\end{align*}
Let $|t|$ be small with $\sin\Omega t\neq0$. Near the axis $\cos\lambda y\neq0$, so $v_z=0$ gives $\sin\gamma z=0$, that is $z_n=n\pi/\gamma$; $v_x=0$ gives $x=0$; and at $x=0$, $z=z_n$ the component $v_y$ vanishes at $y=0$ with $\partial_yv_y=-[\beta^2\cos\omega t+(-1)^n\delta\lambda^2\sin\Omega t]\neq0$, so $y=0$ is the only zero there. Hence near the axis the stagnation set is exactly the isolated points $(0,0,z_n)$, $n\in\mathbb Z$, the critical points of $\eta$, as in \cref{prop:comp}(iii).

\emph{Step 3 (motion and type).} At $t=0$, $\mb a=\delta\Omega(0,-\lambda\sin\lambda y\cos\gamma z,-\gamma\cos\lambda y\sin\gamma z)$ has no normal component on the axis, so $\mb U_n=\mb 0$ and the points do not move. The Hessian at $z_n$ is diagonal,
\begin{align}
\mb E&=\mathrm{diag}\big(-\alpha^2\cos\omega t,\ -\beta^2\cos\omega t-(-1)^n\delta\lambda^2\sin\Omega t,\nonumber \\  &\quad-(-1)^n\delta\gamma^2\sin\Omega t\big),
\end{align}
with trace $-(\omega^2/c^2)\cos\omega t-(-1)^n\delta(\Omega^2/c^2)\sin\Omega t=\varphi_{tt}/c^2$, as \cref{lem:theta} requires. The tangential eigenvalue is $-(-1)^n\delta\gamma^2\sin\Omega t=t\,\eta''(z_n)+O(t^3)$, in agreement with \cref{prop:comp}(iii). The two normal eigenvalues stay negative for small $t$. For $\delta t>0$ the points with $n$ even are sinks $(-,-,-)$ and those with $n$ odd are saddles $(-,-,+)$, with the roles exchanged for $\delta t<0$. The Poincar\'e indices $\mathrm{sign}(t\,\alpha^2\beta^2\eta'')$ alternate. In incompressible potential flow every such point is a saddle, as in \cref{ex:break}; here the sign of $\lambda_1\lambda_2$ allows sinks. As in the incompressible case, a sink sits at a pressure minimum on the line ($\eta=\varphi_t=-p'/\rho_0$ is a maximum there).
\end{example}

\section{Stokes flow}\label{sec:stokes}

\subsection{Setting}

Consider the Stokes equations $\mu\nabla^2\mb v=\nabla p-\mb f$, $\nabla\cdot\mb v=0$, with viscosity $\mu$ and a force density $\mb f$ supported on singularities or outside the fluid domain. There is no time derivative in the momentum equation, so the flow is quasi-steady: $\mb v(\cdot,t)$ is determined at each instant by the data (boundary velocities, forces) at that instant. The operator is linear and independent of time, so differentiating the equations in $t$ at fixed $\mb x$ gives $\mu\nabla^2\mb a=\nabla p_t-\mb f_t$, $\nabla\cdot\mb a=0$: the local acceleration field $\mb a=\partial_t\mb v$ is itself a Stokes flow, driven by the time derivatives of the data (plus a shape-derivative contribution if the boundaries move). \Cref{prop:law,prop:break,cor:sym} apply unchanged, because they used only incompressibility.

Every harmonic potential flow $\mb v=\nabla\varphi$ is a Stokes flow, since $\nabla^2\mb v=\nabla\nabla^2\varphi=\mb 0$ with $p$ constant; in particular, \cref{ex:uniform,ex:break} are exact Stokes flows. Two things are different from \cref{sec:pot}.
\begin{itemize}
\item There is no Bernoulli relation. On the line, $\nabla p=\mu\nabla^2\mb v$ involves second derivatives of $\mb v$ and is unrelated to $\mb a$, which is the acceleration of the fluid particle at the line but not a pressure gradient. The solvability condition is therefore purely kinematic, and the statement ``the pressure is uniform along the line'' has no analogue.
\item $\mb L$ is not symmetric and the vorticity on the line need not vanish.
\end{itemize}

\subsection{Types of stagnation lines and the normal velocity}

By \eqref{eq:Kdecomp}, the normal block is $\mb K=\mb S+\tfrac12\omega_T\mb J$, with eigenvalues $\pm(\ep^2-\tfrac14\omega_T^2)^{1/2}$. If $\ep>|\omega_T|/2$ the line is of saddle type, as in potential flow. If $\ep<|\omega_T|/2$ the line is of centre type: in the normal plane the flow near the line circulates around it with angular frequency $(\tfrac14\omega_T^2-\ep^2)^{1/2}$. This cannot occur in potential flow. 

\emph{The law.} \Cref{prop:law} holds in both cases. Writing out $\mb U_n=-\mb K\mb a_n/\Delta$ in components with $K_{22}=-K_{11}$,
\begin{equation}
\begin{split}
U_1 &= -\frac{K_{11}a_1+K_{12}a_2}{\Delta}, \\[1ex]
U_2 &= -\frac{K_{21}a_1-K_{11}a_2}{\Delta}, \qquad \Delta=\ep^2-\tfrac14\omega_T^2 ,
\end{split}
\label{eq:Vstokes}
\end{equation}
which reduces to \eqref{eq:law} when $\omega_T=0$ and $\mb K=\mathrm{diag}(\ep,-\ep)$, and which fails at $\ep=|\omega_T|/2$, where the line can bifurcate.  The rotation $\tfrac12\omega_T\mb J$ in $\mb K$ turns the displacement of the line away from the strain-only result. In the limit where the rotation dominates ($\ep\ll|\omega_T|/2$, centre type) $\mb K\simeq\tfrac12\omega_T\mb J$, $\Delta\simeq-\omega_T^2/4$, and \eqref{eq:Vstokes} gives $\mb U_n\simeq(2/\omega_T)\mb J\mb a_n$: the line moves at right angles to the acceleration, at speed $2|\mb a_n|/|\omega_T|$, like a vortex core pushed sideways by a force.

\subsection{The solvability condition}

To first order in $t$, the fluid velocity at a point displaced from the line by $t\mb U_n$ is $t(\mb a+\mb L\mb U_n)$, with tangential component $t(a_T+\mb m\cdot\mb U_n)$, where $m_\beta=\partial_\beta v_T$ is the shear of the axial velocity across the line. If the displaced point is to stay stagnant, this must vanish. Multiplying by $-\Delta$ and using $\mb U_n=-\mb K\mb a_n/\Delta$, this is \eqref{eq:solv}:
\be
\Sigma=\mb m\cdot\mb K\mb a_n-\Delta\,a_T=0 .
\ee
By \eqref{eq:vort}, $\mb m\cdot\mb U_n=\mb t\cdot(\gb\omega\times\mb U)$, so the correction to the potential-flow condition $a_T=0$ is governed by the vorticity \emph{normal} to the line, $\gb\omega_n$, acting on the normal motion of the line. The line moves, with velocity $\mb U_n$, across a field in which the axial velocity varies at the rate $\mb m$ per unit normal distance. It therefore ``sees'' a changing axial velocity $\mb m\cdot\mb U_n$ per unit time, and this must be cancelled by the axial acceleration $a_T$ of the fluid. In potential flow $\gb\omega=\mb 0$ and the condition reduces to $\partial_s\varphi_t=0$, the uniform-pressure condition.

\subsection{Persistence by symmetry}

In Stokes flow, invariance of the flow under $G_\theta$ does not by itself imply $\mb v\cdot\gb\xi=0$: the component of the velocity along the orbit (axial velocity, swirl) is constant along an orbit but is in general not zero. An orbit is then a stagnation line only for isolated values of the parameters, and moreover generically not persistent. By \cref{cor:sym}, a $G_\theta$-orbit that is a stagnation line persists if in addition $\mb v\cdot\gb\xi=0$ everywhere. Concretely: a straight line in a two-dimensional flow with no axial component; a circle in an axisymmetric flow without swirl; and a helix in a helically symmetric flow with $v_z+krv_\theta=0$ (here $\gb\xi=\mb e_z+kr\mb e_\theta$ is the screw motion that leaves $\chi=\theta-kz$ invariant). The last condition is restrictive for a chiral flow, so the helix is the least natural of the three in Stokes flow.

\subsection{Loss of a line}

By \cref{prop:break}, to first order in $t$ the stagnation points that replace a line are located at the \emph{zeros} of $\Sigma$, not at the critical points of a function. This has two consequences that are absent from potential flow.
\begin{itemize}
\item[(a)] For a closed stagnation loop, $\Sigma$ need not have any zero. If $\Sigma$ has a single sign on the loop, then for small $t\neq0$ there is no stagnation point near the loop at all, and the loop dissolves without leaving anything behind. In potential flow, $\Sigma=-\ep^2\eta'$ with $\eta$ periodic and $\ep>0$, so $\Sigma$ always has at least two zeros, because the integral of $\eta'$ around the loop vanishes (see \cref{ex:stokesper}).
\item[(b)] When zeros exist, there is an even number of them and the sign of $\Sigma'$ alternates, so the points alternate in index along the line, as in potential flow (\cref{cor:potbreak}). All of them have the same normal type, saddle or centre, since $\Delta$ is continuous and non-zero on the loop and so has one sign. The zeros need not be critical points of any periodic function (\cref{ex:stokesper}).
\end{itemize}
The next example is exact and shows the term $\mb m\cdot\mb K\mb a_n$ of \eqref{eq:solv}.

\begin{example}[Linear Stokes flow: the shear term]\label{ex:linear}
Take $\mb v=\mb L\mb x+\mb W(t)$ with a constant traceless matrix $\mb L$ whose third column vanishes ($\mb L\mb e_z=\mb 0$),
\[
v_1=\ep x+W_1(t),\qquad v_2=-\ep y+W_2(t),\qquad v_3=\gamma x+W_3(t),
\]
with constants $\ep>0$ and $\gamma$.

\emph{Step 1 (it is a Stokes flow).} $\nabla\cdot\mb v=\ep-\ep+0=0$ and $\nabla^2\mb v=\mb 0$, so with $p$ constant the Stokes equations hold.

\emph{Step 2 (the line).} $\mb v=\mb 0$ requires $\ep x=-W_1$, $\ep y=W_2$ and $\gamma x=-W_3$. The first two give a line parallel to the $z$-axis at $(x,y)=(-W_1/\ep,W_2/\ep)$; the third is satisfied on it if and only if
\[
D(t):=W_3-\frac{\gamma}{\ep}W_1=0 ,
\]
and otherwise there is no stagnation point at all.

\emph{Step 3 (comparison with the general law).} Here $\mb t=\mb e_z$, $\mb K=\mathrm{diag}(\ep,-\ep)$, $\Delta=\ep^2$, $\mb m=(\gamma,0)$ (the third row of $\mb L$), $\omega_T=0$, and the vorticity is $(0,-\gamma,0)$: normal to the line, $\omega_2=-\gamma$, consistent with $m_1=-\omega_2$. The acceleration is $\mb a=\dot{\mb W}$, so $\mb U_n=-\mb K^{-1}\mb a_n=(-\dot W_1/\ep,\dot W_2/\ep)$, the time derivative of the position of the line, and
\[
\Sigma=\mb m\cdot\mb K\mb a_n-\Delta a_T=\gamma\ep\dot W_1-\ep^2\dot W_3=-\ep^2\,\frac{\mathrm d D}{\mathrm dt},
\]
in agreement with \eqref{eq:solv}.

The condition is not $\dot W_3=0$ as in potential flow ($\gamma=0$): a stream along the line is allowed if it is compensated by the motion of the line across the shear $\gamma$. If $D(0)=0$ and $\dot D(0)\neq0$, then the line exists at $t=0$ and has dissolved for all $t\neq0$ with $D(t)\neq0$; this is the analogue, for an unbounded straight line, of case (a) above (and it is exact, not only first order in $t$).
\end{example}

\begin{example}[Periodic Stokes line: dissolution, and zeros of $\Sigma$ that are critical points of nothing]\label{ex:stokesper}
Let $\ep>0$, $\gamma,c,\beta\neq0$, $q>0$, and
\[
\mb v=\begin{pmatrix}
    \ep x+ct \\ -\ep y+t\beta\sin(qz)\sinh(qy) \\ \gamma x+t\beta\cos(qz)\cosh(qy)
\end{pmatrix}.
\]
\emph{Step 1 (Stokes flow).} $\mb v=\mb v_0+t\mb a$ with $\mb v_0=(\ep x,-\ep y,\gamma x)$ and
$\mb a=(c,\ \beta\sin qz\sinh qy,\ \beta\cos qz\cosh qy)$. Every component is harmonic, and
$\nabla\cdot\mb v=\ep-\ep+t\beta q(\sin qz\cosh qy-\sin qz\cosh qy)=0$, so with $p$ constant the Stokes equations hold.

\emph{Step 2 (the line and $\Sigma$).} At $t=0$ the $z$-axis is a stagnation line with $\mb K=\mathrm{diag}(\ep,-\ep)$, $\Delta=\ep^2$, $\mb m=(\gamma,0)$,
$\omega_T=0$ and vorticity $(0,-\gamma,0)$. On it $\mb a=(c,0,\beta\cos qz)$, so
\[
\Sigma(z)=\mb m\cdot\mb K\mb a_n-\Delta a_T=-\ep^2\Big(\beta\cos qz-\frac{\gamma c}{\ep}\Big).
\]
\emph{Step 3 (exact stagnation set).} $v_x=0$ gives $x=-ct/\ep$ exactly, which is $t\,U_1=-t\,a_1/\ep$. In $|qy|\le1$,
$v_y=y\,[-\ep+t\beta q\sin qz\,S(qy)]$ with $S(u)=\sinh u/u\le\sinh1<1.2$, so if $|t|\,|\beta|q<\ep/2$ the bracket is non-zero and $y=0$.
Then $v_z=\gamma x+t\beta\cos qz=-t\,\Sigma(z)/\ep^2$ exactly. Hence, near the axis, the stagnation set is
$\{(-ct/\ep,\,0,\,z):\Sigma(z)=0\}$ for $t\neq0$.

\emph{Case $|\gamma c|>\ep|\beta|$.} $\Sigma$ has one sign, there is no stagnation point near the axis, and the line dissolves without leaving anything.

\emph{Case $|\gamma c|<\ep|\beta|$.} $\Sigma$ has two simple zeros per period, $z=\pm z_*+2\pi n/q$ with $z_*=q^{-1}\arccos(\gamma c/\ep\beta)$.
The Jacobian is lower triangular, so the eigenvalues are exactly
$\ep,\ -\ep+t\beta q\sin qz,\ -t\beta q\sin qz=-t\,\Sigma'(z)/\Delta$, in agreement with \eqref{eq:eig}, and the tangential one alternates in sign along the axis.
The mean of $\Sigma$ over a period is $\gamma\ep c\neq0$, so $\Sigma$ is not (a multiple of) the derivative of any periodic function:
these zeros are not the critical points of anything. In the potential-flow limit $\gamma=0$ one recovers $\Sigma=-\ep^2\beta\cos qz$, which has zero mean.
\end{example}

\subsection{Second-order condition}

The condition for persistence beyond first order follows by differentiating $\mb v(\mb X(\sigma,t),t)=\mb 0$ twice in time, and the hierarchy continues as in potential flow. We do not write the second-order condition explicitly: unlike \eqref{eq:second}, it does not reduce to the uniformity of a scalar along the line, because the tangential part of the second derivative contains $\mb m\cdot\dot{\mb U}_n$ with $\dot{\mb U}_n$ determined by the normal part, and the result depends on $\mb K^{-1}$, $\mb m$ and the second derivatives of $\mb v$.

\section{Closure: flows driven by moving singularities}\label{sec:closure}

The law needs $\mb a=\partial_t\mb v$ on the line. This is a property of the whole flow, not of the line: it is fixed by the boundary data and the singularities. The question here is how to evaluate it in a concrete case.

\subsection{Potential flow}

For a flow in free space driven by point sources of strength $q_j(t)$ at positions $\mb{x}_j(t)$,
\be
\varphi=\sum_j q_j(t)\,G(\mb{x}-\mb{x}_j(t)),\qquad G(\mb{r})=-\frac{1}{4\pi|\mb{r}|},
\label{eq:sources}
\ee
where $G$ satisfies $\nabla^2G=\delta(\mb r)$, so that each source emits the volume flux $q_j$.  The time derivative at fixed $\mb x$ has two parts, one from the change of the strength and one from the motion of the source (the argument of $G$ changes at the rate $-\dot{\mb x}_j$):
\be
\varphi_t=\sum_j\Big[\dot q_j\,G(\mb{x}-\mb{x}_j)-q_j\,\dot{\mb{x}}_j\cdot\nabla G(\mb{x}-\mb{x}_j)\Big].
\label{eq:phit}
\ee
Substituting \eqref{eq:phit} into \eqref{eq:law} gives an explicit motion law for the stagnation set in terms of the motion of the singularities, and the solvability condition $\partial_s\varphi_t=0$ becomes a condition on $(\dot q_j,\dot{\mb{x}}_j)$. The same holds for a flow with bodies, with the Neumann data of $\varphi_t$ determined by the body accelerations.

\subsection{Stokes flow}

For Stokes flow driven by point forces $\mb F_j(t)$ at $\mb x_j(t)$ in free space,
\be
\mb v=\sum_j\mb S(\mb x-\mb x_j(t))\cdot\mb F_j(t),\qquad \mb S(\mb r)=\frac1{8\pi\mu}\left(\frac{\mb I}{|\mb r|}+\frac{\mb r\mb r}{|\mb r|^3}\right),
\ee
where $\mb S$ is the Stokeslet (Oseen tensor). Differentiating at fixed $\mb x$, by the same argument as above,
\be
\mb a=\sum_j\Big[\mb S(\mb x-\mb x_j)\cdot\dot{\mb F}_j-(\dot{\mb x}_j\cdot\nabla)\mb S(\mb x-\mb x_j)\cdot\mb F_j\Big],
\ee
where $(\dot{\mb x}_j\cdot\nabla)\mb S$ is the derivative of $\mb S$ with respect to its argument in the direction $\dot{\mb x}_j$. Both $\mb L=\nabla\mb v$ and $\mb a$ are explicit, so $\mb K$, $\mb m$, $\ep$, $\omega_T$ and $\Sigma$ are explicit on the line, the normal velocity is given by \cref{prop:law}(ii), and the solvability condition $\Sigma=0$ becomes a condition on $(\dot{\mb F}_j,\dot{\mb x}_j)$. The same holds for a flow with bodies, with the boundary data of $\partial_t\mb v$ determined by the accelerations of the bodies.

\subsection{Advancing the line, and the strain-rate equation}

When the singular representation is available, the curve can be advanced by $\partial_t\mb X=\mb U_n$ (with $U_T=0$) with $\mb L$ recomputed from the representation at the current position of the curve. 

If a local description of the curve in potential flow in terms of $(\mb X,\ep,\alpha)$ is wanted, the strain rate must be tracked too. Write $\psi=\varphi_t$. Following a point of the curve with $U_T=0$, the strain rate $\ep=\mb e_1\cdot\mb E\mb e_1$ changes at the \emph{exact instantaneous} rate
\be
\frac{\mathrm{D}\ep}{\mathrm{D}t}=\psi_{11}+\varphi_{111}U_1+\varphi_{112}U_2 .
\label{eq:epdot}
\ee
\emph{Derivation.} $\ep=e_{1i}E_{ij}e_{1j}$. Its total rate of change has three parts: the explicit time dependence of $\mb E$, giving $e_{1i}\psi_{ij}e_{1j}=\psi_{11}$; the motion of the point through the gradient of $\mb E$, giving $U_k\varphi_{ijk}e_{1i}e_{1j}=\varphi_{11k}U_k$ (only $U_1,U_2$ appear since $U_T=0$); and the rotation of $\mb e_1$, giving $2\,\dot{\mb e}_1\cdot\mb E\mb e_1=2\ep\,\dot{\mb e}_1\cdot\mb e_1=0$ because $\mb e_1$ is a unit vector. This gives \eqref{eq:epdot}. 

Equation \eqref{eq:epdot} is not closed, but the line fixes more of its right-hand side than the local data
$\varphi_{TT1},\varphi_{TT2}$ of \eqref{eq:third} suggest. The following gives the exact count.

\begin{proposition}[Third-order data on a potential-flow line]\label{prop:third}
Let $\varrho:=\partial_s\alpha+\tau$. On $\mc C$:
\begin{itemize}
\item[(a)] $\varphi_{T11}=\ep'$, $\ \varphi_{T22}=-\ep'$, $\ \varphi_{T12}=2\ep\varrho$, together with \eqref{eq:third}.
\item[(b)] $\varphi_{111}+\varphi_{122}=\kappa\ep\cos\alpha$ and $\varphi_{112}+\varphi_{222}=\kappa\ep\sin\alpha$.
Hence the only third derivatives not fixed by the line are $\varphi_{111}$ and $\varphi_{112}$.
\item[(c)] In tube coordinates $\mb x=\mb X(s)+\rho_1\mb e_1+\rho_2\mb e_2$ the axial velocity is
$v_T=\tfrac12\ep'(\rho_1^2-\rho_2^2)+2\ep\varrho\,\rho_1\rho_2+O(\rho^3)$.
\item[(d)] Following a point with $U_T=0$, with $\psi=\varphi_t$,
\begin{align}
\frac{\mathrm D\ep}{\mathrm Dt}&=\psi_{11}+\varphi_{111}U_1+\varphi_{112}U_2,\\
2\ep\,\Omega&=\psi_{12}+\varphi_{112}U_1+\varphi_{122}U_2,\qquad \Omega:=\mb e_2\cdot\frac{\mathrm D\mb e_1}{\mathrm Dt},
\end{align}
and $\partial_t\alpha=\Omega-\omega_F$ with $\omega_F$ from \cref{prop:geom}. Moreover
$\psi_{11}+\psi_{22}=\kappa\,(a_1\cos\alpha-a_2\sin\alpha)$.
The nonlocal data left undetermined are $\varphi_{111},\varphi_{112},\psi_{11}-\psi_{22},\psi_{12}$.
\end{itemize}
\end{proposition}

\begin{proof}
(a) Differentiate $\mb E=\ep(\mb e_1\mb e_1-\mb e_2\mb e_2)$ along the curve: $\varphi_{ijT}=\partial_sE_{ij}$. From \eqref{eq:rot} and \eqref{eq:Frenet},
$\partial_s\mb e_1=\varrho\,\mb e_2-\kappa\cos\alpha\,\mb t$ and $\partial_s\mb e_2=-\varrho\,\mb e_1+\kappa\sin\alpha\,\mb t$, so
\begin{align}
\partial_s\mb E&=\ep'(\mb e_1\mb e_1-\mb e_2\mb e_2)+2\ep\varrho(\mb e_1\mb e_2+\mb e_2\mb e_1) \nonumber \\ &
-\kappa\ep\big[\cos\alpha(\mb t\mb e_1+\mb e_1\mb t)+\sin\alpha(\mb t\mb e_2+\mb e_2\mb t)\big],
\end{align}
whose components are the stated ones (the last bracket reproduces \eqref{eq:third}).
(b) $\nabla^2\partial_k\varphi=0$ gives $\varphi_{11k}+\varphi_{22k}+\varphi_{TTk}=0$; take $k=1,2$ and use \eqref{eq:third}.
(c) $v_T=\mb t\cdot\nabla\varphi(\mb X+\gb\rho)=\mb t\cdot\mb E\gb\rho+\tfrac12\varphi_{Tjk}\rho_j\rho_k+\dots$ and $\mb t\cdot\mb E=\mb 0$.
(d) The first relation is \eqref{eq:epdot}. For the second, $\mathrm D E_{ij}/\mathrm Dt=\psi_{ij}+\varphi_{ijk}U_k$ and
$\mathrm D\mb e_1/\mathrm Dt=\Omega\mb e_2+(\cdot)\mb t$, $\mathrm D\mb e_2/\mathrm Dt=-\Omega\mb e_1+(\cdot)\mb t$; the $(1,2)$ component of
$\mathrm D\mb E/\mathrm Dt$ is $2\ep\Omega$. Writing $\mb e_1=\cos\alpha\,\mb n+\sin\alpha\,\mb b$ gives $\Omega=\partial_t\alpha+\omega_F$.
The last relation is \cref{lem:local}(d) with $f=\psi$, using $\partial_s^2\psi|_{\mc C}=0$ by \cref{prop:potlaw}(ii).
\end{proof}

Thus the line determines $\ep'$, $\alpha'$, $\kappa$, $\tau$ and the sums in (b), but a closed description of $(\mb X,\ep,\alpha)$ still
needs a closure for four nonlocal functions, for instance from the Biot--Savart-type representation of the singularities.

\section{Discussion}\label{sec:disc}

The kinematic structure $\mb L\mb U=-\mb a$ with $\mb L\mb t=\mb 0$ is common to all flows, with $\mathrm{tr}\,\mb K=0$ for incompressible flows and $\mathrm{tr}\,\mb K=\nabla\cdot\mb v$ otherwise (\cref{rem:wherenc,sec:comp}), and the solvability condition is one scalar condition per point of the line. What the dynamics supplies is the acceleration $\mb a$ on the line: in potential flow it is $-\nabla p$ and the condition becomes uniformity of the pressure along the line, which makes the three symmetric orbits persistent; in Stokes flow it is the kinematic acceleration of the particle at the line, the vorticity enters through $\mb K$ and $\mb m$, and persistence needs the additional symmetry $\mb v\cdot\gb\xi=0$. We now say what changes for the Navier--Stokes equations.

\subsection{Navier--Stokes flow}\label{sec:NS}

\Cref{sec:kin} used only incompressibility, so \cref{prop:law,prop:break,cor:sym} hold unchanged for solutions of
\[
\partial_t\mb v+(\mb v\cdot\nabla)\mb v=-\nabla p+\nu\nabla^2\mb v+\mb f,\qquad\nabla\cdot\mb v=0
\]
(unit density, kinematic viscosity $\nu$). What changes is the origin of $\mb a$. On the line $\mb v=\mb 0$, so the convective term vanishes and the momentum equation gives
\be
\mb a=-\nabla p+\nu\nabla^2\mb v+\mb f\qquad\text{on }\mc C .
\label{eq:NSa}
\ee
Substituting into \eqref{eq:Sigma}, the solvability function is
\be
\Sigma=\Delta\big(\partial_sp-\nu(\nabla^2\mb v)_T-f_T\big)+\mb m\cdot\mb K\big(-\nabla p+\nu\nabla^2\mb v+\mb f\big)_n ,
\label{eq:NSsolv}
\ee
and the normal velocity of the line is $\mb U_n=-\mb K^{-1}\mb a_n$ with $\mb a$ from \eqref{eq:NSa}. Four remarks follow.

\emph{(i) Reductions.} If the flow is irrotational then $\nabla^2\mb v=\mb 0$, $\mb m=\mb 0$ and \eqref{eq:NSsolv} reduces to $\Sigma=\ep^2(\partial_sp-f_T)$, so the solvability condition is $\partial_sp=f_T$, the uniform-pressure condition of \cref{prop:potlaw} (with a conservative force absorbed into $p$). Since every harmonic potential flow also solves the Navier--Stokes equations, \cref{sec:pot} applies to them whenever the boundary conditions admit an irrotational solution. In Stokes flow the momentum equation has no $\partial_t\mb v$, so \eqref{eq:NSa} is not available there and $\mb a$ comes from the time-differentiated data (\cref{sec:stokes}); in Navier--Stokes flow the momentum equation itself supplies $\mb a$, at the price described next. For $\nu=0$, $\mb f=\mb 0$ (Euler), \eqref{eq:NSsolv} reads $\Sigma=\Delta\,\partial_sp-\mb m\cdot\mb K\nabla_np$, so $\Sigma=0$ is
uniformity of $p$ along the line only when the vorticity normal to the line vanishes ($\mb m=\mb 0$).

\emph{(ii) The condition constrains the state, not the driving.} In \cref{sec:pot,sec:stokes} the quantity $\mb a$ is fixed by the data of the problem, so $\Sigma=0$ is a condition one may try to satisfy by choosing the driving. In \eqref{eq:NSsolv} the pressure and $\nabla^2\mb v$ are determined by the instantaneous flow, and the pressure is nonlocal. Taking the divergence of the momentum equation and using $\nabla\cdot\mb v=0$ and $\nabla\cdot[(\mb v\cdot\nabla)\mb v]=\partial_iv_j\partial_jv_i=\mathrm{tr}(\mb L^2)$ gives $\nabla^2p=-\mathrm{tr}(\mb L^2)+\nabla\cdot\mb f$. On the line, $\mb L$ has the structure \eqref{eq:Lstruct}, so $\mb L^2=\begin{pmatrix}\mb K^2&\mb 0\\ \mb m^{\mathsf T}\mb K&0\end{pmatrix}$, $\mathrm{tr}(\mb L^2)=\mathrm{tr}(\mb K^2)=2\Delta$ by \eqref{eq:Ksq}, and
\[
\nabla^2p=-2\Delta+\nabla\cdot\mb f=-2\ep^2+\tfrac12\omega_T^2+\nabla\cdot\mb f\qquad\text{on }\mc C,
\]
which is a local statement about the trace of $\nabla\nabla p$ but does not determine $\partial_sp$ or the normal gradient of $p$. (With $\mb f=\mb 0$ and a saddle-type line, $\Delta>0$, this says $\nabla^2p<0$: the pressure on the line exceeds its average over a small sphere about the point, as at a stagnation point of a steady inviscid flow, while rotation, through $\omega_T^2$, works against this.) In potential flow the uniform-pressure condition was a requirement on the driving; here $\partial_sp\neq0$ is allowed provided it is balanced by the viscous and shear terms in \eqref{eq:NSsolv}. Whether a line persists is therefore a property of the solution, and one expects it to be lost unless a symmetry forces $\Sigma=0$.

\emph{(iii) Persistence by symmetry.} \Cref{cor:sym} applies whenever the solution keeps the symmetry and $\mb v\cdot\gb\xi=0$. For the straight line in a plane flow with no axial component, and for the circle in an axisymmetric flow without swirl, this holds for all time: the axial velocity, respectively the swirl, satisfies a homogeneous equation and stays zero. The helix is different. Take $\gb\xi=\mb e_z+kr\,\mb e_\theta=\mb e_z+k\,\mb e_z\times\mb x$ and a flow with $\gb\xi\cdot\nabla p=\gb\xi\cdot\mb f=0$, and put $w=\mb v\cdot\gb\xi$.

\begin{itemize}
\item[1.] Following a fluid particle, $\mathrm{D}w/\mathrm{D}t=\gb\xi\cdot(\mathrm{D}\mb v/\mathrm{D}t)+\mb v\cdot(\mb v\cdot\nabla)\gb\xi$. The last term is $v_iv_j\partial_j\xi_i=0$ because $\partial_j\xi_i=k\epsilon_{izj}$ is antisymmetric.
\item[2.] The momentum equation gives $\gb\xi\cdot\mathrm{D}\mb v/\mathrm{D}t=\nu\,\gb\xi\cdot\nabla^2\mb v$ (using $\gb\xi\cdot\nabla p=\gb\xi\cdot\mb f=0$ by the definition of invariance).
\item[3.] $\nabla^2(\mb v\cdot\gb\xi)=\gb\xi\cdot\nabla^2\mb v+2\,\partial_jv_i\,\partial_j\xi_i+\mb v\cdot\nabla^2\gb\xi$, and $\nabla^2\gb\xi=\mb 0$ (it is linear in $\mb x$). Moreover $\partial_jv_i\,\partial_j\xi_i=k\,\epsilon_{izj}\partial_jv_i=k\,\omega_z$, with $\omega_z=\mb e_z\cdot\gb\omega$.
\end{itemize}
Combining 1--3,
\[
\partial_tw+\mb v\cdot\nabla w=\nu\big(\nabla^2w-2k\,\omega_z\big).
\]
The source $-2\nu k\omega_z$ vanishes for translations ($k=0$) and for a swirl-free axisymmetric flow ($\omega_z=0$), but not in general for a helical flow. Hence viscosity generates $\mb v\cdot\gb\xi\neq0$: starting from helically symmetric data with $\mb v\cdot\gb\xi=0$ at one instant, one finds on the helix $\mb m=\mb 0$ (by Step 1 of the proof of \cref{cor:sym}), so that $\Sigma=-\Delta a_T$, and, since $\mb v=\mb 0$ and $w\equiv0$ there at that instant,
\[
a_T=\frac{\mb a\cdot\gb\xi}{|\gb\xi|}=\frac{\partial_tw}{|\gb\xi|}=-\frac{2\nu k\,\omega_z}{(1+k^2r_0^2)^{1/2}},\qquad
\Sigma=\frac{2\nu k\,\Delta\,\omega_z}{(1+k^2r_0^2)^{1/2}},
\]
so the helix persists to first order only where $\omega_z=0$. At $\nu=0$ (Euler) it persists. The helix is thus again the least robust of the three orbits, as in Stokes flow, and in Navier--Stokes flow this is a viscous effect.

\emph{(iv) First-order loss of a line needs only the initial field.} For Stokes flow, the acceleration requires differentiating the data in time. For Navier--Stokes flow, \eqref{eq:NSa} gives $\mb a$ from the instantaneous velocity and pressure, so the first-order prediction of \cref{prop:break} (the number, positions and types of the stagnation points replacing a line, or their absence) can be evaluated from the initial field alone, with $p$ obtained from the Poisson equation above. This does not extend to higher order: the second-order condition involves $\partial_t\mb K$, $\partial_t\mb m$ and hence $\nabla\nabla p$ and $\nabla\nabla\nabla\mb v$ on the line, and the hierarchy does not close locally. A closed description of a persisting or evolving line requires solving the coupled problem, in which the line is a feature of the solution rather than a probe of a prescribed flow.

\subsection{An exact Navier--Stokes example: stagnation loops broken by the pressure}\label{sec:nsex}

A shear flow $u(y)$ in a frame in which $u$ vanishes on a plane has $\mb v=\mb 0$ on that whole plane, not on a line. There $\mb L$ is nilpotent (for any choice of $\mb t$ in the plane), so $\varepsilon=|\omega_T|/2$ and $\Delta=0$ and the theory of \cref{sec:kin} is degenerate; the theory for stagnation surfaces will be addressed in a separate work. A stagnation \emph{line} appears only when a cross-flow adds strain in the plane normal to the line; this is the cat's-eye geometry of a Kelvin--Helmholtz or Tollmien--Schlichting wave. A line that is periodic in the axial direction $z$ is then a closed loop. If the cross-flow is modulated along $z$, the incompressibility constraint forces an axial velocity, the pressure on the line varies with $z$, and fluid on the line is accelerated axially. By \cref{prop:law} this destroys the line. The example below is chosen so that every quantity can be computed in closed form.

\begin{example}[Superposition of solenoidal Laplacian eigenmodes with different decay rates: lines that exist only at $t = 0$ and split immediately]\label{ex:NSkol}
Take the Navier--Stokes equations with kinematic viscosity $\nu$ and no body force, in the box that is $2\pi$-periodic in $x$ and $y$ and $2\pi/k$-periodic in $z$ ($k$ a positive integer). Let $A>0$, $\delta\neq0$, and let the velocity at $t=0$ be
\be
\mb v(\mb x,0)=\begin{pmatrix}
    \sin y+\delta k(1-\cos x)\sin kz \\ A\sin x \\ \delta\sin x\cos kz
\end{pmatrix}.
\label{eq:kolv}
\ee
The velocity is linear in $\delta$. Everything below is exact at $t=0$ for every $\delta$ (no expansion in $\delta$); only the statements about $t>0$ are first order in $t$.

\emph{Step 1 (incompressibility and the lines).} $\partial_xv_1+\partial_yv_2+\partial_zv_3=\delta k\sin x\sin kz+0-\delta k\sin x\sin kz=0$. For $\delta=0$, the field is the two-dimensional Childress--Soward flow~\cite{childress1989scalar,rajamanickam2023effective} with stream function $\psi=A\cos x-\cos y$ ($v_1=\partial_y\psi$, $v_2=-\partial_x\psi$), which satisfies $\nabla^2\psi=-\psi$. For $\delta\neq0$, the lines $(x,y)=(0,0)$ and $(x,y)=(0,\pi)$, both parallel to the $z$-axis, are still zeros of \eqref{eq:kolv}, because the modulation terms carry the factors $(1-\cos x)$ and $\sin x$, which vanish at $x=0$. (The flow also has zeros on $x=\pi$; they are not analysed here.) \Cref{fig:ns0} shows the cross-flow.
\begin{figure}[t]
\centering
\includegraphics[width=0.48\textwidth]{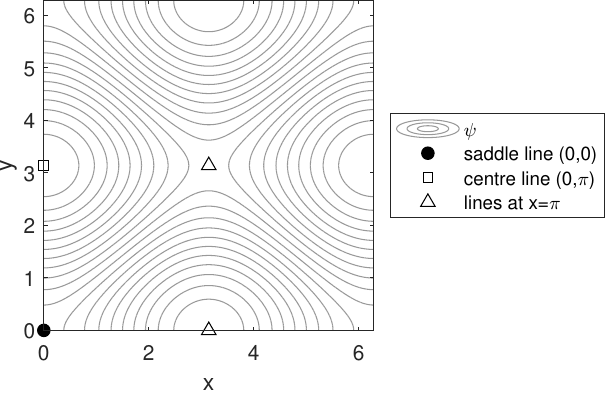}
\caption{Cross-flow of \eqref{eq:kolv} at $\delta=0$: contours of $\psi=A\cos x-\cos y$ in the $(x,y)$ plane. The stagnation lines analysed are the saddle line at $(0,0)$ and the centre line at $(0,\pi)$; the lines at $x=\pi$ are not analysed.}
\label{fig:ns0}
\end{figure}

\emph{Step 2 (normal block and axial shear).} On both lines, $\mb t=\mb e_z$ and, with $\varpi=+1$ on $y=0$ and $\varpi=-1$ on $y=\pi$,
\begin{align}
\mb K=\begin{pmatrix}0&\varpi\\ A&0\end{pmatrix},\quad \mb K^{-1}=\begin{pmatrix}0&1/A\\ \varpi&0\end{pmatrix},\\ \Delta=-\det\mb K=\varpi A,\quad \mb m=(\delta\cos kz,\,0).
\end{align}
So $y=0$ is a saddle line ($\Delta=A>0$) and $y=\pi$ is a centre line ($\Delta=-A<0$): the cross-flow strain moves $\Delta$ off zero and its sign selects the type. The tangential vorticity is $\omega_T=K_{21}-K_{12}=A-\varpi$.

\emph{Step 3 (the pressure, exactly).} Taking the divergence of the momentum equation gives $\nabla^2p=-\mathrm{tr}(\mb L^2)$. Write $\mb L=\mb L_0+\delta\mb L_1$ with $\mb L_0$ the gradient of $(\sin y,A\sin x,0)$ so that $\mathrm{tr}(\mb L_0^2)=2A\cos x\cos y$. $\mb L_0$ has only the entries $(1,2)$ and $(2,1)$, and $\mb L_1$ has zeros there ($v_1$ does not depend on $y$ in the $\delta$-part, and the $\delta$-part of $v_2$ is zero), so $\mathrm{tr}(\mb L_0\mb L_1)=0$ and
\begin{align}
\mathrm{tr}(\mb L^2)&=2A\cos x\cos y+\delta^2\,\mathrm{tr}(\mb L_1^2),\\
\mathrm{tr}(\mb L_1^2)&=2k^2\sin^2x\sin^2kz+2k^2(1-\cos x)\cos x\cos^2kz .
\end{align}
Using the double-angle formulas, $\mathrm{tr}(\mb L_1^2)=k^2(\cos x-\cos2x)+k^2(\cos x-1)\cos2kz$. Each term is a single Fourier mode, so $\nabla^2p=-\mathrm{tr}(\mb L^2)$ is solved mode by mode, with $\nabla^2\to-(n^2+m^2+l^2)$ on $\cos nx\cos my\cos lz$. The result, up to an additive constant, is
\begin{equation}
\begin{split}
p &= A\cos x\cos y+\delta^2k^2\left(\cos x-\tfrac14\cos2x\right) \\
&\quad +\delta^2\left(\frac{k^2\cos x}{1+4k^2}-\frac14\right)\cos2kz .
\end{split}
\label{eq:kolp}
\end{equation}

This is exact at $t=0$ for every $\delta$, since $\mathrm{tr}(\mb L^2)$ is a quadratic polynomial in $\delta$.

\emph{Step 4 (acceleration on the lines and $\Sigma$).} On the lines $\partial_xp=\partial_yp=0$, and
\[
\partial_zp=\frac{\delta^2k\,\sin2kz}{8k^2+2}.
\]
The viscous term of \eqref{eq:NSa} is $\nu\nabla^2\mb v$, with $\nabla^2\mb v_0=-\mb v_0=\mb 0$ on the lines (here $\mb v_0$ is the part of $\mb v$ independent of $\delta$) and $\nabla^2[(1-\cos x)\sin kz]=[\cos x-k^2(1-\cos x)]\sin kz=\sin kz$, $\nabla^2(\sin x\cos kz)=-(1+k^2)\sin x\cos kz=0$ at $x=0$. Hence, on both lines,
\be
\mb a=-\nabla p+\nu\nabla^2\mb v=\Big(\nu\delta k\sin kz,\ 0,\ -\frac{\delta^2k\sin2kz}{8k^2+2}\Big).
\label{eq:kola}
\ee
Since $\mb K\mb a_n=(\varpi a_y,\,Aa_x)$ with $a_y=0$, and $\mb m=(m_1,0)$, we have $\mb m\cdot\mb K\mb a_n=m_1\varpi a_y=0$, and \eqref{eq:solv} (equivalently \eqref{eq:NSsolv}) gives
\begin{equation}
\begin{split}
\Sigma(z) &= -\Delta a_z = \varpi A\,\partial_zp = \frac{\varpi A\,\delta^2k\,\sin2kz}{8k^2+2}, \\[1ex]
\Sigma'(z) &= \frac{\varpi A\,\delta^2k^2\cos2kz}{4k^2+1}.
\end{split}
\label{eq:kolh}
\end{equation}

Its zeros are the same on the saddle and centre lines; $\Sigma$ differs between them only through the factor $\Delta=\varpi A$. It does not depend on $\nu$, and it has zero mean over a period: it is $\Delta$ times the axial pressure gradient on the line, i.e.\ an inviscid effect. It is $O(\delta^2)$ because the $O(\delta)$ part of $\mathrm{tr}(\mb L^2)$ vanishes.

\emph{Step 5 (what replaces the lines, first order in $t$).} The zeros of $\Sigma$ are $z_n=n\pi/(2k)$, $n\in\{0,1,2,3\}$ in each period $2\pi/k$ of the modulation, and they are simple. By \cref{prop:break}, for small $t\neq0$ each line is replaced by four stagnation points per period, at
\begin{align}
&\mb X_n(t)=\left(0,\ y_0-\varpi\nu\delta k\sin(n\pi/2)\,t,\ z_n\right)+t\,c_n\mb e_z+O(t^2),\nonumber \\ &y_0= \begin{cases}
    0 \quad (\varpi=+1),\\
    \pi \quad (\varpi=-1),
\end{cases} \nonumber
\end{align}
where the lateral displacement is $-t\,\mb K^{-1}\mb a_n=(0,-\varpi\nu\delta k\sin kz_n\,t)$ and $c_n$ is a tangential drift that this order does not determine. The eigenvalues of $\nabla\mb v$ at $\mb X_n$ are
\begin{align}
&\lambda_{1,2}=\pm(\varpi A)^{1/2}+O(t),\\ &\lambda_T=-\frac{t\,\Sigma'(z_n)}{\Delta}=-\frac{(-1)^n\,t\,\delta^2k^2}{4k^2+1}+O(t^2).
\end{align}
The sign of $\lambda_T$ alternates with $n$, so the Poincar\'e indices alternate, and by \cref{rem:types} they are $\mathrm{sign}(t\,\Sigma'(z_n))=\varpi(-1)^n$ for $t>0$, that is $(-1)^n$ on the saddle line and $-(-1)^n$ on the centre line. The total index per period is zero.
Viscosity does not break the line: it only enters through $a_x=\nu\delta k\sin kz$, which moves the line sideways in the shear direction $y$ with opposite signs for the saddle and centre lines.

\emph{Step 6 (the unmodulated case).} For $\delta=0$, $\mb v=e^{-\nu t}(\sin y,A\sin x,0)$ solves the Navier--Stokes equations (the nonlinear term is a gradient because $\nabla^2\psi=-\psi$ and $\omega=\psi$, and $\nu\nabla^2\mb v=-\nu\mb v$ only reduces the amplitude). The lines persist for all time and do not move. Here $\mb m=\mb 0$ and $\Sigma\equiv0$, in agreement with \eqref{eq:kolh}.
\end{example}

\begin{remark}[Relation to eddy solutions]\label{rem:walsh}
For $\delta=0$, the stream function $\psi=A\cos x-\cos y$ is a sum of the $(1,0)$ and $(0,1)$
eigenfunctions of the Laplacian with eigenvalue $\lambda=-1$. Step 6 is therefore the case
$\lambda=-1$ of the exact solutions $\mb v=e^{\nu\lambda t}\mb b$, $\nabla^2\mb b=\lambda\mb b$,
$\nabla\cdot\mb b=0$, noted by Taylor~\cite{taylor1923lxxv} and Walsh~\cite{walsh2006eddy}, for which
$\nabla p=-\mb v\cdot\nabla\mb v$. The stream lines do not change in time, so every stagnation point
is frozen, which is why $\Sigma\equiv0$ and $\mb U_n=\mb 0$. For $A=1$,
$\psi=-2\sin\frac{x+y}{2}\sin\frac{x-y}{2}$ is Taylor's cellular vortex in axes rotated by $\pi/4$, with wavenumber $1/\sqrt 2$.

For $\delta\neq0$, \eqref{eq:kolv} is a sum of four solenoidal Laplacian eigenmodes: the shears
$(\sin y,0,0)$, $(0,A\sin x,0)$ and $\delta k(\sin kz,0,0)$, with $\lambda=-1,-1,-k^2$, respectively, and
Taylor's vortex $\delta(-k\cos x\sin kz,0,\sin x\cos kz)$ in the $(x,z)$-plane, with
$\lambda=-(1+k^2)$. Their sum vanishes on $x=0$ at $t=0$, but the last two decay at rates
differing by $\nu$, so on $x=0$ the $x$-velocity is $\delta k\sin kz\,(e^{-\nu k^2t}-e^{-\nu(1+k^2)t})
=\nu\delta k\sin kz\,t+O(t^2)$. This is the viscous part $a_x=\nu\delta k\sin kz$ of \eqref{eq:kola}.
The modes lie in three different eigenspaces, so the nonlinear term is not a gradient, and the
$O(\delta^2)$ pressure~\eqref{eq:kolp} and $\Sigma\neq0$ result from it. The eigenvalues $-1$, $-k^2$ and $-(1+k^2)$ are distinct for $k\ge2$; for $k=1$ (the case computed in the
numerical check) the first three modes share $\lambda=-1$. In either case, at least two eigenspaces occur, so
Walsh's lemma does not apply, the nonlinear term is not a gradient, and the $O(\delta^2)$ pressure
\eqref{eq:kolp} and $\Sigma\neq0$ result from it.
\end{remark}

\emph{Numerical check.} We integrated the full Navier--Stokes equations from \eqref{eq:kolv} with a pseudo-spectral method ($32^3$ Fourier modes, 2/3 dealiasing, fourth-order Runge--Kutta, time step $2.5\times10^{-4}$) for $A=0.8$, $k=1$, $\delta=0.3$, $\nu=0.05$. The pressure of \eqref{eq:kolp} was confirmed at $t=0$ against the spectral solution of $\nabla^2p=-\mathrm{tr}(\mb L^2)$. At the times $t=2.5\times10^{-4},\,10^{-3},\,4\times10^{-3}$, we located, in each plane $z=\mathrm{const}$, the point where $v_x=v_y=0$ near each line (the curve $\mc C^*$ of \cref{sec:loss}) and compared the axial velocity there with the prediction $-t\,\Sigma(z)/\Delta$ of \cref{sec:loss}, \cref{fig:ns1}. We also located the isolated zeros of $\mb v$ near $z_n$ by a three-dimensional Newton iteration and computed the eigenvalues of $\nabla\mb v$ there.
\Cref{tab:ns} shows that the error in $v_z/t$ and in $\lambda_T/(-t\,\Sigma'/\Delta)$ decreases linearly with $t$, as expected from the $O(t^2)$ remainders in \cref{prop:break}, and the lateral displacement of the line agrees with $-t\,\mb K^{-1}\mb a_n$ (\cref{fig:ns2}). The tangential eigenvalue of each point approaches $-t\,\Sigma'(z_n)/\Delta$ as $t\to0$ (\cref{fig:ns3}). The computation was restricted to short times: at $t=O(1)$ and large $\delta$ the first-order law no longer applies, and we make no claim there.

\begin{figure*}[t]
\centering
\includegraphics[width=0.8\textwidth]{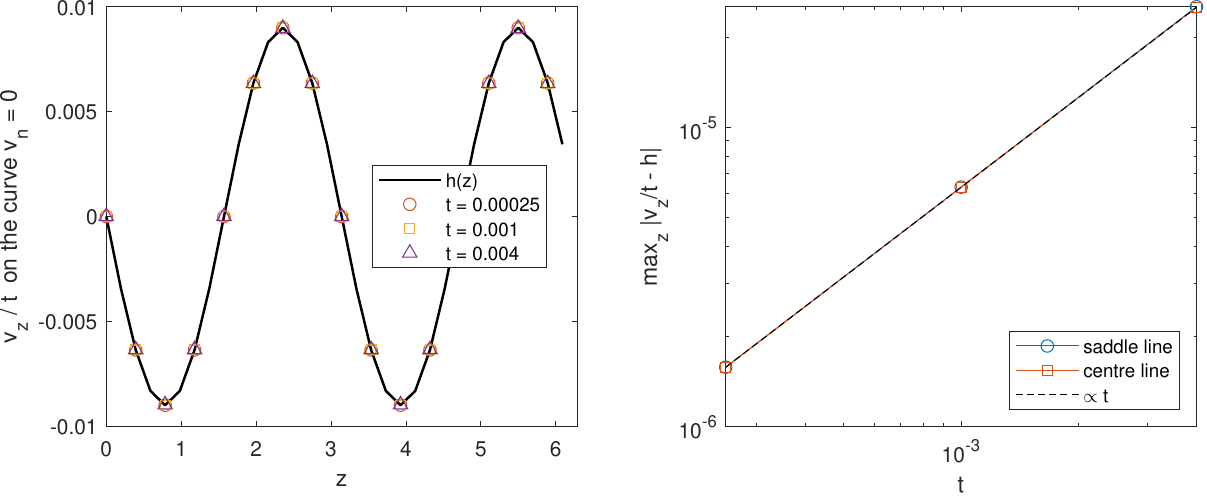}
\caption{(a) Axial velocity $v_z/t$ on the curve where the normal velocity vanishes (symbols, three times) against the prediction $-h(z)=\Sigma(z)/\Delta$ from \eqref{eq:kolh} (line), saddle line. (b) Maximum error in $z$, showing linear decrease with $t$ for both lines.}
\label{fig:ns1}
\end{figure*}

\begin{figure}[t]
\centering
\includegraphics[width=0.4\textwidth]{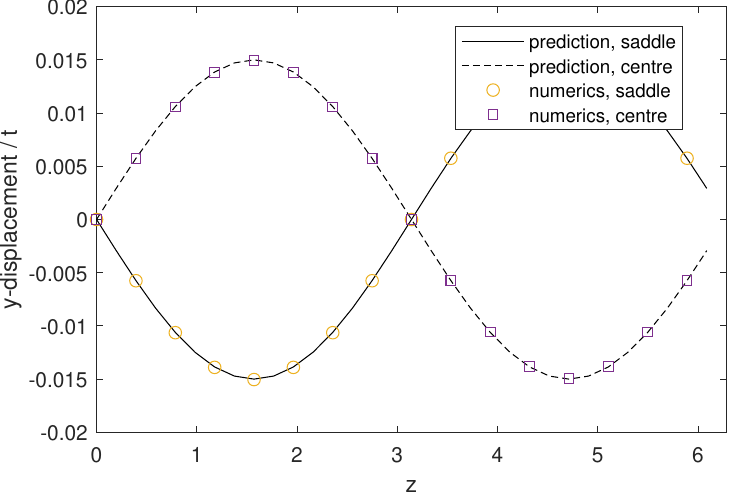}
\caption{Displacement of the curve in the shear direction $y$ divided by $t$ (symbols) against the
prediction $(-\mb K^{-1}\mb a_n)_y=-\varpi\nu\delta k\sin kz$ (lines), at $t=10^{-3}$. The saddle and centre lines move in opposite directions.}
\label{fig:ns2}
\end{figure}

\begin{figure*}[t]
\centering
\includegraphics[width=0.8\textwidth]{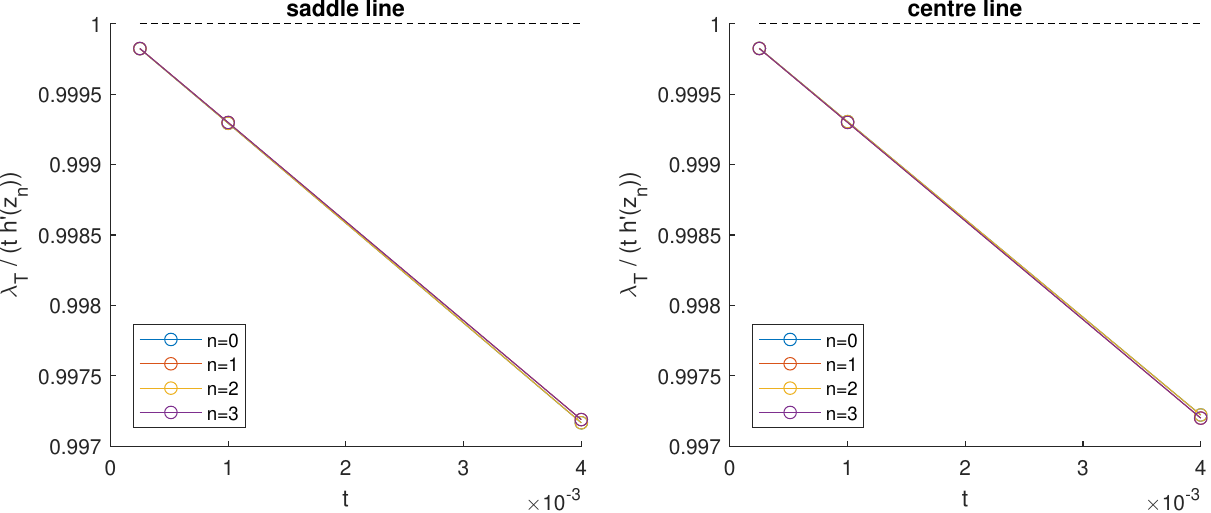}
\caption{Tangential eigenvalue $\lambda_T$ of the isolated stagnation points near $z_n=n\pi/(2k)$, divided by the prediction $-t\,\Sigma'(z_n)/\Delta\equiv th'(z_n)$ of \eqref{eq:eig}, against $t$. The ratio tends to $1$ as $t\to0$ on both the saddle line (a) and the centre line (b).}
\label{fig:ns3}
\end{figure*}

\begin{table}[t]
\centering
\caption{Errors of the first-order prediction for \cref{ex:NSkol}: $E_v=\max_z|v_z/t+\Sigma/\Delta|$ on $\mc C^*$ and $E_\lambda=\max_n|\lambda_T/(-t\,\Sigma'(z_n)/\Delta)-1|$ at the isolated points, for the saddle and centre lines.}
\label{tab:ns}
\begin{tabular}{c|cc|cc}
$t$ & $E_v$ (saddle) & $E_\lambda$ (saddle) & $E_v$ (centre) & $E_\lambda$ (centre)\\ \hline
$2.5\times 10^{-4}$ & $1.6\times 10^{-6}$ & $1.8\times 10^{-4}$ & $1.6\times 10^{-6}$ & $1.8\times 10^{-4}$ \\
$1.0\times 10^{-3}$ & $6.3\times 10^{-6}$ & $7.1\times 10^{-4}$ & $6.3\times 10^{-6}$ & $7.0\times 10^{-4}$ \\
$4.0\times 10^{-3}$ & $2.5\times 10^{-5}$ & $2.8\times 10^{-3}$ & $2.5\times 10^{-5}$ & $2.8\times 10^{-3}$ \\
\end{tabular}
\end{table}

\emph{Meaning.} In \cref{sec:NS}(ii) we said that for the Navier--Stokes equations, the solvability function is a property of the instantaneous solution through the nonlocal pressure. This example shows it concretely: $\Sigma$ is computed from a Poisson equation that can be solved in closed form, its zeros are the same on lines of both types, and it has zero mean, so each loop breaks into an even number of points whose indices alternate, as in potential flow. The two-sided structure, a saddle and a centre line with the same zeros of $\Sigma$ and opposite lateral drift, is what one expects near the cat's eyes of a shear-layer instability.

\section{Stagnation lines with a singular normal block ($\Delta=0$)}\label{sec:delta0}

The motion law of \cref{prop:law} divides by $\Delta=\ep^2-\tfrac14\omega_T^2$. On a smooth line, $\Delta(s)$ is a scalar function of arclength, so a line whose type changes between saddle and centre passes through isolated points where $\Delta=0$, and a shear-layer geometry (\cref{sec:nsex}) has $\Delta\equiv0$ on a whole plane. This section replaces $\mb K^{-1}$ by the structure of $\mb L$ itself. The kinematic theory survives, but the number of solvability conditions is set by the rank of $\mb L$, and fractional powers of $t$ appear when the rank drops. The results are local and kinematic, obtained by the method of \cref{sec:kin}. Two situations make $\Delta=0$ important: lines that change between saddle and centre type, which pass through isolated points with $\Delta=0$, and lines in a shear layer, where the linear part of $\mb v$ vanishes on a whole plane (the case $m_1\equiv0$ below).

Throughout, $(\mb e_1,\mb e_2,\mb t)$ is the right-handed basis of \cref{sec:grad}. Second derivatives of the velocity on the line are denoted by
\[
Z_{\alpha\beta\gamma}:=\mb e_\alpha\cdot(\mb e_\beta\cdot\nabla)(\mb e_\gamma\cdot\nabla)\mb v,\qquad \alpha\in\{1,2,T\},\ \beta,\gamma\in\{1,2\},
\]
and $\dot{\mb a}:=\partial_t\mb a=\partial_t^2\mb v$, $a_{\alpha,\beta}:=\mb e_\alpha\cdot(\mb e_\beta\cdot\nabla)\mb a$, all evaluated on $\mc C$ at $t=0$. For a point displaced by $\gb\delta=\delta_1\mb e_1+\delta_2\mb e_2$ from $\mb X(s)$, \eqref{eq:Taylor} extends to
\be
v_\alpha(\mb X+\gb\delta,t)=t\,a_\alpha+(\mb L\gb\delta)_\alpha+\tfrac12Z_{\alpha\beta\gamma}\delta_\beta\delta_\gamma+O(t^2,\ t|\gb\delta|,\ |\gb\delta|^3).
\label{eq:Bsecond}
\ee

\subsection{The nilpotent normal block}

\begin{lemma}\label{lem:nil}
Let $\Delta=0$ and $\mb K\neq\mb 0$ at a point of $\mc C$. Then there is a right-handed orthonormal basis $(\mb e_1,\mb e_2,\mb t)$ in which
\be
\mb K=\zeta\,\mb e_1\otimes\mb e_2,\qquad \zeta=-\omega_T,\qquad |\zeta|=2\ep>0,
\label{eq:Bnil}
\ee
that is $\mb K\mb e_1=\mb 0$, $\mb K\mb e_2=\zeta\mb e_1$, and
\be
\mb L\mb e_1=m_1\mb t,\qquad \mb L\mb e_2=\zeta\,\mb e_1+m_2\mb t,\qquad \mb L\mb t=\mb 0 .
\label{eq:BL}
\ee
\end{lemma}

\begin{proof}
By \eqref{eq:trK} and \eqref{eq:Ksq}, $\mb K^2=\Delta\mb I=\mb 0$, so the range of $\mb K$ lies in its kernel. Since $\mb K\neq\mb 0$ both are the same line; let $\mb e_1$ span it and let $\mb e_2$ complete a right-handed basis. Then $\mb K\mb e_1=\mb 0$ and $\mb K\mb e_2=\zeta\mb e_1+\beta\mb e_2$ with $\beta=0$ because $\mathrm{tr}\,\mb K=0$. By \eqref{eq:vort}, $\omega_T=K_{21}-K_{12}=-\zeta$, and $\Delta=0$ gives $\ep=|\omega_T|/2$.
\end{proof}

\emph{Meaning.} The normal flow near the line is a simple shear, $\dot\delta_1=\zeta\delta_2$, $\dot\delta_2=0$ (the limit $\Delta\to0$ of \eqref{eq:expK}, where $e^{\mb K\vartheta}=\mb I+\vartheta\mb K$): $\mb e_1$ is the streamwise direction and $\mb e_2$ the direction of the shear gradient. Strain and rotation are exactly balanced. By \eqref{eq:vort}, $m_1=-\omega_2$ and $m_2=\omega_1$, so $m_1=0$ means that the vorticity normal to the line is parallel to $\mb e_1$. If $m_1=0$ the linear velocity field vanishes on the whole plane $\delta_2=0$, which is the shear-layer plane of \cref{sec:nsex}.

\begin{remark}[Potential flow]\label{rem:potdelta0}
In incompressible potential flow $\mb K$ is symmetric and $\mb m=\mb 0$, so $\Delta=\ep^2=0$ forces $\mb K=\mb 0$ and $\mb L=\mb 0$ at that point: only the rank-zero case of \cref{prop:rank} below occurs. The nilpotent cases require vorticity. This is special to incompressibility: in compressible potential flow, $\det\mb K=0$ allows $\mb K=\theta\,\mb e_1\mb e_1$ of rank one at ordinary points of a line, and the pairs of points that result are described in \cref{rem:fold}.
\end{remark}

\subsection{Counting the solvability conditions}

Since $\mb L\mb t=\mb 0$, the range of $\mb L$ is spanned by $\mb L\mb e_1$ and $\mb L\mb e_2$, and $\mb L\mb U=-\mb a$ is solvable if and only if $\mb a\in\mathrm{range}\,\mb L$.

\begin{proposition}[Conditions by rank]\label{prop:rank}
Let $r=\mathrm{rank}\,\mb L\in\{0,1,2\}$. Then $\mb L\mb U=-\mb a$ imposes $3-r$ scalar solvability conditions on $\mb a$, and leaves a $(2-r)$-dimensional family of normal velocities $\mb U_n$ undetermined at first order. By \cref{lem:nil} and \cref{prop:law}, we have \cref{tab:my_table}.
\begin{table*}[t]
\centering
\begin{tabular}{c|l|c|c}
$r$ & where & conditions & free directions in $\mb U_n$\\ \hline
2 & $\Delta\neq0$; or $\Delta=0$, $\mb K\neq\mb0$, $m_1\neq0$ & 1 & 0\\
1 & $\Delta=0$, $\mb K\neq\mb0$, $m_1=0$; or $\mb K=\mb0$, $\mb m\neq\mb0$ & 2 & 1\\
0 & $\mb K=\mb 0$, $\mb m=\mb 0$ & 3 & 2
\end{tabular}
\caption{Number of first-order solvability conditions and free directions of $\mb U_n$, by the rank $r$ of $\mb L$ on the line.}
\label{tab:my_table}
\end{table*}
\end{proposition}

\begin{proof}
$\mb a\in\mathrm{range}\,\mb L$ is $3-r$ conditions, namely $\mb w\cdot\mb a=0$ for a basis of the orthogonal complement of the range. The kernel of $\mb L$ contains $\mb t$ and, restricted to the normal plane, has dimension $2-r$. The table follows from \eqref{eq:BL} (rank one if $m_1=0$, since then $\mb L\mb e_1=\mb0$) and, when $\mb K=\mb 0$, from $\mb L\mb w_n=(\mb m\cdot\mb w_n)\mb t$.
\end{proof}

\subsection{Rank-two lines, including isolated points with $\Delta=0$}

For $r=2$ the single condition is $\Sigma=0$, with $\Sigma=\mb a\cdot(\mb L\mb e_1\times\mb L\mb e_2)$ as in \eqref{eq:Sigma}: it vanishes if and only if $\mb a$ lies in the plane spanned by $\mb L\mb e_1$ and $\mb L\mb e_2$. The second form in \eqref{eq:Sigma}, $\Sigma=\mb m\cdot\mb K\mb a_n-\Delta a_T$, used only $\mathrm{tr}\,\mb K=0$ and so holds for every $\Delta$.

\begin{proposition}[Rank-two lines]\label{prop:rank2}
Let $\mc C$ be a stagnation curve with $\mathrm{rank}\,\mb L=2$ everywhere on it, that is $\mb L\mb e_1\times\mb L\mb e_2\neq\mb0$.
\begin{itemize}
\item[(i)] The equation $\mb a+\mb L\mb U=\mb 0$ is solvable if and only if $\Sigma=0$, and then $\mb U_n$ is unique. For $\Delta\neq0$ this is \cref{prop:law}. For $\Delta=0$, $\mb K\neq\mb 0$, $m_1\neq0$,
\be
\Sigma=\zeta\,m_1\,a_2,\quad U_2=-\frac{a_1}{\zeta},\quad U_1=\frac{m_2a_1/\zeta-a_T}{m_1}.
\label{eq:BA}
\ee
\item[(ii)] If $\Sigma$ has only simple zeros $s_1,\dots,s_N$ ($N=0$ allowed), then for small $0<|t|<t_0$ the zero set of $\mb v(\cdot,t)$ near $\mc C$ consists of exactly $N$ isolated points
\be
\mb X_j(t)=\mb X(s_j)+t\big(\mb U_n+c_j\mb t\big)+O(t^2),
\ee
where $\mb U_n$ is evaluated at $s_j$ (where $\Sigma=0$) and $c_j$ is fixed at the next order.
\end{itemize}
\end{proposition}

\begin{proof}
(i) The orthogonal complement of the range is spanned by $\mb L\mb e_1\times\mb L\mb e_2$. In the nilpotent case, \eqref{eq:BL} gives $\mb L\mb e_1\times\mb L\mb e_2=\zeta m_1\mb e_2$. The three components of $\mb L\mb U_n=-\mb a$ then read $\zeta U_2=-a_1$, $0=-a_2$, $m_1U_1+m_2U_2=-a_T$, which give \eqref{eq:BA}.

(ii) Since the $3\times2$ matrix $\mb M=(\mb L\mb e_1,\mb L\mb e_2)$ has rank two, near any point of $\mc C$ two of its rows (say $i,j$) form an invertible $2\times2$ block, and the remaining row is $k$. As in Steps 1--3 of the proof of \cref{prop:break}, the implicit function theorem solves $v_i=v_j=0$ for $\gb\delta^*(s,t)=O(t)$, and the remaining component on this curve is $v_k=t\,(a_k-\mb M_k\cdot\mb M_{ij}^{-1}\mb a_{ij})+O(t^2)=t\,\Sigma/(\mb L\mb e_1\times\mb L\mb e_2)_k+O(t^2)$, where the denominator is $\pm\det\mb M_{ij}\neq0$. Zeros of $v_k$ on this curve are therefore the simple zeros of $\Sigma$, and the argument continues as in \cref{prop:break}.
\end{proof}

\emph{Meaning.} The functional $\Sigma$ is regular where $\Delta=0$, whereas $\mb K^{-1}$, and hence the normal velocity $-\mb K\mb a_n/\Delta$ and the axial velocity $-t\Sigma/\Delta$ on $\mc C^*$, are not. Thus the divergence of \cref{rem:weak} at an isolated point $s_*$ of $\Delta=0$ with $m_1\neq0$ is removable for a persisting line (for which $\Sigma\equiv0$ and $\mb U_n$ is continuous), and for a line that breaks, the points sit at the zeros of $\Sigma$, exactly as in \cref{prop:break} but now through $\Delta=0$. In the nilpotent case, the relevant condition is simply $a_2=0$: the acceleration must have no component along the shear-gradient direction, which $\mb K$ cannot produce.

\begin{remark}[Degenerate points emerging from the line]
For $\Delta=0$, $m_1\neq0$, \eqref{eq:BL} gives $\mb L^2\mb e_2=\zeta m_1\mb t\neq\mb0$ and $\mb L^3=\mb 0$: $\mb L$ is a single $3\times3$ Jordan block. All three eigenvalues of $\nabla\mb v$ therefore vanish at $t=0$, and at the points of \cref{prop:rank2}(ii) the characteristic polynomial is $\lambda^3=O(t)$, so generically the eigenvalues scale as $t^{1/3}$. This contrasts with \eqref{eq:eig}, where two eigenvalues stay $O(1)$.
\end{remark}

\subsection{Rank-one lines: the nilpotent block with $m_1=0$}

Here $\mb L\mb e_1=\mb 0$ and $\mb L\mb e_2=\zeta\mb e_1+m_2\mb t$. By \cref{prop:rank}, two conditions are needed and $U_1$ is free at first order. Define
\be
\Sigma_B:=\zeta\,a_T-m_2\,a_1 ,
\ee
the rank-one analogue of $\Sigma$ (it is the contraction of $\mb a$ with the left null vector $\zeta\mb t-m_2\mb e_1$ of $\mb L$, and has no division by $\zeta$).
At second order the persistence equation is obtained by differentiating $\mb v(\mb X(\sigma,t),t)=\mb0$ twice in $t$ at fixed $\sigma$:
\begin{equation}
\begin{split}
\mb L\dot{\mb U} &= -\mb g, \qquad \mb g = \dot{\mb a}+2(\mb U\cdot\nabla)\mb a+\mb Z[\mb U,\mb U], \\[1ex]
&\quad (\mb Z[\mb U,\mb U])_i = U_jU_k\,\partial_j\partial_kv_i .
\end{split}
\label{eq:Bsecondlaw}
\end{equation}
We take $U_T=0$; the tangential velocity drops out of the solvability conditions, as it must, by $\mb L\mb t=\mb 0$ and its $s$-derivative. With $U_2$ fixed by the first-order equation and $u:=U_1$, the components of $\mb g$ are the quadratics
\be
g_\alpha(u)=Z_{\alpha11}u^2+2\big(a_{\alpha,1}+Z_{\alpha12}U_2\big)u+\dot a_\alpha+2a_{\alpha,2}U_2+Z_{\alpha22}U_2^2 .
\label{eq:Bg}
\ee

\begin{proposition}[Nilpotent block with $m_1=0$]\label{prop:rank1}
Let $\mc C$ be an arc with $\Delta=0$, $\mb K\neq\mb0$ and $m_1\equiv0$.
\begin{itemize}
\item[(i)] \emph{First order.} $\mb L\mb U=-\mb a$ is solvable if and only if $a_2=0$ and $\Sigma_B=0$. Then $U_2=-a_1/\zeta$ and $U_1$ is undetermined.
\item[(ii)] \emph{Loss of the line: pairs of points at distance $\sqrt t$.} Let $Z_{211}\neq0$ and
\be
H_B(s):=\Sigma_B-\frac{a_2}{Z_{211}}\big(\zeta Z_{T11}-m_2Z_{111}\big).
\label{eq:BH}
\ee
Suppose $H_B$ has a simple zero at $s_*$ with $a_2(s_*)\neq0$. Then for $t$ with $t\,a_2/Z_{211}<0$ there are exactly two stagnation points near $\mb X(s_*)$,
\begin{equation}
\begin{split}
\mb X_\pm(t) &= \mb X(s_*)\pm\Big(-\frac{2t\,a_2}{Z_{211}}\Big)^{1/2}\mb e_1 \\[1ex]
&\quad -\frac{t}{\zeta}\Big(a_1-\frac{a_2Z_{111}}{Z_{211}}\Big)\mb e_2+\sigma_\pm t^{1/2}\mb t+O(t),
\end{split}
\label{eq:Bpairs}
\end{equation}
where the constants $\sigma_\pm$ (a tangential shift of the same order $t^{1/2}$) are fixed by third-order data, and there are none for the opposite sign of $t$.
\item[(iii)] \emph{If $a_2\equiv0$ but $\Sigma_B\not\equiv0$.} If $\Sigma_B$ has simple zeros $s_j$, then near each there are points $\mb X(s_j)+t(\xi\mb e_1-(a_1/\zeta)\mb e_2)+O(t^2)$, one for each real simple root $\xi$ of $g_2(\xi)=0$ (two if the discriminant of $g_2$ is positive, none if it is negative).
\item[(iv)] \emph{Persistence.} If the line persists smoothly on a time interval, then $a_2\equiv0$, $\Sigma_B\equiv0$, and its speed $u=U_1$ satisfies
\be
g_2(u)=0,\qquad \zeta\,g_T(u)-m_2\,g_1(u)=0 .
\label{eq:Bpersist}
\ee
The first is a quadratic for $u$ with roots
\begin{equation}
\begin{split}
u_\pm &= \frac{-B\pm\sqrt{B^2-AC}}{A}, \qquad A=Z_{211}, \\[1ex]
B &= a_{2,1}+Z_{212}U_2, \qquad C=\dot a_2+2a_{2,2}U_2+Z_{222}U_2^2,
\end{split}
\label{eq:u_quadratic}
\end{equation}

which are real if and only if $B^2\geq AC$ (when $B^2=AC$ the two branches merge, the line analogue of a fold). The second equation in \eqref{eq:Bpersist} is a compatibility condition on the driving, one more scalar condition per point.
\end{itemize}
\end{proposition}

\begin{proof}
(i) With $m_1=0$ the components of $\mb L\mb U_n=-\mb a$ are $\zeta U_2=-a_1$, $0=-a_2$ and $m_2U_2=-a_T$; eliminating $U_2$ gives $\zeta a_T-m_2a_1=\Sigma_B=0$.

(ii) By \eqref{eq:Bsecond}, with $\gb\delta=(\delta_1,\delta_2)$,
\begin{align*}
v_1&=t\,a_1+\zeta\delta_2+\tfrac12Z_{1\beta\gamma}\delta_\beta\delta_\gamma+\dots,\\
v_2&=t\,a_2+\tfrac12Z_{2\beta\gamma}\delta_\beta\delta_\gamma+\dots,\\
v_T&=t\,a_T+m_2\delta_2+\tfrac12Z_{T\beta\gamma}\delta_\beta\delta_\gamma+\dots .
\end{align*}
There is no linear term in $\delta_1$ in any component. Set $\delta_1=t^{1/2}\xi$, $\delta_2=t\eta$. To leading order the three equations in $(\xi,\eta,s)$ are
\begin{align}
a_2+\tfrac12Z_{211}\xi^2=0,\quad a_1+\zeta\eta+\tfrac12Z_{111}\xi^2=0,\\ a_T+m_2\eta+\tfrac12Z_{T11}\xi^2=0 .
\end{align}
The first gives $\xi^2=-2a_2/Z_{211}$, which is real for the stated sign of $t$. The second gives $\eta$. Substituting both into the third and multiplying by $\zeta$ gives $H_B(s)=0$. The Jacobian with respect to $(\xi,\eta,s)$ is non-singular when $\xi\neq0$ and $H_B'(s_*)\neq0$, so the implicit function theorem in the variable $t^{1/2}$ gives two smooth branches, $\xi=\pm|\xi|$, and \eqref{eq:Bpairs}.

(iii) If $a_2\equiv0$, take $\delta_1=t\xi$. Then $v_2=\tfrac12t^2g_2(\xi)+O(t^3)$ by \eqref{eq:Bg}, while $v_T=t\,\Sigma_B/\zeta+O(t^2)$ on the curve where $v_1=0$, so the points are at the zeros of $\Sigma_B$ and the roots of $g_2$.

(iv) The first-order conditions are (i). Equation \eqref{eq:Bsecondlaw} is solvable only if $\mb w\cdot\mb g=0$ for every $\mb w$ orthogonal to the range of $\mb L$. These $\mb w$ are $\mb e_2$ and $\zeta\mb e_T-m_2\mb e_1$, which give \eqref{eq:Bpersist}.
\end{proof}

\emph{Meaning.} The normal velocity $-\mb K\mb a_n/\Delta$ of \cref{prop:law} diverges as $\Delta\to0$ because the restoring flow weakens. Here it is resolved: a small acceleration across the shear direction cannot be balanced by the linear flow at all (since $a_2$ has no linear counterpart), so the quadratic terms must balance it, which costs a displacement of order $t^{1/2}$ along $\mb e_1$, the direction in which the linear flow vanishes. Because $H_B$ is a function of arclength, the result is isolated points, not a moving line. They exist for one sign of $t$ and not for the other, so stagnation points are born or annihilated in pairs at $t=0$. This is the line analogue of the saddle--node event for an isolated stagnation point. A line that really persists needs, in addition to the two first-order conditions, a real root of the quadratic and the compatibility condition \eqref{eq:Bpersist}.

\subsection{Rank-one lines: $\mb K=\mb 0$, $\mb m\neq\mb 0$}

Orient the basis so that $\mb m=|\mb m|\mb e_1$. Then $\mb L\mb e_1=|\mb m|\mb t$, $\mb L\mb e_2=\mb 0$, and $\mb L=|\mb m|\,\mb t\otimes\mb e_1$. The linear flow is purely an axial shear with no normal component, and vanishes on the plane $\delta_1=0$.

\begin{proposition}[$\mb K=\mb 0$, $\mb m\neq\mb 0$]\label{prop:K0}
\begin{itemize}
\item[(i)] $\mb L\mb U=-\mb a$ is solvable if and only if $a_1=a_2=0$. Then $U_1=-a_T/|\mb m|$ and $U_2$ is undetermined.
\item[(ii)] Let $Z_{122}\neq0$ and $H_C(s):=a_2-a_1Z_{222}/Z_{122}$. If $H_C$ has a simple zero at $s_*$ with $a_1(s_*)\neq0$, then for $t\,a_1/Z_{122}<0$ there are exactly two stagnation points near $\mb X(s_*)$,
\begin{align}
\mb X_\pm(t)&=\mb X(s_*)\pm\Big(-\frac{2t\,a_1}{Z_{122}}\Big)^{1/2}\mb e_2 \nonumber \\ &-\frac{t}{|\mb m|}\Big(a_T-\frac{a_1Z_{T22}}{Z_{122}}\Big)\mb e_1+\sigma_\pm t^{1/2}\mb t+O(t),
\end{align}
and none for the opposite sign of $t$.
\item[(iii)] If the line persists smoothly, then with $u=U_2$ and $U_1=-a_T/|\mb m|$ both quadratics
\[
g_\alpha(u)=Z_{\alpha22}u^2+2\big(a_{\alpha,2}+Z_{\alpha12}U_1\big)u+\dot a_\alpha+2a_{\alpha,1}U_1+Z_{\alpha11}U_1^2,
\]
($\alpha=1,2$) 
must vanish: two equations for one unknown, hence one compatibility condition per point.
\end{itemize}
\end{proposition}

\begin{proof}
(i) $\mb L\mb U_n=|\mb m|U_1\mb t=-\mb a$. (ii) Since $\mb K=\mb 0$, $v_1=t\,a_1+\tfrac12Z_{1\beta\gamma}\delta_\beta\delta_\gamma+\dots$, $v_2=t\,a_2+\tfrac12Z_{2\beta\gamma}\delta_\beta\delta_\gamma+\dots$ and $v_T=t\,a_T+|\mb m|\delta_1+\tfrac12Z_{T\beta\gamma}\delta_\beta\delta_\gamma+\dots$. The last gives $\delta_1=O(t)$, so $\delta_2=O(t^{1/2})$ dominates $v_1,v_2$: $a_1+\tfrac12Z_{122}\delta_2^2/t=0$ and $a_2+\tfrac12Z_{222}\delta_2^2/t=0$. Eliminating $\delta_2^2$ gives $H_C(s)=0$, and the rest follows as in \cref{prop:rank1}(ii). (iii) The left null space of $\mb L$ is $\{\mb w:w_T=0\}$, so $g_1=g_2=0$ in \eqref{eq:Bsecondlaw}.
\end{proof}

\subsection{Rank zero and further degeneracy}

If $\mb K=\mb 0$ and $\mb m=\mb 0$ (in incompressible potential flow, any point with $\ep=0$, \cref{rem:potdelta0}), then $\nabla\mb v=\mb 0$ on the line at that point and all three components of $\mb a$ must vanish for a first-order line. Otherwise the displacement is of order $t^{1/2}$ in the normal plane, set by the quadratic system $t\,a_\alpha+\tfrac12Z_{\alpha\beta\gamma}\delta_\beta\delta_\gamma=0$ together with the axial direction; we do not work this out here. Where the leading quadratic coefficients themselves vanish ($Z_{211}=0$ in \cref{prop:rank1}, or $Z_{122}=0$ in \cref{prop:K0}), one has to go to cubic terms, and the exponents change again ($t^{1/3}$, \dots). The general pattern is that each drop in the rank of $\mb L$ adds one solvability condition and one free direction at first order, and the free direction is fixed at the next order in $t^{1/2}$ by a polynomial equation whose real roots (fold structure) decide whether the line persists, breaks into pairs of points, or disappears.

\section{Limitations and outlook}\label{sec:outlook}

The first-order theory fails where $\Delta\to0$ (\cref{prop:law}), and the motion law is singular there; the local structure at $\Delta=0$ is analysed in \cref{sec:delta0}, but its global bifurcations are not. The converse of \cref{cor:potbreak}, namely which drivings maintain a stagnation line of arbitrary shape, is open.  \Cref{sec:comp} is restricted to smooth, irrotational, homentropic flow. Transonic flows with shocks are excluded. A shock produces entropy and vorticity, so the potential model and Bernoulli's equation fail across it, and the steady flow equation is of mixed elliptic--hyperbolic type, so $\mb a$ on the line is no longer fixed by an instantaneous elliptic problem. The law $\mb L\mb U=-\mb a$ itself still holds for an Euler flow, with $\mb a=-\nabla p/\rho$ on the line, but the potential-flow consequences do not. Viscous compressible flow is not treated. The results are kinematic: they say what must happen given $\mb a$, not which flows realise a given $\mb a$. The vortical flows for which this is most interesting are stagnation lines of centre type, such as a vortex line in a weak strain (for which $\ep<|\omega_T|/2$), and vortex filaments near a stagnation line; there the pressure must be computed from the Poisson equation together with the Biot--Savart law, so that $\Sigma=0$ becomes a statement about $\partial_sp$ for a flow with its own dynamics. A numerical check of \cref{prop:break} is given in \cref{sec:nsex}.

As noted in the introduction, the zero set of a unidirectional flow $\mb v=u(y,z,t)\,\mb e_x$ is a surface, not a line. A systematic treatment of stagnation surfaces, covering the motion law, the loss of a surface, the degenerate case $\mb w=\mb 0$ and an exact example with a numerical check, will be addressed in a forthcoming paper.

\section{Conclusions}\label{sec:concl}

The motion of a stagnation curve in an incompressible flow is governed by $\mb L\mb U=-\mb a$: the normal velocity is $-\mb K^{-1}\mb a_n$ and the single scalar solvability condition $\Sigma=0$ is the only constraint on the driving. Where $\Sigma\neq0$ the line is replaced, to first order in time, by isolated stagnation points at the zeros of $\Sigma$, which in potential flow are the critical points of $\varphi_t$ on the line and in Stokes flow need not exist. Persistence by symmetry is sufficient but not necessary. The Navier--Stokes example shows that a saddle line and a centre line can split into four points per period, with the first-order predictions confirmed numerically. For inviscid, irrotational, homentropic compressible flow the same structure holds with $\mathrm{tr}\,\mb K=\varphi_{tt}/c^2$ on the line: lines may be of node type, persistence needs a uniform enthalpy, and the breaking law $\lambda_T=t\,\eta''$ is unchanged. The analogous problem for stagnation surfaces, which occur in unidirectional shear flows and are not covered by the present line theory, will be addressed in future work.

The results are of use wherever the structure of the velocity field near a stagnation set matters. In Lagrangian terms, fluid near a stagnation line follows the normal-plane flow $\dot{\gb\delta}=\mb K\gb\delta$ (\cref{sec:grad}): a saddle-type line stretches material exponentially at the rate $\Delta^{1/2}$ and organises the material surfaces along which fluid is drawn in and expelled, whereas a centre-type line holds the fluid circulating around it. Whether such a line persists, moves or is replaced by isolated points therefore decides whether these structures survive, and the law $\mb L\mb U=-\mb a$ predicts where to look for them at a later time. The points that replace a line carry a weak axial strain, $\lambda_T=O(t)$ (\cref{sec:loss}), so fluid is stretched only slowly along the former line. The same law can serve as a predictor when stagnation or vortex-core lines are tracked in numerical or experimental data, and the breaking condition ($\Sigma\neq0$) indicates when such a tracked line should be expected to fragment. For flows driven by moving sources or point forces, \cref{sec:closure} gives the motion of the stagnation set explicitly in terms of the motion of the singularities. Finally, idealised planar, radial and helical stagnation-point flows owe their stagnation lines to symmetry; \cref{prop:break} states what replaces such a line, in number, position and type, when the symmetry is broken, and so identifies which conclusions drawn from symmetric models are robust and which are not.

\section*{Author Declarations}
\subsection*{Conflict of Interest}
The author has no conflicts to disclose.

%\subsection*{Author Contributions}
%\textbf{Prabakaran Rajamanickam:} Conceptualization (lead); Formal analysis (lead); Investigation (lead); Methodology (lead); Software (lead); Validation (lead); Visualization (lead); Writing -- original draft (lead); Writing -- review \& editing (lead).

\section*{Data Availability Statement}
The data that support the findings of this study are available from the corresponding author upon reasonable request.

\appendix

\section{Helical stagnation-point flow}\label{app:helix}

This appendix gives the steady helically symmetric potential flow whose stagnation set is a helix. It is the third member of the family of orbit-type stagnation lines of \cref{sec:pot} (straight, circular, helical), and it supplies the notation $(\mb e_r,\hchi)$, $\Lambda$ used in \eqref{eq:helix}. Throughout, $\mb v=\nabla\varphi$ is steady and incompressible, $\nabla^2\varphi=0$, and $\mb E=\nabla\nabla\varphi$.

\subsection{Reference flows: straight and circular stagnation lines}

With Cartesian coordinates $(x,y,z)$, the planar stagnation-point flow with a stagnation line along the $z$-axis is
\[
\varphi=\tfrac12\ep\,(x^2-y^2),\qquad \mb v=\ep\,(x\,\mb e_x-y\,\mb e_y),
\]
which is exact at all distances.

In cylindrical polar coordinates $(r,\theta,z)$, take an axisymmetric flow with a stagnation circle $r=R$, $z=0$ whose saddle has axes $\mb e_r$ and $\mb e_z$, with uniform axial strain $v_z=-\ep z$. 

\emph{Step 1.} Incompressibility $r^{-1}\partial_r(rv_r)+\partial_zv_z=0$ gives $r^{-1}\partial_r(rv_r)=\ep$, so $rv_r=\tfrac12\ep r^2+\text{const}$.

\emph{Step 2.} The condition $v_r(R,z)=0$ fixes the constant, $-\tfrac12\ep R^2$:
\be
\mb v=\ep\Big[\tfrac12\Big(r-\frac{R^2}{r}\Big)\mb e_r-z\,\mb e_z\Big],\qquad \varphi=\tfrac14\ep\big(r^2-2R^2\ln r\big)-\tfrac12\ep z^2 .
\label{eq:circle}
\ee
Since $\partial_zv_r=\partial_rv_z=0$ the azimuthal vorticity $\partial_zv_r-\partial_rv_z$ vanishes and the flow is irrotational. The term $-\ep R^2/(2r)$ in $v_r$ is a line sink on the axis of strength $\pi\ep R^2$ per unit axial length (for $\ep>0$), so the axis is singular; and for $r,|z|\gg R$ the flow tends to the axisymmetric stagnation-point flow $\ep(\tfrac12 r\,\mb e_r-z\,\mb e_z)$ of Homann type, with strain eigenvalues $(\ep/2,\ep/2,-\ep)$.

\subsection{General helically symmetric solution}

A right-handed helix of radius $r_0$ and pitch $2\pi/k$ ($k>0$) is the curve $r=r_0$, $\theta=kz+\text{const}$. A flow with this helical symmetry depends on $\theta$ and $z$ only through $\chi=\theta-kz$, so $\varphi=F(r,\chi)$. 

\emph{Step 1 (Laplace's equation).} In cylindrical coordinates $\nabla^2\varphi=F_{rr}+r^{-1}F_r+r^{-2}\partial_\theta^2F+\partial_z^2F$. Since $\partial_\theta F=F_\chi$ and $\partial_zF=-kF_\chi$, we have $\partial_\theta^2F=F_{\chi\chi}$ and $\partial_z^2F=k^2F_{\chi\chi}$, so
\be
F_{rr}+\frac1rF_r+\Lambda_r^2F_{\chi\chi}=0,\qquad \Lambda_r^2=\frac1{r^2}+k^2 .
\label{eq:helixLap}
\ee
\emph{Step 2 (separation).} We look for solutions even in $\chi$ (reflection symmetry about the helix $\chi=0$). The $\chi$-independent solutions are $d_0\ln r$ and a constant. For $F=f(r)\cos m\chi$, $m\ge1$, \eqref{eq:helixLap} gives $f''+f'/r-m^2(r^{-2}+k^2)f=0$. With $\varsigma=mkr$ and $f(r)=g(\varsigma)$ we have $f'=mk\,g'$, $f''=m^2k^2g''$ and $f'/r=m^2k^2g'/\varsigma$, and the equation becomes $g''+g'/\varsigma-(1+m^2/\varsigma^2)g=0$, the modified Bessel equation of order $m$. Hence
\be
\varphi=d_0\ln r+\sum_{m=1}^\infty\big[c_mI_m(mkr)+d_mK_m(mkr)\big]\cos m\chi .
\ee
\emph{Step 3 (growth condition).} The $m$th Fourier coefficient in $\chi$ is bounded by $\max_\chi|\varphi|$. As $kr\to\infty$, $I_m(mkr)\sim e^{mkr}/\sqrt{2\pi mkr}$ grows exponentially while $K_m(mkr)\sim\sqrt{\pi/(2mkr)}\,e^{-mkr}$ decays. Requiring at most polynomial growth therefore forces every $c_m=0$, and
\be
\varphi=d_0\ln r+\sum_{m=1}^\infty d_m\,K_m(mkr)\cos m\chi .
\label{eq:helixsol}
\ee
The flow decays at large distance, $\mb v\to(d_0/r)\mb e_r$, and is singular on the axis $r=0$: $d_0\ln r$ is a line source of volume flux $2\pi d_0$ per unit axial length (the same term as in \eqref{eq:circle}, where $d_0=-\ep R^2/2$), and since $K_m(\varsigma)\sim\tfrac12(m-1)!\,(2/\varsigma)^m$ as $\varsigma\to0$, the $m$th term behaves like $r^{-m}\cos m\chi$, a helical multipole line singularity. Because $\varphi\equiv0$ when all $d_m$ vanish, \emph{a stagnation helix in a flow of polynomial growth requires singularities on the axis}.

\emph{Step 4 (velocity).} From $\nabla\psi=\psi_r\mb e_r+\psi_\chi(\mb e_\theta/r-k\mb e_z)$ (\cref{sec:persym}),
\be
\begin{aligned}
v_r&=\frac{d_0}{r}+k\sum_{m=1}^\infty m\,d_m\,K_m'(mkr)\cos m\chi,\\
v_\theta&=-\frac1r\sum_{m=1}^\infty m\,d_m\,K_m(mkr)\sin m\chi,\\
v_z&=k\sum_{m=1}^\infty m\,d_m\,K_m(mkr)\sin m\chi,
\end{aligned}
\label{eq:helixu}
\ee
with $K_m'(\varsigma)=\mathrm dK_m/\mathrm d\varsigma$. Note that $v_z+krv_\theta=0$ identically, that is, $\mb v\cdot\gb\xi=0$ for the screw field $\gb\xi=\mb e_z+kr\,\mb e_\theta$, as required by \cref{cor:sym}.

\subsection{Condition for a stagnation helix}

Every term of $v_\theta$ and $v_z$ in \eqref{eq:helixu} contains $\sin m\chi$, so both vanish on $\chi=0$ for every $r$. The curve $r=r_0$, $\chi=0$ is therefore a stagnation helix if and only if $v_r(r_0,0)=0$:
\be
\frac{d_0}{r_0}+k\sum_{m=1}^\infty m\,d_m\,K_m'(mkr_0)=0 .
\label{eq:helixcond}
\ee
Since $K_m'<0$ on $(0,\infty)$, a mode $m\ge1$ contributes with the sign opposite to that of $d_m$. Either $d_0\ne0$, and a single mode with $d_m$ of the sign of $d_0$ can balance the source, or $d_0=0$ and at least two modes with coefficients of opposite sign are needed.

\subsection{Local structure}

Set $\xi=r-r_0$ and $F_0=F(r_0,0)$. At the helix $\nabla\varphi=\mb0$: $F_r(r_0,0)=0$ by \eqref{eq:helixcond}, $F_\chi(r_0,0)=0$ because $F_\chi\propto\sin m\chi$, and $F_{r\chi}(r_0,0)=0$ because $F$ is even in $\chi$. The leading terms are quadratic. Using \eqref{eq:helixLap} at $(r_0,0)$ with $F_r=0$,
\begin{equation}
\begin{split}
F_{\chi\chi}(r_0,0) &= -\sum_{m=1}^\infty m^2d_mK_m(mkr_0), \\
F_{rr}(r_0,0) &= -\Lambda^2F_{\chi\chi}(r_0,0),
\end{split}
\label{eq:helix_relations}
\end{equation}
so the principal strain rate is
\begin{equation}
\begin{split}
\ep &= F_{rr}(r_0,0)=\Lambda^2\sum_{m=1}^\infty m^2\,d_m\,K_m(mkr_0), \\[1ex]
\Lambda &= \Big(\frac1{r_0^2}+k^2\Big)^{1/2}=|\nabla\chi| .
\end{split}
\label{eq:helixeps}
\end{equation}

Near the helix the phase $\chi$ is, to first order, $\Lambda$ times the distance $\eta$ measured along
\be
\hchi=\frac{\nabla\chi}{|\nabla\chi|}=\frac{\mb e_\theta/r_0-k\,\mb e_z}{\Lambda},
\ee
since $\nabla\chi=\Lambda\hchi$. Then, to second order in the distance $\ell=(\xi^2+\eta^2)^{1/2}$ from the helix, with $\chi=\Lambda\eta$ and $F_{\chi\chi}\Lambda^2=-F_{rr}=-\ep$,
\begin{equation}
\begin{split}
\varphi &= F_0+\tfrac12\ep\,(\xi^2-\eta^2)+O(\ell^3), \\[1ex]
\mb v &= \ep\,(\xi\,\mb e_r-\eta\,\hchi)+O(\ell^2).
\end{split}
\label{eq:helixlocal}
\end{equation}

This is a planar saddle in the normal plane: fluid is drawn in along $\hchi$ and expelled along $\mb e_r$ if $\ep>0$ (the reverse if $\ep<0$). The axis singularities enter $\ep$ only through the modal coefficients $d_m$ and the stagnation condition \eqref{eq:helixcond}; the source strength $d_0$ does not appear in \eqref{eq:helixeps} explicitly. The tangential eigenvalue is exactly zero (the velocity vanishes along the whole helix, so $\mb E\mb t=\mb 0$ and, by symmetry of $\mb E$, its tangential row vanishes), and
\[
\mb E=\ep\,(\mb e_r\mb e_r-\hchi\hchi),
\]
with eigenvalues $(\ep,-\ep,0)$, as in \eqref{eq:Hat} with $\mb e_1=\mb e_r$, $\mb e_2=\hchi$.

\subsection{Geometry of the stagnation helix}

Put $b=1/k$. The helix has unit tangent $\mb t=(r_0k\,\mb e_\theta+\mb e_z)/\sqrt{1+r_0^2k^2}$, curvature and torsion
\be
\kappa=\frac{r_0}{r_0^2+b^2}=\frac{r_0k^2}{1+r_0^2k^2},\qquad \tau=\frac{b}{r_0^2+b^2}=\frac{k}{1+r_0^2k^2},
\ee
principal normal $\mb n=-\mb e_r$ (towards the axis) and binormal $\mb b=\mb t\times\mb n$. Evaluating the cross product with $\mb e_\theta\times\mb e_r=-\mb e_z$ and $\mb e_z\times\mb e_r=\mb e_\theta$ gives $\mb b=-(\mb e_\theta-r_0k\,\mb e_z)/\sqrt{1+r_0^2k^2}=-\hchi$, because $\hchi=(\mb e_\theta-r_0k\,\mb e_z)/(r_0\Lambda)$ and $r_0\Lambda=\sqrt{1+r_0^2k^2}$. Both $\mb e_r$ and $\hchi$ are orthogonal to $\mb t$, and $\mb e_r\times\hchi=(\mb e_z/r_0+k\mb e_\theta)/\Lambda=\mb t$, so $(\mb e_r,\hchi,\mb t)$ is the right-handed basis of \cref{sec:kin}. The strain axes are fixed relative to the Frenet frame, hence they rotate relative to a parallel-transported frame at the rate $\tau$ per unit arclength (in the notation of \cref{sec:pot}, $\alpha=\pi$ and $\partial_s\alpha+\tau=\tau$). The expansion \eqref{eq:helixlocal} is valid for $\Lambda\ell\ll1$.

\emph{Limits.} For $k\to0$ the helix tends to a straight line parallel to the axis ($\kappa=r_0k^2\to0$, $\tau=k\to0$) and \eqref{eq:helixlocal} is the planar stagnation-point flow. For $k\to\infty$ (pitch $\to0$), $\kappa\to1/r_0$, $\tau\to0$ and $\hchi\to-\mb e_z$, so the strain axes $(\mb e_r,\hchi)$ tend to the axes $(\mb e_r,\mb e_z)$ of the radial flow \eqref{eq:circle} with $R=r_0$. The helix therefore interpolates between the two geometries with $\kappa$ and $\tau$ as independent parameters. The two flows differ globally: neighbouring turns are $2\pi/k$ apart, so \eqref{eq:helixlocal} holds only for $\ell\ll1/k$, and the linear far field of \eqref{eq:circle} is axisymmetric but not invariant under the screw motion $\theta\to\theta+\alpha$, $z\to z+\alpha/k$, so it cannot occur in \eqref{eq:helixsol}.

\begin{remark}
Terms odd in $\chi$ ($\sin m\chi$ modes) and screw-invariant terms such as a swirl $\Gamma/(2\pi r)\,\mb e_\theta$ or a uniform axial stream are excluded because they give $v_\theta\ne0$ or $v_z\ne0$ on $\chi=0$. The even solution \eqref{eq:helixsol} may have further stagnation curves, for example near $\chi=\pi$, where $\cos m\chi=(-1)^m$.
\end{remark}

\begin{remark}[Relation to helical vortex filaments]\label{rem:helixvortex}
The Bessel-function representation \eqref{eq:helixsol} is the one used for the velocity potential outside a helical vortex filament or tube~\cite{hardin1982velocity,kuibin1998self,fukumoto2005velocity}, where the odd terms $\sin m\chi$ and a circulation (swirl) term are present. The stagnation helix of this appendix is a different object. It is irrotational ($\omega_T=0$), of saddle type ($\Delta=\ep^2>0$), steady in the laboratory frame, and persists by \cref{cor:sym} because $\mb v\cdot\gb\xi=0$ everywhere. It is supported by line singularities on the axis, not by circulation. The centreline of a helical vortex filament, by contrast, is stagnant only in a screw-moving frame, carries $\omega_T\neq0$ and is of centre type when $\ep<|\omega_T|/2$. It persists by the steady-frame argument of \cref{rem:filament}, not by \cref{cor:sym}.
\end{remark}

\begin{example}[Line source and one helical mode]\label{ex:helixone}
Take $d_m=0$ for $m\ge2$. Then \eqref{eq:helixcond} reads $d_0/r_0+kd_1K_1'(kr_0)=0$, so
\[
d_1=-\frac{d_0}{r_0kK_1'(kr_0)},\qquad \ep=-\Lambda^2\,\frac{d_0\,K_1(kr_0)}{r_0k\,K_1'(kr_0)} .
\]
Since $K_1'<0$, $d_1$ and $\ep$ have the sign of $d_0$. Using $-K_1(\varsigma)/K_1'(\varsigma)\simeq\varsigma$ for $\varsigma\ll1$ (from $K_1\simeq1/\varsigma$) and $\to1$ for $\varsigma\gg1$ (from $K_1'\simeq-K_1$), and $\Lambda^2\simeq r_0^{-2}$ or $k^2$ in the two limits,
\[
\ep\simeq\frac{d_0}{r_0^2}\quad(kr_0\ll1),\qquad \ep\simeq\frac{k\,d_0}{r_0}\quad(kr_0\gg1).
\]
For $kr_0\ll1$ this is the planar flow of a line source of strength $d_0$ and a line dipole at distance $r_0$, which has a stagnation point at $r=r_0$, $\theta=0$ with strain $d_0/r_0^2$. (Check: with $K_1(kr)\simeq1/(kr)$, $\varphi\simeq d_0\ln r+(d_1/k)\cos\theta/r$, $v_r=0$ at $\theta=0$ gives $d_1/k=d_0r_0$, and $\varphi_{rr}=-d_0/r_0^2+2d_0/r_0^2=d_0/r_0^2$.)
\end{example}

\begin{remark}[Time dependence]
Let $d_m=d_m(t)$. Then \eqref{eq:helixcond} holds at each $t$ with $r_0=r_0(t)$, and differentiating $v_r(r_0(t),0,t)=0$ gives $\dot r_0=-F_{tr}/F_{rr}=-F_{tr}/\ep$, which is the first relation in \eqref{eq:helix}. The helix is an orbit of the screw motion, so by \cref{cor:sym} it persists for all $t$ (it can only change radius, never pitch).
\end{remark}

\bibliographystyle{unsrt}
\bibliography{aipsamp}% Produces the bibliography via BibTeX.

\end{document}